\documentclass[11pt]{article}
\usepackage[T1]{fontenc}
\usepackage[utf8]{inputenc}
\usepackage{lmodern,microtype,amsmath,amssymb,amsthm,mathtools}
\usepackage[margin=.94in]{geometry}
\usepackage{booktabs,array,enumitem,quantikz,graphicx}
\usepackage{algorithm,algpseudocode}
\usepackage{flafter}
\usepackage{authblk}

\usepackage{tocloft,titlesec,needspace,etoolbox}
\apptocmd{\thebibliography}{\raggedright}{}{}

\usepackage{hyperref}
\hypersetup{colorlinks=true,linkcolor=blue,citecolor=blue,urlcolor=blue,
 pdftitle={Quantum Algorithms for Multivariable Polynomial Transformations: From Efficient Synthesis to Quantum Channel Transformations},pdfauthor={Zheyu Shen, Yusen Wu, Xiao Yuan, Xiao-Ming Zhang, Yukun Zhang}}
\numberwithin{equation}{section}
\newtheorem{theorem}{Theorem}[section]
\newtheorem{lemma}[theorem]{Lemma}
\newtheorem{proposition}[theorem]{Proposition}
\newtheorem{corollary}[theorem]{Corollary}
\theoremstyle{definition}\newtheorem{definition}[theorem]{Definition}
\newtheorem{example}[theorem]{Example}
\theoremstyle{remark}\newtheorem{remark}[theorem]{Remark}
\newcommand{\C}{\mathbb C}
\newcommand{\R}{\mathbb R}

\renewcommand{\ket}[1]{\lvert#1\rangle}
\renewcommand{\bra}[1]{\langle#1\rvert}

\newcommand{\tr}{\operatorname{Tr}}
\newcommand{\rank}{\operatorname{rank}}
\newcommand{\diag}{\operatorname{diag}}
\newcommand{\poly}{\operatorname{poly}}
\newcommand{\SA}{\mathrm{SA}}
\newcommand{\dd}{\mathrm d}
\newcommand{\Aco}{\mathsf A}
\newcommand{\Bco}{\mathsf B}
\newcommand{\Cco}{\mathsf C}
\newcommand{\Dco}{\mathsf D}
\allowdisplaybreaks[2]
\title{Quantum Algorithms for Multivariable Polynomial Transformations: From Efficient Synthesis to Quantum Channel Transformations}
\newcommand{\pku}{Center on Frontiers of Computing Studies, School of Computer Science, Peking University, Beijing 100871, China}
\newcommand{\bnu}{School of Artificial Intelligence, Beijing Normal University, Beijing 100875, China}
\newcommand{\oxford}{Mathematical Institute, University of Oxford, Woodstock Road, Oxford OX2 6GG, United Kingdom}
\newcommand{\pkuphy}{School of Physics, Peking University, Beijing 100871, China}
\newcommand{\scnuphysics}{Key Laboratory of Atomic and Subatomic Structure and Quantum Control (Ministry of Education), Guangdong Basic Research Center of Excellence for Structure and Fundamental Interactions of Matter, School of Physics, South China Normal University, Guangzhou 510006, China}
\newcommand{\scnuquantum}{Guangdong Provincial Key Laboratory of Quantum Engineering and Quantum Materials, Guangdong-Hong Kong Joint Laboratory of Quantum Matter, South China Normal University, Guangzhou 510006, China}
\author[1,2]{Zheyu Shen\thanks{Corresponding author: \href{mailto:zheyushen@stu.pku.edu.cn}{\texttt{zheyushen@stu.pku.edu.cn}}.}}
\author[3]{Yusen Wu\thanks{Corresponding author: \href{mailto:yusen.wu@bnu.edu.cn}{\texttt{yusen.wu@bnu.edu.cn}}.}}
\author[1]{Xiao Yuan\thanks{Corresponding author: \href{mailto:xiaoyuan@pku.edu.cn}{\texttt{xiaoyuan@pku.edu.cn}}.}}
\author[1,5,6]{Xiao-Ming Zhang\thanks{Corresponding author: \href{mailto:phyxmz@gmail.com}{\texttt{phyxmz@gmail.com}}.}}
\author[1,4]{Yukun Zhang\thanks{Corresponding author: \href{mailto:yukun.zhang@maths.ox.ac.uk}{\texttt{yukun.zhang@maths.ox.ac.uk}}.}}
\affil[1]{\pku}
\affil[2]{\pkuphy}
\affil[3]{\bnu}
\affil[4]{\oxford}
\affil[5]{\scnuphysics}
\affil[6]{\scnuquantum}
\date{6 October 2026}
\begin{document}
\maketitle

\begin{abstract}
Polynomial transformations are basic primitives in quantum algorithms: quantum signal processing and singular value transformation compile univariate polynomials into circuits with query complexity essentially set by degree. Multivariable transformations of noncommuting matrices, however, lack a comparable synthesis theory. We develop a complete constructive theory under joint block access for all polynomials contractive on the prescribed matrix domain. Given a compact finite-state description of a degree-$D$ polynomial $P$ and $0<\tau<1$, we synthesize $P$ with $(1+\tau)$-optimal normalization using $O(D/\sqrt{\tau})$ queries, and exactly $D$ under row-block access, matching a degree lower bound. The complete circuit realization is classically computable in polynomial time with polylogarithmic dependence on accuracy. The key to our construction is a finite algorithmic Schur--Agler theorem: the coefficient recurrence defines a polynomial-dimensional continuation space supporting a complete semidefinite certificate for contractivity on the prescribed domain. Factoring this certificate yields the contractive realization underlying the query algorithm.

Beyond matrix transformations, our framework lifts multivariable polynomial synthesis to quantum channel transformations. Given coherent Kraus access $K=\{K_a\}_a$, we synthesize any finite jointly contractive family of noncommutative polynomial maps $K\mapsto\{F_b(K)\}_b$ as completely positive operations, allowing coherent interference among Kraus histories. Furthermore, channel-level transformations specified by causal Choi data can be synthesized explicitly as fixed-order quantum combs. Collectively, these results point to a broader program: multivariable approximation as a language for multi-operator quantum algorithms and higher-order quantum information processing.
\end{abstract}
\clearpage
{\hypersetup{linkcolor=black}\tableofcontents}
\setcounter{tocdepth}{2}
\hypersetup{bookmarksdepth=2}
\clearpage

\section{Introduction}
Polynomial transformations provide a way to design quantum operations from the matrices available to an algorithm. Quantum signal processing~\cite{lowchuang2017} and qubitization~\cite{lowchuang2019} use this approach for Hamiltonian simulation, while quantum singular value transformation~\cite{gilyen2019} supports a broad range of matrix arithmetic. Their usefulness rests on a constructive connection: a classical description of the desired polynomial determines the known gates between quantum queries. We develop such a connection for polynomials of several noncommuting matrices, including transformations that retain quantum input and output labels.

For several inputs, the order of multiplication is part of the task. A target is a finite sum of ordered matrix products, with known scalar coefficients, and even a short coefficient recurrence can describe exponentially many such products. Expanding the recurrence into separate words can therefore destroy its computational advantage. There is also a quantum cost to treating the products separately: a term-by-term linear combination~\cite{childswiebe2012} is normalized by the sum of their coefficient magnitudes, together with their encoding normalizations. This can be much larger than the norm of the complete polynomial. We seek a synthesis procedure that retains the compact coefficient description and exploits the norm bound on the joint transformation.

We receive the matrices as designated blocks of one unitary, together with a block that acts as identity on the signal. For instance, placing each matrix at a different diagonal label encodes several independent contractions. A row block $[A_1\ \cdots\ A_k]$ instead combines input labels and obeys $\sum_jA_jA_j^\dagger\preceq I$. The identity branch allows terms of different degrees to share a query sequence. We use the same forward unitary at every call, interleaved with known auxiliary gates; Section~\ref{sec:unitary} specifies the input and shows how coherent selection~\cite{gilyen2019} can supply it, with its access and normalization costs.

For a nonconstant polynomial of degree $D$, let $B$ be its norm over all allowed input tuples and all finite matrix sizes. Given a rational finite-state description of its coefficients and rational $0<\tau,\epsilon<1$, we construct a block encoding with error at most $\epsilon$ and normalization
\[
 B<\beta\le(1+\tau)B.
\]
The circuit uses $O(D/\sqrt\tau)$ queries for a general joint input, and exactly $D$ queries for a row input. Its classical synthesis cost and known-gate count are polynomial in the coefficient description, the oracle-register width, $\tau^{-1/2}$, and $\log(1/\epsilon)$; for a row input the dependence on the margin is polynomial in $\log(1/\tau)$. Thus the normalization approaches the exact-implementation threshold $B$, and the row construction attains the degree lower bound for exact synthesis in the stated forward-query model. Theorem~\ref{cor:scalar-synthesis} gives the scalar guarantee, and Theorem~\ref{thm:matrix-synthesis} extends it to coherent data labels.

Our construction addresses two questions. First, how can the norm of the complete polynomial determine a tractable set of gate parameters? Schur--Agler realization theory~\cite{ballmarxvinnikov2018} relates a uniform norm bound to a positive defect decomposition and a common contractive update. We prove that the residual polynomials left after fixing initial letters of the coefficient recurrence already support a complete finite certificate. Their number is controlled by the recurrence, even when it describes exponentially many words. This makes the norm constraint accessible to classical computation, and the certificate factors determine known auxiliary matrices with certified precision bounds.

Second, how can those matrices be used in a finite circuit whose selected block is the desired polynomial? We construct a query sequence that coherently combines their ordered coefficient products and control the change in normalization introduced by this combination. A row input permits a stronger factorization: the certificate completes the target to a polynomial column, from which each step removes one degree and can be reversed by one row query. The two constructions therefore give different query bounds from the same finite certificate theory. Section~\ref{sec:general} develops their intermediate matrices and full unitary circuits.

Retaining an output label gives a direct application to quantum operations. For a specified coherent Kraus implementation, our algorithm can produce several polynomials $F_b(K)=\sum_w c_{b,w}K_w$ on distinct labels. After a successful output test, tracing those labels realizes
\begin{equation}\label{eq:intro-polynomial-operation}
 \mathcal M_\beta(\rho)=\frac1{\beta^2}\sum_bF_b(K)\rho F_b(K)^\dagger.
\end{equation}
Here a single certificate bounds the entire output column, so its normalization makes the operation trace-nonincreasing. Keeping a label instead allows different polynomials to act coherently on different data values, with a query count governed by the largest degree. Section~\ref{sec:channels} develops these two applications. For comparison with channel-only access, Appendix~\ref{app:comb} gives a separate gate compiler for controllers specified by explicit causal Choi matrices.

\subsection{Relation to earlier work}
For a single signal, classical synthesis separates naturally into completing a polynomial to a unitary-valued function and recovering the gate parameters of that function. Haah's factorization~\cite{haah2019} and the methods of Chao et al.~\cite{chao2020} and Dong et al.~\cite{dong2021} give constructive phase sequences. For the completion step, Berntson and S\"underhauf~\cite{berntson2025} derive a canonical complementary polynomial through a scalar Riemann--Hilbert problem and give contour/FFT algorithms with error bounds. For phase recovery, Alexis, Mnatsakanyan, and Thiele~\cite{alexis2024nlfa} connect QSP with nonlinear Fourier analysis; the Riemann--Hilbert--Weiss algorithm of Alexis et al.~\cite{alexis2026rhw} extends the analytic representation to Szeg\H{o} functions and proves stability of phase recovery with a strict norm margin. Ni and Ying~\cite{niying2024} exploit structured linear algebra in the Half-Cholesky method, and Ni, Sarkar, Ying, and Lin~\cite{ni2025inverse} give an inverse nonlinear Fourier algorithm with nearly linear arithmetic complexity for completed Laurent data, with stability guarantees when the designated analytic factor is zero-free on the closed unit disk. These developments distinguish the cost of completion from that of phase inversion, and both from polynomial approximation; Lin's survey~\cite{lin2026survey} reviews their mathematical and numerical connections. Our corresponding classical problem is to compute a norm certificate on a specified noncommutative matrix domain from a compact coefficient recurrence.

Generalized QSP enlarges the admissible single-unitary polynomial class~\cite{motlagh2024}, and Laneve~\cite{laneve2025nlft} identifies its relation to the nonlinear Fourier transform. Increasing the output dimension is another extension: Lu, Liu, and Lin~\cite{lu2026} give matrix-valued polynomial synthesis on $U(N)$ and also a bivariate construction, while Alexis et al.~\cite{alexis2026matrix} develop an $\mathrm{SU}(2n)$-valued nonlinear Fourier transform. For several scalar signals, Rossi and Chuang's architecture~\cite{rossi2022} has factorization constraints investigated through counterexamples~\cite{mori2024}, constructive decision algorithms~\cite{ito2024}, and representability results~\cite{lanevewolf2025}. N\'emeth et al.~\cite{nemeth2023} also study homogeneous noncommuting variants. Thus adding variables and adding coherent output labels have several distinct precedents; in our setting, both the admissible targets and their normalization are determined by the supplied joint block encoding.

For general square matrices, the encoding also determines which function is being transformed. Low and Su's eigenvalue processing~\cite{lowsu2026} combines polynomial approximation with spectral geometry, whereas regular block encodings preserve matrix powers for QSP-based evaluation~\cite{gutierrez2026}. Interpolation-based synthesis can construct a diagonal encoding of sampled function values~\cite{ko2026}. Our function applications instead supply compact polynomial approximations that are uniform on the declared joint domain. Appendix~\ref{sec:applications} gives such an approximation, together with its loading assumptions and success bound.

Another approach uses positive matrix data to describe quantum protocols. Semidefinite query characterizations~\cite{barnum2003,barnum2007} and completely bounded polynomial formulations~\cite{arunachalam2019} connect feasibility to quantum implementation; the associated tensor norm has an SDP formulation~\cite{gribling2019}, and finite Christensen--Sinclair factors can be recovered through semidefinite duality~\cite{escudero2024}. In state conversion, adversary data guide the algorithm~\cite{leemittal2011}; Laneve~\cite{laneve2025} uses feasible adversary data to construct QSP protocols, and transducers provide an accounting of composition and time~\cite{belovsjefferyyolcu2024}. We begin instead with a rational recurrence for the target coefficients. Our residual-space theorem supplies a complete polynomial-size certificate for its joint norm, and quantitative strict feasibility makes this certificate and the resulting gates computable. The exact fixed-query characterization in Appendix~\ref{app:fixed-query} uses full histories, whose number can instead be exponential in the query count.

The mathematical form of our certificate comes from Agler's defect decompositions~\cite{agler1990} and noncommutative realization theory~\cite{ballgroenewaldmalakorn2005}. Bounded-real theory relates finite structured realizations to norm inequalities~\cite{ballgroenewaldmalakorn2006,ballgroenewaldterhorst2016}, and finite rational defect decompositions on the noncommutative polydisc are available~\cite{bhattacharyya2026}. We state the general free Schur--Agler representation principle~\cite{ballmarxvinnikov2018} in Section~\ref{sec:sa-foundation}. The target-dependent finite certificate is then proved independently, using hereditary positivity methods~\cite{augat2018,heltonklepnelson2015}; its Gram factors give the common contraction through Douglas factorization~\cite{douglas1966}. The additional requirements for synthesis are completeness in the chosen residual space, bounds in the original description length, and a circuit for the resulting numerical matrices.

For quantum operations, Choi matrices~\cite{choi1975,jamiolkowski1972}, Kraus operators~\cite{watrous2018}, and Stinespring dilations~\cite{stinespring1955} describe different aspects of the same completely positive map. Interference between Kraus histories uses the extra information retained by a specified coherent implementation~\cite{abbott2020,chiribellakristjansson2019}. Channel-only higher-order maps are described by supermaps~\cite{chiribella2008} and quantum strategies~\cite{gutoskiwatrous2007}, with sequential realization under a fixed causal order given by comb theory~\cite{chiribella2009,bisio2010}. Our joint polynomial construction acts on the specified coherent implementation. We treat compilation of a controller from explicit causal Choi data separately, using that existing sequential realization theory.

\section{The synthesis problem and its mathematical foundation}\label{sec:unitary}
To determine the gates of the transformation, we first specify the joint oracle and the classical description of the target, since they determine both the relevant norm and the cost of synthesis. The Schur--Agler representation principle then identifies a common contractive update, and a finite polynomial certificate will provide its numerical entries.

\subsection{Inputs, oracle access, and the synthesis guarantee}
The signal register $\mathsf S$ has $n_s$ qubits and dimension $n=2^{n_s}$. The matrices $A_1,\ldots,A_k$ act on this same register. For a word $w=j_1\cdots j_d$ over $\{1,\ldots,k\}$, we write $A_w=A_{j_1}\cdots A_{j_d}$ and $|w|=d$, with $A_\emptyset=I_n$. Our principal target is the scalar-coefficient polynomial
\begin{equation}\label{eq:target}
 p(A)=\sum_w c_wA_w,\qquad D=\max\{|w|:c_w\ne0\}.
\end{equation}
The sum is finite, and $D$ denotes the true degree of a nonzero target. We write $x_w$ for the corresponding formal word. The rightmost factor acts first on the signal, so implementing $A_1A_2$ requires the multiplication by $A_2$ before the multiplication by $A_1$. The coefficients are known classical data; the matrices $A_j$ are accessed through the unitary specified below.

For one matrix, a block encoding identifies an ancillary input and output of a unitary whose selected block is that matrix, possibly divided by a known scale~\cite{gilyen2019}. We use the same idea for several matrices by placing them at different ancillary labels. For example, with two separately bounded matrices the supplied unitary has the form
\begin{equation}\label{eq:basic-joint-oracle}
 U_A=\begin{pmatrix}
 I_n&0&0&0\\
 0&A_1&0&*\\
 0&0&A_2&*\\
 0&*&*&*
 \end{pmatrix},\qquad \|A_1\|,\|A_2\|\le1.
\end{equation}
The first three block rows and columns correspond to ancillary labels $0,1,2$; the last group contains all remaining labels. If the oracle register $\mathsf O$ has $q_o$ qubits and $d_o=2^{q_o}\ge5$ labels, these groups have lengths $n,n,n,(d_o-3)n$. Thus Eq.~\eqref{eq:basic-joint-oracle} denotes a complete square unitary. The stars are arbitrary compatible blocks that complete its columns and rows. On label zero the oracle acts as identity on the signal. This coherent identity branch will allow components requiring different numbers of matrix multiplications to undergo the same number of oracle calls.

One query is one forward application of the entire $U_A$. Let $q_{\rm anc}$ be the total number of non-signal qubits, including the oracle register. Between queries, we apply known unitaries $W_t$ of order $2^{q_{\rm anc}}$ to these qubits, leaving the signal unchanged. For $Q\ge1$, writing $\widehat U_A$ for the supplied unitary extended by identity on the added registers, the resulting circuit is
\begin{equation}\label{eq:target-query-sequence}
 \mathcal U_A=(W_Q\otimes I_n)\widehat U_A
 (W_{Q-1}\otimes I_n)\cdots\widehat U_A(W_0\otimes I_n).
\end{equation}
The input on $\mathsf S$ is arbitrary. The oracle register and the added auxiliary registers start at zero. For a positive normalization $\beta$, the ideal circuit is required to satisfy
\begin{equation}\label{eq:selected-output-task}
 (\bra{0^{q_{\rm anc}}}\otimes I_n)\mathcal U_A
 (\ket{0^{q_{\rm anc}}}\otimes I_n)=p(A)/\beta.
\end{equation}
We compute both $\beta$ and the known gates. Their finite-gate approximations give the prescribed block error $\epsilon$. The same gate description must work for every compatible choice of the starred blocks, so its parameters depend on the coefficients and the declared block positions. The guarantee also holds when the signal is entangled with an untouched reference.

\paragraph{Obtaining the input and choosing its layout.}
The diagonal arrangement has a familiar implementation through coherent selection. Suppose $U_j$ is a block encoding of a physical matrix $M_j/\alpha_j$, with known $\alpha_j>0$. A selector $U_{\rm sel}$ applies $U_j$ on index $j$ and identity on index zero. Keeping this index and the designated encoding coordinates gives Eq.~\eqref{eq:basic-joint-oracle}, or its $k$-variable version, with loaded matrices $A_j=M_j/\alpha_j$. A supplied selector therefore implements one joint query. If instead the individually controlled $U_j$ are supplied, applying each under the corresponding index condition constructs the selector using one controlled call per input. These are distinct input resources: after the joint unitary has been obtained, Eq.~\eqref{eq:target-query-sequence} calls it unconditionally, while constructing it from arbitrary unknown unitaries requires the additional control access~\cite{araujo2014}.

Other layouts are useful because they combine auxiliary labels differently. The active part of a diagonal encoding retains the variable label, a row combines input labels into one output, and a column distributes one input over several outputs. Their constraints are
\begin{equation}\label{eq:two-input-domains}
 \begin{aligned}
 \Delta_{\rm diag}(A)&=\operatorname{diag}(A_1,\ldots,A_k),
 &\|\Delta_{\rm diag}(A)\|\le1&\iff\max_j\|A_j\|\le1,\\
 \Delta_{\rm row}(A)&=[A_1\ \cdots\ A_k],
 &\|\Delta_{\rm row}(A)\|\le1&\iff\sum_jA_jA_j^\dagger\preceq I_n,\\
 \Delta_{\rm col}(A)&=\begin{pmatrix}A_1\\\vdots\\A_k\end{pmatrix},
 &\|\Delta_{\rm col}(A)\|\le1&\iff\sum_jA_j^\dagger A_j\preceq I_n.
 \end{aligned}
\end{equation}
For a row input, the designated block, including the identity branch, is
\[
 L_{\rm row}(A)=\begin{pmatrix}I_n&0&\cdots&0\\0&A_1&\cdots&A_k\end{pmatrix}.
\]
Its common active output allows the degree-reducing construction of Section~\ref{sec:row} to restore one degree with each query. A column is the natural input for a coherent Kraus implementation, where the output labels record the Kraus index. These applications motivate retaining rectangular blocks in the general theorem.

A known rotation converts a selector into a weighted row. We choose a unit column $g=(g_1,\ldots,g_k)^{\mathsf T}$ with nonzero, efficiently computable entries, and let $P_g$ fix index zero and send $\sum_jg_j\ket j$ to $\ket1$. Then
\begin{equation}\label{eq:row-access-construction}
 U_{\rm row}=(P_g\otimes I)U_{\rm sel}
\end{equation}
has active row $[\overline g_1M_1/\alpha_1\ \cdots\ \overline g_kM_k/\alpha_k]$. The identity acts on the common encoding register and signal. This construction uses one selector and a known rotation, as in linear-combination-of-unitaries preparation and selection~\cite{childswiebe2012}. Equal weights give a factor $1/\sqrt{k}$, which is needed because separate contraction bounds do not imply the row bound in Eq.~\eqref{eq:two-input-domains}.

The target must be expressed in the matrices actually loaded. With $Z_j=\overline g_jM_j/\alpha_j$, replacing each word coefficient by $c_w\prod_{t=1}^{|w|}(\alpha_{j_t}/\overline g_{j_t})$ gives a polynomial $\widetilde p$ with $\widetilde p(Z)=p(M)$. In particular, when $\alpha_1=\alpha_2=1$, an equally loaded two-variable row contributes a factor $1/2$ to every quadratic product. The $D$-query row result concerns this loaded domain and the appropriately rewritten target; its normalization incorporates the loading scales. Appendix~\ref{app:loaders} specifies all registers, gives rational choices of the weights, and counts the constituent calls. It also constructs known linear combinations of blocks using the block-encoding combination rule~\cite{gilyen2019}.

\paragraph{The general joint oracle.}
We now give one set of assumptions covering these layouts and known linear mixtures of their entries. The active block is
\begin{equation}\label{eq:pencil}
 \Delta(A)=\left[\sum_{j=1}^k f_{\mu\nu,j}A_j\right]_{{\mu=1,\ldots,a};\,{\nu=1,\ldots,b}},
 \qquad \|\Delta(A)\|\le1,
\end{equation}
where the known $a\times b$ matrices $F_j=[f_{\mu\nu,j}]$ are linearly independent and have rational real and imaginary parts. Hence $\Delta(A)$ has $an$ rows and $bn$ columns. Independence means that these blocks distinguish the variables: a left inverse of the matrix of coefficients expresses every $A_j$ as a known linear combination of the blocks. For example, $F_1=F_2$ would reveal only $A_1+A_2$, which cannot determine a target depending separately on $A_1$. The domain in Eq.~\eqref{eq:pencil} bounds each variable, but its individual bound need not be one; all norm statements below refer to the joint block.

The supplied oracle acts on $\mathsf O,\mathsf S$, with $d_o=2^{q_o}$ ancillary labels. In known computational-basis input and output orders, it has the complete matrix
\begin{equation}\label{eq:oracle-full}
 U_A=\begin{pmatrix}
 I_{in}&0&0\\
 0&\Delta(A)&B_A\\
 0&C_A&D_A
 \end{pmatrix},\qquad U_A^\dagger U_A=U_AU_A^\dagger=I_{d_on}.
\end{equation}
Here $i\ge1$, the block row lengths are $in,an,(d_o-i-a)n$, and the block column lengths are $in,bn,(d_o-i-b)n$. The designated input labels are $0,\ldots,i+b-1$ and the output labels are $0,\ldots,i+a-1$, selecting
\begin{equation}\label{eq:visible-block}
 L(A)=\begin{pmatrix}I_{in}&0_{in\times bn}\\0_{an\times in}&\Delta(A)\end{pmatrix}.
\end{equation}
We assume $d_o\ge i+a+b$, so every allowed active block admits a unitary completion; Appendix~\ref{app:unitary-details} proves this dimension condition. The matrices $B_A,C_A,D_A$ are otherwise unrestricted. The identity block and its phase relative to the active entries are part of the supplied resource. Its columns already have norm one, which forces the adjoining zeros and makes this branch carry the signal unchanged. Each query in Eq.~\eqref{eq:target-query-sequence} uses this same full unitary, including on components outside the designated block.

\paragraph{Classical coefficients and normalization.}
A polynomial such as $(x_1+x_2)^D$ contains $2^D$ ordered products, although a recurrence with $D+1$ states generates all its coefficients. We therefore measure synthesis cost against a finite linear representation of the supplied coefficients:
\begin{equation}\label{eq:finite-state}
 c_w=\ell T_wr_0,\qquad
 T_w=T_{j_1}\cdots T_{j_d},\qquad T_\emptyset=I_s.
\end{equation}
The row $\ell$, matrices $T_j$, and column $r_0$ have sizes $1\times s$, $s\times s$, and $s\times1$, with rational real and imaginary parts encoded in binary. Such representations are standard for weighted automata~\cite{carlylepaz1971} and noncommutative series~\cite{berstelreutenauer2011}; they allow different words to share a coefficient calculation. For a loaded row, multiplying $T_j$ by $\alpha_j/\overline g_j$ performs the letter-dependent rescaling above. The rational-input theorem requires these multipliers to be rational, and Appendix~\ref{app:loaders} gives suitable weights when the original scales are rational.

Sparse word lists and explicit branching programs also have polynomial-cost conversions to Eq.~\eqref{eq:finite-state}. The input checks use rational weighted-automaton reduction~\cite{kiefer2020}, followed by the polynomiality and degree tests in Appendix~\ref{app:coefficient-inputs}. For a nonzero minimally represented polynomial, the resulting degree satisfies $D\le s-1$. The synthesis algorithm reads these coefficient matrices and the oracle-layout data, and returns the known gates to be interleaved with the supplied quantum queries.

The uniform norm associated with the layout is the free Schur--Agler norm~\cite{ballmarxvinnikov2018},
\begin{equation}\label{eq:SA-norm}
 B=\|p\|_{\SA,\Delta}
   =\sup_{n\ge1}\ \sup_{\|\Delta(A)\|\le1}\|p(A)\|.
\end{equation}
It measures the target over the entire declared domain, rather than at the particular matrices in one run. We use all finite sizes because the coefficient-based circuit must work on every signal register; zero embedding shows that power-of-two sizes suffice, and radial continuity gives the same norm on the open domain. Contractivity of a unitary block forces $\beta\ge B$ for an exact implementation of $p/\beta$. With uniform selected-block error $\epsilon$, the same argument gives $\beta\ge B/(1+\epsilon)$. Our algorithm approaches the exact normalization threshold while allocating the known-gate error separately.

\begin{theorem}[Scalar polynomial transformation]\label{cor:scalar-synthesis}
Let $p$ have a rational finite-state representation of degree $D\ge1$, and let $0<\tau,\epsilon<1$ be rational. Under the joint-oracle assumptions above, a classical algorithm computes a rational $\beta$ and a quantum circuit $\mathcal U_A$ on the signal and $q_{\rm anc}$ non-signal qubits such that
\begin{equation}\label{eq:main-output}
 \left\|(\bra{0^{q_{\rm anc}}}\otimes I_n)\mathcal U_A
 (\ket{0^{q_{\rm anc}}}\otimes I_n)-p(A)/\beta\right\|\le\epsilon,
 \qquad B<\beta\le(1+\tau)B.
\end{equation}
The circuit uses $O(D/\sqrt\tau)$ queries, or $D$ for a row input. Classical computation and known-gate count are polynomial in the rational input length, $q_o$, $\tau^{-1/2}$, and $\log(1/\epsilon)$; the row case depends polynomially on $\log(1/\tau)$. The guarantees hold uniformly over signal sizes, admissible tuples, compatible oracle completions, and states entangled with an untouched reference. With exact known gates, the selected block equals $p(A)/\beta$.
\end{theorem}
The parameter $\tau$ controls the excess normalization, whereas $\epsilon$ controls approximation of the known gates. The ideal success probability on a unit input $\psi$ is $\|p(A)\psi\|^2/\beta^2$, so an application that conditions on success also needs a lower bound on this quantity. Section~\ref{sec:coherent-use} gives the coherent-use and error guarantees. Nonzero constants use a known scalar block encoding with zero queries, and the zero polynomial has a zero-query realization at any positive normalization.

\begin{example}[A noncommutative target used below]\label{ex:commutator-thread}
On the diagonal domain, we consider
\begin{equation}\label{eq:running-target}
 p(x)=\frac14+\frac{x_1x_2-x_2x_1}{12}.
\end{equation}
Scalar substitutions give $1/4$, whereas the norm-one matrices
\[
 A_1=\begin{pmatrix}0&1\\1&0\end{pmatrix},\qquad
 A_2=\begin{pmatrix}0&1\\-1&0\end{pmatrix}
\]
give commutator $\operatorname{diag}(-2,2)$. This pair attains the triangle-inequality bound $B=5/12$. We will follow the same target through its residual coordinates and the compensation of its quadratic coefficient, illustrating how the synthesis preserves the order of the two matrix products.
\end{example}

\subsection{Schur--Agler representations as query parameters}\label{sec:coefficient-mechanism}\label{sec:sa-foundation}
The coefficient recurrence specifies the ordered products we need, but its individual transition norms may give a poor normalization. We instead seek one contraction that accounts for the norm of the whole target. Agler's defect decompositions~\cite{agler1990} and the free Schur--Agler realization theorem~\cite{ballmarxvinnikov2018} provide the mathematical reason for making this choice.

We state the representation principle for a scalar free polynomial $q(x)=\sum_wq_wx_w$, which will later be the normalized target supplied to the synthesis routine. For an auxiliary Hilbert space $\mathcal H$, repeating the input block gives the $a\times b$ operator matrix
\[
 \Delta_{\mathcal H}(A)=
 \left[\sum_jf_{\mu\nu,j}(I_{\mathcal H}\otimes A_j)\right]_{\mu,\nu}.
\]
Here and below, a subscript $n$ on a fixed block denotes its tensor product with $I_n$; for a finite numerical matrix $C$, this is $C_n=[C_{rs}I_n]_{r,s}$.

\begin{lemma}[Free Schur--Agler realization]\label{lem:free-sa}
A scalar free polynomial $q$ satisfies $\|q\|_{\SA,\Delta}\le1$ if and only if there are a Hilbert space $\mathcal H$ and a contraction
\[
 V^\circ=\begin{pmatrix}\Aco^\circ&\Bco^\circ\\\Cco^\circ&\Dco^\circ\end{pmatrix}
\]
with block input $[z;c]$, where $z$ is a column of $a$ vectors in $\mathcal H$ and $c\in\C$, and block output consisting of $b$ such vectors and a scalar, such that
\begin{equation}\label{eq:analytic-sa-transfer}
 q(A)=\Dco_n^\circ+\Cco_n^\circ\Delta_{\mathcal H}(A)
       \bigl(I-\Aco_n^\circ\Delta_{\mathcal H}(A)\bigr)^{-1}\Bco_n^\circ
       \qquad(\|\Delta(A)\|<1).
\end{equation}
\end{lemma}
The four blocks describe injection, an update between successive uses of the input block, collection, and a constant contribution. The lemma guarantees that one contraction can perform these tasks consistently with the target norm. Its auxiliary space may be infinite dimensional, so we still need finite data from which to compute gates. We will obtain those data as a polynomial defect certificate.\footnote{Lemma~\ref{lem:free-sa} is the scalar specialization of the free Schur--Agler realization theorem of Ball, Marx, and Vinnikov~\cite{ballmarxvinnikov2018}. A representation of the finite input matrix algebra is an identity amplification, and applying $(I-XY)^{-1}X=X(I-YX)^{-1}$ to the transfer formula gives the displayed order. Contractivity makes the inverse exist when $\|\Delta(A)\|<1$; polynomial continuity extends the norm bound to the boundary.}

For a positive integer $m$, we use $m$ internal coordinates and order the factors as block label, internal index, and signal. The resulting matrices are
\begin{equation}\label{eq:F-memory}
 F_j^{(m)}=[f_{\mu\nu,j}I_m]_{\mu,\nu},\qquad
 F_j^{(m)}:\ am\times bm,
\end{equation}
\begin{equation}\label{eq:delta-amplification}
 \Delta_m(A)=\sum_jF_j^{(m)}\otimes A_j,
 \qquad \Delta_m(A):\ amn\times bmn.
\end{equation}
It acts as $\Delta(A)$ on each internal value, so its norm is $\|\Delta(A)\|$. The update we will compute is the numerical matrix
\begin{equation}\label{eq:V}
 V=\begin{pmatrix}\Aco&\Bco\\\Cco&\Dco\end{pmatrix},\qquad
 V:\ (bm+1)\times(am+1),\qquad \|V\|\le1.
\end{equation}
The four block sizes are $bm\times am$, $bm\times1$, $1\times am$, and $1\times1$. The column $\Bco$ injects the signal into the active input coordinates, $\Aco$ returns a queried output to those coordinates, $\Cco$ collects an output, and $\Dco$ contributes the constant term. Their joint contraction bound says that, for an active column $z$ of length $amn$ and a signal column $\psi$,
\begin{equation}\label{eq:joint-update-norm}
 \|\Aco_nz+\Bco_n\psi\|^2+\|\Cco_nz+\Dco_n\psi\|^2
 \le\|z\|^2+\|\psi\|^2.
\end{equation}
This is the norm inequality we will preserve when completing the update to a unitary gate.

Each occurrence of $\Delta_m(A)$ contributes one matrix letter. If $q_d(A)=\sum_{|w|=d}q_wA_w$ is the degree-$d$ part, the first required identities are
\begin{equation}\label{eq:transfer-first-terms}
 \begin{aligned}
 q_0(A)&=\Dco_n,\\
 q_1(A)&=\Cco_n\Delta_m(A)\Bco_n,\\
 q_2(A)&=\Cco_n\Delta_m(A)\Aco_n\Delta_m(A)\Bco_n.
 \end{aligned}
\end{equation}
Expanding the products shows how the update generates each ordered coefficient:
\begin{equation}\label{eq:coefficient-factorization}
 q_\emptyset=\Dco,\qquad
 q_{j_1\cdots j_d}
 =\Cco F_{j_1}^{(m)}\Aco F_{j_2}^{(m)}\cdots
       \Aco F_{j_d}^{(m)}\Bco\quad(d\ge1).
\end{equation}
A degree-one word has no $\Aco$ factor. Summing this sequence gives the finite-dimensional transfer formula~\cite{ballgroenewaldmalakorn2005}.

\begin{definition}[Finite contractive transfer realization]\label{def:transfer-realization}
The contraction in Eq.~\eqref{eq:V} realizes $q$ if, for every tuple with $\|\Delta(A)\|<1$,
\begin{equation}\label{eq:sa-transfer}
 q(A)=\Dco_n+\Cco_n\Delta_m(A)
          (I_{bmn}-\Aco_n\Delta_m(A))^{-1}\Bco_n.
\end{equation}
We call $m$ the multiplicity of the realization.
\end{definition}
Its convergent geometric expansion on the open domain gives precisely Eq.~\eqref{eq:coefficient-factorization}. Thus the finite matrix must reproduce the entire coefficient sequence, including zero coefficients above the target degree. The circuit construction in Section~\ref{sec:general} will combine the required degrees in one output.

\paragraph{The finite data to be computed.}
For a nonconstant target of degree $D$, we seek polynomial columns $G,H$ with $t,bm$ entries and degrees at most $D,D-1$, respectively, satisfying
\begin{equation}\label{eq:polynomial-agler}
 1-q(x)^*q(x)=G(x)^*G(x)
       +H(x)^*\bigl(I_{bm}-\Delta_m(x)^*\Delta_m(x)\bigr)H(x).
\end{equation}
Here $x_1,\ldots,x_k$ are formal noncommuting letters, $\Delta_m(x)=\sum_jF_j^{(m)}x_j$, and the formal adjoint conjugates the coefficients and reverses the adjointed word. The evaluations of $G,H$ have sizes $tn\times n$ and $bmn\times n$. We call the identity a polynomial defect certificate. Applied to a signal vector $\psi$, its second term is the nonnegative difference $\|H(A)\psi\|^2-\|\Delta_m(A)H(A)\psi\|^2$, and the first term supplies the remaining loss. It therefore certifies the target's norm and, being an identity of polynomial coefficients, gives numerical data for a common update.

\begin{lemma}[From a finite certificate to query parameters]\label{lem:certificate-to-realization}
A polynomial defect certificate for $q$ determines a finite contraction $V$ of the form in Eq.~\eqref{eq:V} that realizes $q$. In a common independent polynomial basis, let $X,Y$ be the numerical coefficient matrices of $[\Delta_mH;1]$ and $[H;q]$. Writing $X^+$ for the Moore--Penrose inverse, the choice $V=YX^+$ satisfies $VX=Y$ and all identities in Eq.~\eqref{eq:coefficient-factorization}.
\end{lemma}
\begin{proof}
Rearranging the certificate gives
\begin{equation}\label{eq:certificate-to-contraction}
 \begin{pmatrix}\Delta_m H\\1\end{pmatrix}^{\!*}
 \begin{pmatrix}\Delta_m H\\1\end{pmatrix}
 -\begin{pmatrix}H\\q\end{pmatrix}^{\!*}
 \begin{pmatrix}H\\q\end{pmatrix}=G^*G.
\end{equation}
We express the three polynomial columns $[\Delta_mH;1]$, $[H;q]$, and $G$ in a common independent scalar-polynomial basis. Their numerical coefficient matrices $X,Y,Z$ have $am+1,bm+1,t$ rows and the same number of columns. The unique split between starred and unstarred letters makes comparison of coefficients injective in this basis. Eq.~\eqref{eq:certificate-to-contraction} therefore gives $X^\dagger X-Y^\dagger Y=Z^\dagger Z$. It follows that the rule $Xc\mapsto Yc$ is well defined and contractive. Extending it by zero on the orthogonal complement gives the Douglas factor~\cite{douglas1966} $V=YX^+$ with $VX=Y$.

The block rows of this last equality give
\begin{equation}\label{eq:sa-feedback}
 H(A)=\Aco_n\Delta_m(A)H(A)+\Bco_n,\qquad
 q(A)=\Cco_n\Delta_m(A)H(A)+\Dco_n.
\end{equation}
Solving the first equation on the open domain gives Eq.~\eqref{eq:sa-transfer}. The same equality can be read coefficient by coefficient. Writing $H_w$ for the numerical column multiplying word $w$ in $H$, we obtain
\begin{equation}\label{eq:h-recurrence}
 \begin{aligned}
 H_\emptyset&=\Bco,&q_\emptyset&=\Dco,\\
 H_{jw}&=\Aco F_j^{(m)}H_w,&q_{jw}&=\Cco F_j^{(m)}H_w.
 \end{aligned}
\end{equation}
Iteration gives Eq.~\eqref{eq:coefficient-factorization} in the prescribed word order. Since $H$ is a polynomial of degree at most $D-1$, these identities also verify that all target coefficients above degree $D$ vanish. Thus the finite certificate determines both the numerical contraction and the complete coefficient sequence it generates.
\end{proof}

The converse norm calculation explains the two terms in the certificate. For any finite contractive realization and an interior tuple, we form its response and queried input column,
\[
 H_V(A)=(I_{bmn}-\Aco_n\Delta_m(A))^{-1}\Bco_n,\qquad
 X_V(A)=\begin{pmatrix}\Delta_m(A)H_V(A)\\I_n\end{pmatrix}.
\]
The transfer equation gives $V_nX_V(A)=[H_V(A);q(A)]$. Subtracting the corresponding Gram matrices proves
\begin{equation}\label{eq:transfer-energy}
 \begin{aligned}
 I_n-q(A)^\dagger q(A)
 &=H_V(A)^\dagger\bigl(I_{bmn}-\Delta_m(A)^\dagger\Delta_m(A)\bigr)H_V(A)\\
 &\quad+X_V(A)^\dagger\bigl((I_{am+1}-V^\dagger V)\otimes I_n\bigr)X_V(A).
 \end{aligned}
\end{equation}
Both terms are positive by the contraction bounds on $\Delta(A)$ and $V$. This proves the norm bound on the open domain, and radial continuity of the polynomial extends it to the closed domain.

We have identified the finite data sufficient for synthesis. The remaining classical question is whether the target's compact recurrence supplies a complete search space for $G,H$, and whether a certificate in that space can be found with controlled precision. Section~\ref{sec:classical} proves both assertions directly. This is the step that connects the representation principle to an algorithm on the supplied coefficient description.

\section{Classical synthesis of the query parameters}\label{sec:classical}
The Schur--Agler criterion identifies the form of the parameters; we now compute them from the target's coefficient recurrence. The unknowns are the coefficients of the auxiliary polynomial columns in the defect identity. Our main result in this section is that the target residuals and their one-letter extensions provide complete coordinates for the defect identity, with dimension bounded by the supplied state count. In these coordinates the certificate is an affine equation in two positive semidefinite matrices.

After constructing the coordinates, we prove that the semidefinite equation captures the norm bound exactly and show how a feasible pair determines the numerical contraction $V$. Finding that pair is the remaining computational task: a strict norm margin will control both the search and the precision of the matrix entries. Throughout this section, $p$ denotes the polynomial supplied to the classical routine; the general query algorithm will also call this routine on a degree-reweighted target.

\subsection{Residual coordinates for the certificate}\label{sec:residual}
For a polynomial $f(x)=\sum_w f_wx_w$, its left residual at letter $j$ is
\[
 \partial_jf(x)=\sum_w f_{jw}x_w,
 \qquad f(x)=f_\emptyset+\sum_jx_j\partial_jf(x).
\]
A residual fixes the first letter of a coefficient word and retains its remaining letters. In the finite-state representation, the row $\lambda$ describing coefficients $\lambda T_wr_0$ becomes $\lambda T_j$ after this deletion. All residuals therefore share the supplied state coordinates, as in the theory of noncommutative rational series~\cite{berstelreutenauer2011}. We take $f_1,\ldots,f_r$ to be a basis of the span of the constant $1$ and all nonempty residuals of $p$. This space is closed under every $\partial_j$, and its members have degree at most $D-1$. Closure will allow us to reconstruct a polynomial by successively restoring its initial letters, which is the step needed in the completeness proof below.

For the running polynomial in Eq.~\eqref{eq:running-target}, the first residuals are $\partial_1p=x_2/12$ and $\partial_2p=-x_1/12$. Deleting one further letter leaves a constant, so we may choose
\[
 f=(1,x_1,x_2)^{\mathsf T},\qquad
 p=\frac14+x_1\frac{x_2}{12}+x_2\frac{-x_1}{12}.
\]
We will seek the entries of the auxiliary polynomial column $H$ in this three-dimensional span. Applying the query to such an entry also produces $x_jf_\ell$, so the defect identity needs the common coordinate list
\[
 g=(1,x_1,x_1^2,x_1x_2,x_2,x_2x_1,x_2^2)^{\mathsf T}.
\]
The terms $x_1^2$ and $x_2^2$ arise from applying the entire input block to the residual space. Their coordinates belong to the certificate even though their target coefficients are zero. Rescaling its quadratic part preserves this residual span, so the later degree compensation uses the same coordinates.

The compression can retain linear relations among exponentially many words. For $p_D(x)=(x_1+x_2)^D$, with $D\ge2$, the representation in Eq.~\eqref{eq:finite-state} uses $T_1=T_2=\sum_{t=0}^{D-1}e_te_{t+1}^{\mathsf T}$, initial row $e_0^{\mathsf T}$, and final column $e_D$, where $e_0,\ldots,e_D$ are the coordinate columns of length $D+1$. It has $2^D$ nonzero coefficients, but we can choose
\[
 f_\ell(x)=(x_1+x_2)^{\ell-1}\quad(1\le\ell\le D),
 \qquad r=D.
\]
At either letter, taking a residual lowers each positive exponent by one and annihilates the constant; the different degrees make the displayed polynomials independent. A monomial basis containing the same residuals would have to include every word through degree $D-1$. The constant and the two sets of one-letter successors give $1+2D$ independent coordinates. The certificate below uses these whole residual polynomials; its completeness is what makes this compression useful for synthesis. This example concerns the certificate representation, while the power itself also admits a direct specialized construction.

In general, we record the constant coefficients, the residual maps, and the target residuals in rational arrays
\begin{equation}\label{eq:residual-arrays}
 a_0=(a_{0,1},\ldots,a_{0,r})^{\mathsf T},\qquad
 A_j^{\rm res}=[a_{j,\ell t}]_{\ell,t=1}^r,
 \qquad p_j=(p_{j,1},\ldots,p_{j,r}).
\end{equation}
Their meaning is fixed by
\begin{equation}\label{eq:residual-identities}
 f_\ell(x)=a_{0,\ell}+\sum_{j=1}^k\sum_{t=1}^r
                         a_{j,\ell t}x_jf_t(x),
 \qquad
 p(x)=c_\emptyset+\sum_{j=1}^k\sum_{t=1}^r p_{j,t}x_jf_t(x).
\end{equation}
We then use the constant and the one-letter successors as the common basis:
\begin{equation}\label{eq:g-list}
 g_0(x)=1,\qquad g_{(j,\ell)}(x)=x_jf_\ell(x)
       \quad(1\le j\le k,\ 1\le\ell\le r),\qquad h=1+kr.
\end{equation}
Reachability and reduction for weighted automata compute these coordinates without expanding the word list~\cite[Propositions~2.1 and~3.2, Theorem~3.6]{kiefer2020}. The following construction also bounds their rational encoding lengths, as required for the semidefinite calculation.

\begin{lemma}[Computing the residual coordinates]\label{lem:residual}
For a nonconstant rational $s$-state polynomial of degree $D$, the coordinates in Eqs.~\eqref{eq:residual-arrays}--\eqref{eq:g-list} can be computed in polynomial bit time, with $r\le s+1$. Each $f_\ell$ has degree at most $D-1$, every product of $D$ residual transition matrices vanishes, and the $h$ successors are linearly independent and span $p,f_1,\ldots,f_r$.
\end{lemma}
\begin{proof}
We first adjoin one constant state to the input representation. Its transition matrices have the explicit form
$\begin{pmatrix}T_j&0\\0&0\end{pmatrix}$, its final column is $\begin{pmatrix}r_0\\1\end{pmatrix}$, and the row selecting the new state represents the constant polynomial. The columns obtained by applying transition words to this final column span the states distinguishable by an initial row. We compute a basis by repeatedly adjoining the images of the current columns under each transition and retaining independent columns. Each strict increase adds at least one dimension, so at most $s+1$ increases occur. Restricting to these columns makes equality of represented scalar series equivalent to equality of their coordinate rows.

In this reduced representation, we start with the constant row and the rows $[\ell\ 0]\begin{pmatrix}T_j&0\\0&0\end{pmatrix}$ for $1\le j\le k$. We repeatedly close their row span under right multiplication by the transition matrices. Again the dimension increases at most $s+1$ times. A retained rational basis gives $f_1,\ldots,f_r$, and coordinate solving gives Eq.~\eqref{eq:residual-arrays}. The first identity in Eq.~\eqref{eq:residual-identities} is the decomposition into the constant coefficient and the possible initial letters; the second is the same decomposition for $p$. Taking a residual reduces degree, so every sufficiently long residual of any $f_\ell$ vanishes, giving the joint nilpotence assertion. Constants and polynomials with different initial letters occupy disjoint coefficient positions. A relation among Eq.~\eqref{eq:g-list} would therefore force, for each letter, a relation among the independent $f_\ell$, as well as a zero constant coefficient. All its coefficients must be zero.

To bound the bit cost, we retain rows and columns that are products of the original rational matrices of length at most $s+1$. A common input denominator bounds the bit length of every such product polynomially in the input size. Coordinate solutions are ratios of minors of these retained rational arrays; fraction-free elimination controls their encoding length~\cite{bareiss1968}. Thus the basis algorithm cited before the lemma yields the required rational coordinates with a bit bound, in addition to its arithmetic-operation bound.
\end{proof}

Polynomiality, the zero case, and the true degree are checked on the supplied rational matrices before this construction; Appendix~\ref{app:coefficient-inputs} gives those tests and the conversion from an explicit branching program. With the order $0,(1,1),\ldots,(1,r),\ldots,(k,r)$ for the indices of $g$, we record the resulting polynomial identities in the numerical matrices
\begin{equation}\label{eq:EJP}
 E=[1\ 0\ \cdots\ 0],\qquad
 J=[a_0\ A_1^{\rm res}\ \cdots\ A_k^{\rm res}],\qquad
 P=[c_\emptyset\ p_1\ \cdots\ p_k].
\end{equation}
Their sizes are $1\times h$, $r\times h$, and $1\times h$. Thus $f_\ell=\sum_tJ_{\ell t}g_t$ and $p=\sum_tP_tg_t$. We denote by $L_j$ the $r\times h$ zero--one matrix with $(L_j)_{\ell,(j,t)}=\delta_{\ell t}$ and every other entry zero. It records the identity $x_jf_\ell=\sum_t(L_j)_{\ell t}g_t$. Thus $Eg=1$, $Jg=f$, $Pg=p$, and $L_jg=x_jf$. These four identities supply every coordinate map needed below.

\subsection{The complete semidefinite certificate}
The auxiliary column $H$ enters the defect identity through its coefficient Gram matrix. Assigning an element of $\operatorname{span}f$ to each input block label $\alpha$ and each memory label gives a coefficient matrix $K$ with $br$ columns, grouped by $\alpha$. Its Gram matrix $T=K^\dagger K$ is one of the unknown positive matrices. Known coordinate maps express both $H$ and $\Delta_mH$ in the same basis, so their difference of Gram matrices is linear in $T$.

For each input block label $1\le\alpha\le b$, we define $E_\alpha$ to have $br$ rows and $h$ columns, with $J$ in block row $\alpha$ and zeros elsewhere. For each output block label $1\le\mu\le a$, we define
\begin{equation}\label{eq:EM}
 E_\alpha=
 \begin{pmatrix}0\\\vdots\\J\\\vdots\\0\end{pmatrix},\qquad
 M_\mu=
 \begin{pmatrix}
 \sum_j f_{\mu1,j}L_j\\
 \sum_j f_{\mu2,j}L_j\\
 \vdots\\
 \sum_j f_{\mu b,j}L_j
 \end{pmatrix}.
\end{equation}
The columns $KE_\alpha g$ contain the internal polynomials at input label $\alpha$, and $KM_\mu g$ contains their image at output label $\mu$. Their difference of squared norms is therefore represented by the real-linear map
\begin{equation}\label{eq:Phi}
 \Phi(T)=\sum_{\alpha=1}^bE_\alpha^\dagger T E_\alpha
                   -\sum_{\mu=1}^aM_\mu^\dagger T M_\mu.
\end{equation}
This map takes a Hermitian matrix of order $br$ to one of order $h$, with rational coefficients in its real and imaginary coordinates. A second positive Gram matrix $S$ represents $G^*G$ in the basis $g$. The two terms separate polynomial sums of squares~\cite{helton2002} from the input-domain constraint, as in free convex Positivstellens\"atze~\cite{heltonklepmccullough2012}. The desired defect certificate will consequently be the affine equation $\beta^2E^\dagger E-P^\dagger P=S+\Phi(T)$. Its coordinate matrices $E,P$ and the map $\Phi$ are computed from the target and the input block; only the two positive semidefinite matrices $S,T$ remain to be found.

To display the certificate matrices explicitly, we use the smaller target $q(x)=(1+x_1x_2)/4$ on the diagonal domain. The commutator example continues to track word order and degree compensation; this example will follow the numerical calculation from the certificate to the four update blocks. Here $f=(1,x_2)^{\mathsf T}$ and
\[
 g=(1,x_1,x_1x_2,x_2,x_2^2)^{\mathsf T},\qquad
 E=(1,0,0,0,0),\qquad P_q=(1,0,1,0,0)/4.
\]
Thus the unknown matrices have orders five and four. The following pair certifies $1-q^*q$:
\begin{equation}\label{eq:example-certificate}
 T=\tfrac1{64}\begin{pmatrix}1&0&0&0\\0&8&0&0\\0&0&16&0\\0&0&0&1\end{pmatrix},
 \qquad S=\tfrac1{64}\begin{pmatrix}
 43&0&-4&0&0\\0&1&0&0&0\\-4&0&4&0&0\\
 0&0&0&7&0\\0&0&0&0&1
 \end{pmatrix}.
\end{equation}
Substitution into Eq.~\eqref{eq:Phi} gives $\Phi(T)=\operatorname{diag}(17,-1,-8,-7,-1)/64$ and hence $E^\dagger E-P_q^\dagger P_q=S+\Phi(T)$. Both matrices are positive definite: in $64S$, the only coupled $2\times2$ block has positive leading entry $43$ and determinant $156$, and the remaining diagonal entries are positive. Section~\ref{sec:extract} will use this pair to compute the four transfer blocks.

For the general certificate theorem, we prove completeness for every Hermitian coefficient matrix $Q_0$ of order $h$. This additional generality will also produce the interior margin for certificate discovery. We write
\begin{equation}\label{eq:hereditary}
 q_{Q_0}(x)=\sum_{t,u}(Q_0)_{tu}\,g_t(x)^*g_u(x).
\end{equation}
Such a polynomial is called \emph{hereditary}~\cite{mcculloughputinar2005}: its starred letters precede its unstarred letters. Each ordered word pair therefore has a fixed position in its coefficient Gram. The target defect is obtained by taking $Q_0=\beta^2E^\dagger E-P^\dagger P$.

We must prove that restricting the auxiliary polynomials to these residuals preserves every feasible norm bound. We use separation and finite-dimensional representation, following the method of hereditary positivity~\cite[Theorem~1.11 and Section~3]{augat2018}. Helton, Klep, and Nelson obtain finite positivity certificates on chip spaces spanned by monomials~\cite[Section~1.5 and Corollary~4.4]{heltonklepnelson2015}, while finite sum-of-squares formulas also underlie recent realization results~\cite{bhattacharyya2026}. Our basis can consist of whole residual polynomials, as in the preceding power example. The completeness argument must therefore preserve their linear relations, as well as the bound on the rectangular input block. A separating matrix specifies candidate Gram vectors for the polynomials in $g$. We will construct input matrices that act on these vectors by restoring their initial letters, prove that their joint block is contractive, and use residual closure to identify the vectors with actual polynomial evaluations. The separator then becomes a counterexample to the assumed norm bound.

\begin{theorem}[An exact certificate in residual coordinates]\label{thm:certificate}
The polynomial $q_{Q_0}$ is positive semidefinite on $\|\Delta(A)\|\le1$ at every matrix size if and only if $Q_0=S+\Phi(T)$ for positive semidefinite matrices $S,T$ of orders $h$ and $br$. Consequently, for $\beta\ge0$,
\begin{equation}\label{eq:certificate}
 \|p\|_{\SA,\Delta}\le\beta
 \quad\text{if and only if}\quad
 \beta^2E^\dagger E-P^\dagger P=S+\Phi(T),\qquad S,T\succeq0.
\end{equation}
\end{theorem}
\begin{proof}
A feasible pair gives a positive expression on the input domain. To see this, a factor $K$ with $T=K^\dagger K$ determines scalar polynomials
\[
 h_{\alpha,\ell}(x)=\sum_{t=1}^rK_{\ell,(\alpha,t)}f_t(x),
 \qquad
 b_{\mu,\ell}(x)=\sum_{\alpha=1}^b\sum_{j=1}^k
                       f_{\mu\alpha,j}x_jh_{\alpha,\ell}(x).
\]
Here $\ell$ indexes the rows of the numerical factor, and each resulting polynomial acts on the original signal. The coordinate definitions imply that the polynomials associated with $K E_\alpha$ and $K M_\mu$ are exactly $h_{\alpha,\ell}$ and $b_{\mu,\ell}$. Therefore $\Phi(T)$ represents
\begin{equation}\label{eq:localizing-sum}
 \sum_{\ell,\alpha}h_{\alpha,\ell}(x)^*h_{\alpha,\ell}(x)
       -\sum_{\ell,\mu}b_{\mu,\ell}(x)^*b_{\mu,\ell}(x).
\end{equation}
At any admissible matrix tuple, contractivity of the explicit block matrix in Eq.~\eqref{eq:pencil} makes this expression positive semidefinite: for every signal vector, we apply $\Delta(A)$ to the $b$ components $h_{\alpha,\ell}(A)$ for each fixed $\ell$ and sum their norm inequalities. A factor of $S$ similarly produces a sum of known polynomial squares. Hence $Q_0=S+\Phi(T)$ implies positivity of Eq.~\eqref{eq:hereditary}, proving sufficiency.

For necessity, we first identify hereditary polynomials with their numerical coefficient matrices. The coefficients of the independent $g_t$ form a full-row-rank matrix $C_g$, with rows indexed by $t$ and columns by words. Since each $x_v^*x_w$ has a unique starred--unstarred division, $q_{Q_0}=0$ implies $C_g^\dagger Q_0C_g=0$. Multiplication by a right inverse of $C_g$ gives $Q_0=0$, proving that this coefficient identification is injective.

Before applying separation, we prove that the image cone is closed by bounding its positive semidefinite preimages, following the finite-functional method used for noncommutative sums of squares~\cite[Section~3, Proposition~3.4]{mcculloughputinar2005}. For this purpose, we use the finite matrices $S_j^{(D)}$ whose rows and columns are indexed by words of length at most $D$, with entry one at row $jw$, column $w$ when $|w|<D$, and zero elsewhere. Their distinguished coordinate is the empty word. A sufficiently small positive rational $\rho$ makes $\Delta(\rho S^{(D)})$ a contraction with norm at most $1/2$. Evaluation at this tuple on the empty-word coordinate defines a positive functional $\ell_\rho$ on the hereditary polynomials. The weighted coefficient Grams of the independent lists $f_\ell$ and $g_t$ are positive definite. Consequently there is $\chi>0$ such that, for $S,T\succeq0$,
\begin{equation}\label{eq:coercive}
 \ell_\rho\bigl(q_{S+\Phi(T)}\bigr)
                      \ge\chi(\tr S+\tr T).
\end{equation}
The independent coefficient lists and the strict defect bound give $\chi>0$, which suffices for closedness. Appendix~\ref{app:classical-margins} computes a rational choice with controlled bit length from the residual recurrence. Thus a convergent sequence $S_v+\Phi(T_v)$ has bounded positive semidefinite preimages $S_v,T_v$. A subsequence of these finitely many bounded matrix entries converges to positive semidefinite $S,T$, proving closedness of the image cone.

We now assume that a positive $q_{Q_0}$ has no such representation and derive a contradiction. Finite-dimensional separation of the closed convex cone gives a Hermitian numerical matrix $H$ of order $h$, acting on the coefficient basis $g$, such that
\begin{equation}\label{eq:separator}
 H\succeq0,\qquad
 \Phi^*(H):=\sum_\alpha E_\alpha H E_\alpha^\dagger
                   -\sum_\mu M_\mu H M_\mu^\dagger\succeq0,
 \qquad \tr(HQ_0)<0.
\end{equation}
We now realize the separator as polynomial evaluation on an allowed tuple. Its Gram vectors must support multiplication by each $x_j$, including when the separating matrix is singular. We adapt the finite Gram construction in hereditary positivity proofs~\cite{augat2018}: its columns represent polynomials, and multiplication by a letter must reproduce the corresponding successor column. The kernel argument below makes this construction well defined for a possibly singular separator in our residual space. For the functional $\ell(q_{Q_0})=\operatorname{Tr}(HQ_0)$, the coefficient pairing is $\ell(g_t^*g_u)=H_{ut}$; this fixes the ordinary transpose in the factorization below.

We factor $H^{\mathsf T}=W^\dagger W$, choosing $W$ with $n_H=\rank H$ independent rows. With ordinary transpose in this factorization, we form
\begin{equation}\label{eq:GNS-matrices}
 R_H=WJ^{\mathsf T},\qquad Z_{j,H}=WL_j^{\mathsf T},\qquad
 Y_j=Z_{j,H}R_H^+,\qquad \omega=WE^{\mathsf T}.
\end{equation}
The first two matrices have $n_H$ rows and $r$ columns, $Y_j$ is $n_H\times n_H$, and $\omega$ is a numerical column of length $n_H$. To justify the pseudoinverse construction, the ordinary transpose of $\Phi^*(H)$ is the $b\times b$ block matrix
\begin{equation}\label{eq:GNS-localizing}
 \left[\delta_{\alpha\beta}R_H^\dagger R_H
     -\sum_{\mu=1}^a B_{\mu\alpha}^\dagger B_{\mu\beta}\right]_{\alpha,\beta=1}^b
 \succeq0,
 \qquad B_{\mu\alpha}=\sum_j f_{\mu\alpha,j}Z_{j,H}.
\end{equation}
Each block has order $r$. If $R_Hv=0$, applying Eq.~\eqref{eq:GNS-localizing} with $v$ in input block $\alpha$ and zeros elsewhere gives $B_{\mu\alpha}v=0$ for all $\mu,\alpha$. Independence of the coefficient matrices $F_j$ forces $Z_{j,H}v=0$ for each $j$. We have therefore established $\ker R_H\subseteq\ker Z_{j,H}$, and hence
\begin{equation}\label{eq:GNS-reproduction-step}
 Y_jR_H=Z_{j,H}.
\end{equation}
Applying Eq.~\eqref{eq:GNS-localizing} to the columns $R_H^+x_\alpha$ now gives
\[
 \sum_\mu\left\|\sum_{\alpha,j}f_{\mu\alpha,j}Y_jx_\alpha\right\|^2
 \le\sum_\alpha\|R_HR_H^+x_\alpha\|^2
 \le\sum_\alpha\|x_\alpha\|^2.
\]
Thus the explicit block matrix $\Delta(Y)$ is contractive.

Contractivity makes $Y$ an admissible input. To obtain the contradiction, its polynomial evaluations must also reproduce the Gram vectors prescribed by the separator. Restoring one initial letter already acts correctly by Eq.~\eqref{eq:GNS-reproduction-step}; the following induction shows that this property extends to every polynomial in the residual space because its remaining-letter polynomials stay in that space.
For a row $\lambda=(\lambda_1,\ldots,\lambda_r)$, we write
$f_\lambda=\sum_\ell\lambda_\ell f_\ell$ and define its represented
column by $\zeta(f_\lambda)=R_H\lambda^{\mathsf T}$; the independence of the
$f_\ell$ ensures that this definition is unambiguous.  The residual arrays
give $\partial_jf_\lambda=f_{\lambda A_j^{\rm res}}$ and
$(f_\lambda)_\emptyset=\lambda a_0$.  At the level of the $h$ numerical
coordinates, the same identities read
\[
 J^{\mathsf T}\lambda^{\mathsf T}
 =E^{\mathsf T}(\lambda a_0)
   +\sum_jL_j^{\mathsf T}(A_j^{\rm res})^{\mathsf T}
                         \lambda^{\mathsf T}.
\]
Multiplication by $W$, followed by $Z_{j,H}=Y_jR_H$, therefore proves
\[
 \zeta(f_\lambda)
 =\omega(f_\lambda)_\emptyset
       +\sum_jY_j\zeta(\partial_jf_\lambda).
\]
If $f_\lambda$ is constant, this is already
$\zeta(f_\lambda)=f_\lambda(Y)\omega$.  Otherwise every polynomial in the
sum has smaller degree, so induction gives that equality as well.  In
particular, the column of $f_\ell$ is $R_He_\ell$, while the column of
$x_jf_\ell$ is $Y_jR_He_\ell=Z_{j,H}e_\ell=We_{(j,\ell)}$.
Together with the constant column $We_0=\omega$, these equalities show that every successor polynomial is reproduced:
\begin{equation}\label{eq:GNS-reproduction}
 g_t(Y)\omega=We_t\qquad(0\le t<h),
\end{equation}
where $e_t$ is the coordinate column of length $h$. The induction terminates because the residual degree decreases at each step. Using the reproduced columns to evaluate $q_{Q_0}$ now gives
\[
 \omega^\dagger q_{Q_0}(Y)\omega
 =\sum_{t,u}(Q_0)_{tu}(W^\dagger W)_{tu}
 =\sum_{t,u}(Q_0)_{tu}H_{ut}
 =\tr(HQ_0)<0.
\]
This contradicts positivity of $q_{Q_0}$ at all admissible finite tuples. It proves the general certificate. For $Q_0=\beta^2E^\dagger E-P^\dagger P$, the hereditary expression is $\beta^2-p^*p$, so the claimed norm equivalence follows.
\end{proof}

The separation argument uses exactly the closure that the coefficient recurrence supplies. It establishes completeness of the finite search space, including its singular boundary. The same construction gives norm witnesses of order at most $h$; Corollary~\ref{cor:finite-norm} in Appendix~\ref{app:arithmetic} records that consequence. Combining the certificate theorem with the finite-matrix construction of Section~\ref{sec:sa-foundation} now gives the following criterion.

\begin{corollary}[Finite polynomial Schur--Agler criterion]\label{thm:sa-realization}
For a nonconstant rational scalar polynomial $q$ of degree $D$, the following are equivalent: $\|q\|_{\SA,\Delta}\le1$; the existence of a polynomial defect certificate in Eq.~\eqref{eq:polynomial-agler}; and the existence of a finite contractive transfer realization. The certificate can use $m\le br$, where $r$ is the residual dimension of Lemma~\ref{lem:residual} for $q$.
\end{corollary}
\begin{proof}
Applying Theorem~\ref{thm:certificate} to $1-q^*q$ gives the certificate in residual coordinates. A square factor of its $br\times br$ Gram matrix gives $m=br$, with degrees at most $D-1$ for $H$ and $D$ for $G$. Lemma~\ref{lem:certificate-to-realization} gives the finite realization, and Eq.~\eqref{eq:transfer-energy} proves its norm bound. These steps establish the equivalence and the stated finite size.
\end{proof}

We now use a strictly feasible certificate to compute the numerical query parameters and control the precision of their evaluation.

\subsection{Extracting the numerical query parameters}\label{sec:extract}
Suppose we have a positive definite certificate for a normalized target $q$ of degree $D$. We first determine the query parameters from this certificate; the next subsection computes the certificate itself. A factor of $T$ supplies the coefficients of the auxiliary column $H$, and the common basis $g$ expresses both $[\Delta_mH;1]$ and $[H;q]$. We can therefore apply the Gram factorization of Section~\ref{sec:sa-foundation} directly to numerical matrices.

Let $S,T\succ0$ be the rational certificate matrices for $q$. In its residual coordinates, we set $m=br$, take $K=T^{1/2}$ to be the full $m\times br$ positive square root, and form the two ordinary numerical matrices
\begin{equation}\label{eq:XY}
 X=\begin{pmatrix}KM_1\\\vdots\\KM_a\\E\end{pmatrix},
 \qquad
 Y=\begin{pmatrix}KE_1\\\vdots\\KE_b\\P_q\end{pmatrix},
 \qquad
 X:\ (am+1)\times h,\quad
 Y:\ (bm+1)\times h.
\end{equation}
The row $P_q$ contains the coordinates of $q$ itself, including its normalization. The rows of $X$ describe the queried auxiliary column and scalar input, while those of $Y$ describe the auxiliary column and target output. All their entries are known from the certificate. Computing their Gram matrices gives
\begin{equation}\label{eq:XY-Gram}
 T_0=X^\dagger X=E^\dagger E+\sum_\mu M_\mu^\dagger T M_\mu,
 \qquad X^\dagger X-Y^\dagger Y=S\succ0.
\end{equation}
In particular, $T_0$ is rational and invertible. Independence of the $a\times b$ matrices $F_j$ implies $k\le ab$, so $am+1=abr+1\ge kr+1=h$; thus $X$ has enough rows for the full column rank forced by its positive Gram. The extraction below produces a numerical matrix with $bm+1$ rows and $am+1$ columns, matching the queried coordinates and scalar endpoints in Eq.~\eqref{eq:V}. We compute
\begin{equation}\label{eq:V-extraction}
 V=YT_0^{-1}X^\dagger
   =\begin{pmatrix}\Aco&\Bco\\\Cco&\Dco\end{pmatrix}
\end{equation}
with the block sizes in Eq.~\eqref{eq:V}. Every factor in this expression is computed from the rational certificate and its known positive square root.

For the certificate in Eq.~\eqref{eq:example-certificate}, $m=4$ and $h=5$, so $X,Y$ have nine rows and five columns. We obtain
\[
 K=\operatorname{diag}(1,\sqrt8,4,1)/8,\qquad
 T_0=\operatorname{diag}(1,1/64,1/8,1/4,1/64).
\]
Let $e_1,\ldots,e_8$ denote the coordinate columns of length eight, with indices starting at one. Substituting into Eq.~\eqref{eq:V-extraction} gives all four blocks:
\begin{equation}\label{eq:example-transfer-blocks}
 \Aco=\left(\frac{e_2}{\sqrt2}+\frac{e_8}{4}\right)e_7^{\mathsf T},\quad
 \Bco=\frac{e_1}{8}+\frac{e_7}{2},\quad
 \Cco=\frac{e_2^{\mathsf T}}{\sqrt2},\quad \Dco=\frac14.
\end{equation}
Here $F_1^{(4)}=\operatorname{diag}(I_4,0_{4\times4})$ and $F_2^{(4)}=\operatorname{diag}(0_{4\times4},I_4)$. The only nonzero quadratic path is
\[
 \Cco F_1^{(4)}\Aco F_2^{(4)}\Bco
 =\frac{e_2^{\mathsf T}}{\sqrt2}
       \frac{e_2}{2\sqrt2}=\frac14.
\]
The linear coefficients vanish because $\Cco$ selects coordinate two, while $\Bco$ occupies coordinates one and seven. After one nonzero transition $\Aco F_j^{(4)}$, the internal output occupies coordinates two and eight, which every further such transition kills. Thus these blocks generate exactly $(1+x_1x_2)/4$. Section~\ref{sec:unitary-implementation} completes this example by identifying the two-query output and its selected registers; Appendix~\ref{app:worked-scalar} gives the detailed coordinates. For an arbitrary strict certificate, we next specify how the known matrix entries are returned with controlled precision.

For the numerical gates, we describe an $h_2\times h_1$ matrix by its dimensions and an entry algorithm that, on input $(r,s,p)$, returns a complex rational within $2^{-p}$ of entry $(r,s)$. Evaluation is polynomial-time when its cost is polynomial in the description length, dimensions, and $p$. Entry error at most $\eta/\sqrt{h_1h_2}$ gives operator-norm error at most $\eta$. The spectral bounds of the certificate and the arithmetic in Appendix~\ref{app:arithmetic} give such evaluators for the roots, inverses, and products used below.\label{def:matrix-evaluator}

\begin{lemma}[Computable transfer realization]\label{lem:extract}
The matrix in Eq.~\eqref{eq:V-extraction} is a strictly contractive transfer realization of $q$. It satisfies Eq.~\eqref{eq:coefficient-factorization} at every word and has certified entry evaluators with polynomial bit cost in the rational certificate length and the requested precision bits.
\end{lemma}
\begin{proof}
The equality $VX=Y$ follows from the definition of $T_0$. If $S\succeq\mu I_h$ and $T_0\preceq t_0I_h$, then
\[
 Y^\dagger Y\preceq(1-\mu/t_0)X^\dagger X.
\]
On a column $Xz$, the inequality gives $\|VXz\|\le\sqrt{1-\mu/t_0}\|Xz\|$, while the factor $X^\dagger$ makes $V$ vanish on the orthogonal complement of these columns. Hence $\|V\|^2\le1-\mu/t_0<1$, with the rational bounds $\mu>0$ and $t_0$ computed from the strictly feasible certificate. This quantitative form of Douglas factorization~\cite{douglas1966} supplies the contraction margin used to evaluate the known gates.

To apply Lemma~\ref{lem:certificate-to-realization}, we form the polynomial column $H=[KE_1g;\ldots;KE_bg]$. The identities $L_jg=x_jf$ and Eq.~\eqref{eq:EM} show that the matrices in Eq.~\eqref{eq:XY} are exactly the coefficient matrices of $[\Delta_mH;1]$ and $[H;q]$ in the independent basis $g$. Since $X$ has full column rank, its Moore--Penrose inverse is $T_0^{-1}X^\dagger$, so the present formula is the finite factor from that lemma. In particular, Eq.~\eqref{eq:h-recurrence} gives every word coefficient. The entries of $H$ have degree at most $D-1$, which makes the higher coefficients zero and supplies a polynomial solution of Eq.~\eqref{eq:sa-feedback}. Here the strict bound on $V$ also makes the transfer inverse well defined on the closed domain.

Exact elimination and minor bounds give rational entries of polynomial encoding length for $T_0$ and its inverse. Positive definiteness of the rational matrices $S,T$ gives lower eigenvalue bounds with polynomial logarithmic size by the same minor estimates. Together with $T_0\succeq S$, these bounds control the inverse and square-root evaluations in Eq.~\eqref{eq:V-extraction}, so the finite matrix products have certified entry approximations in polynomial time. Appendix~\ref{app:arithmetic} supplies the corresponding guard-bit analysis, including the singular roots that occur later in Julia gates. Using the full square root of $T$ sets $m=br$, so both the multiplicity and the precision cost of the contraction are controlled by the computed certificate.
\end{proof}

\subsection{Finding a certificate near the optimal normalization}\label{sec:conditioning}
We now find the certificate at a normalization strictly larger than the target norm. To bound the number of classical steps, we need both a bounded search region and a positive distance from its feasible certificate to the boundary. A strict norm promise $\|p/\gamma\|_{\SA,\Delta}\le\nu<1$ supplies these bounds through the unused squared norm $1-\nu^2$. The general hereditary statement of Theorem~\ref{thm:certificate} first gives a certificate of the constant $1$ with both Gram matrices positive definite: subtracting a sufficiently small rational multiple of the uniformly bounded polynomial represented by $I_h+\Phi(I_{br})$ preserves positivity, and its certificate can then be rearranged. Adding $1-\nu^2$ times this constant certificate to one for $\nu^2-p^*p/\gamma^2$ supplies a strict certificate for the normalized target. 

Lemma~\ref{lem:margin} bounds the smallest eigenvalues of both certificate matrices, while Eq.~\eqref{eq:trace-cap} bounds their total trace. Appendix~\ref{app:classical-margins} computes these rational bounds and proves that their logarithms are polynomial in the input length and $\log(1/(1-\nu))$.

\Needspace{8\baselineskip}
\begin{proposition}[Polynomial-bit strict certificate discovery]\label{prop:strict}
For rational $\gamma>0$ and $0\le\nu<1$ with $\|p/\gamma\|_{\SA,\Delta}\le\nu$, a deterministic algorithm returns rational positive definite $S,T$ satisfying Eq.~\eqref{eq:certificate} for $p/\gamma$ exactly. Its bit complexity is polynomial in the rational finite-state and joint-encoding input length, the bit lengths of $\gamma$ and $\nu$, and $\log(1/(1-\nu))$.
\end{proposition}
\begin{proof}
Eliminating $S$ leaves the rational matrix inequalities in Eq.~\eqref{eq:search-region} for the real coordinates of $T$. The strict margins give a contained ball and a bounded search region with polynomial logarithmic radius ratio. Exact Hermitian elimination tests each matrix inequality and supplies a rational separating vector when it fails. The rounded ellipsoid method~\cite{groetschell1981,groetschell1988} therefore finds a rational feasible point within a predetermined polynomial iteration cap. The recovered matrix $S=E^\dagger E-P^\dagger P/\gamma^2-\Phi(T)$ satisfies the coefficient identity exactly. Appendix~\ref{app:classical-margins} proves the radius and bit bounds, and Appendix~\ref{app:arithmetic} specifies the finite-bit separation and rounding operations.

Every returned pair is checked against the exact rational inequalities, so its certificate remains valid even if the promised norm bound was false. Conversely, if the procedure reaches its cap without a verified pair, the iteration guarantee rules out the interior ball implied by that promise and hence disproves the promise itself; this provides the lower-bound test used in norm estimation below.
\end{proof}

For a nonzero target, we can compute positive rational lower and upper bounds on $B=\|p\|_{\SA,\Delta}$ from the same coefficient recurrence. We then use the strict search to refine this bracket. The search has two useful outcomes at a proposed threshold $b_0$: a verified certificate proves a nearby upper bound, while reaching the iteration cap proves $B>b_0$.

\begin{theorem}[Relative norm estimation]\label{thm:norm-search}
For a nonzero rational finite-state polynomial and rational $0<\eta<1$, one can compute a rational number $\widehat B$ satisfying $B\le\widehat B\le(1+\eta)B$ in time polynomial in the input length and $\log(1/\eta)$.
\end{theorem}
\begin{proof}
For a rational threshold $b_0>0$, we run the procedure of Proposition~\ref{prop:strict} with $\gamma=(1+\eta/8)b_0$ and $\nu=(1+\eta/8)^{-1}$. If $B\le b_0$, its promise holds and the search returns a certificate. Every returned certificate proves $B\le(1+\eta/8)b_0$ by exact verification. If the iteration cap is reached, the promise is false, so $B>b_0$.

We use the midpoint of the current bracket $[\ell,U]$ as $b_0$. A failed search raises $\ell$ to $b_0$; a successful search lowers $U$ to $\min\{U,(1+\eta/8)b_0\}$. These updates preserve $\ell\le B\le U$. Appendix~\ref{app:classical-norm-search} supplies the initial bracket and proves that its width decreases by a fixed factor whenever $U>(1+\eta)\ell$. After polynomially many iterations in the input bit length and $\log(1/\eta)$, the rational upper endpoint is the required $\widehat B$. Constant targets are handled there by rational bisection of $|c_\emptyset|$.
\end{proof}

The two routines provide exactly the classical inputs needed by the quantum construction. We first estimate the norm of the requested polynomial $p$ and choose its output normalization. For a row input we then certify $p/\beta$ directly. For a general input, Section~\ref{sec:general} first adjusts each homogeneous degree by a known scalar and proves a strict norm bound for the resulting polynomial $q$. Applying Proposition~\ref{prop:strict} to that $q$, followed by Lemma~\ref{lem:extract}, computes its numerical realization. Thus each invocation takes an explicit coefficient recurrence and a proved margin, and returns the matrices used by the circuit.

\section{Quantum circuit synthesis}\label{sec:general}
The contraction from the preceding section provides the factors of every ordered coefficient. We now arrange these factors so that one finite circuit has selected block $p(A)/\beta$. The construction has two parts: a sequence of known contractions combines contributions of different degrees, and their unitary completions implement that sequence using the supplied oracle. Computing the polynomial produced by the first part tells us how to adjust the classical target before finding its realization. For a row input we will instead factor a complementary polynomial column, obtaining one query per degree.

\subsection{A finite sequence and its normalization}\label{sec:weighted-sequence}
We first fix a contraction $V$ in Eq.~\eqref{eq:V}, with multiplicity $m$, and write $q_d$ for the degree-$d$ polynomial generated by its coefficient products. To combine them using $Q\ge1$ queries, we allow an amplitude to enter the active calculation at any of the known updates $0,\ldots,Q$ and to leave at a later update. The number of intervening active queries then determines its degree.

We store an active component $h$, a pending input $s$, and a collected output $y$ on disjoint ancillary coordinates. Before an update, these are amplitude columns of lengths $amn,n,n$; afterwards the active column has length $bmn$. The next query acts by $\Delta_m(A)$ on that active column and by identity on the other two. For a signal input $\psi$, we start with $(h,s,y)=(0,\psi,0)$ and select $y$ after the last update. Each selected contribution is thereby specified by the time at which it enters and the time at which it is collected.

We choose positive weights $v_0,\ldots,v_Q$ with $\sum_tv_t^2=1$ so that injection at $t$ and collection at $r$ contribute amplitudes $v_t$ and $v_r$. The tail norm measures the input still available at each stage, while the prefix norm determines how a newly collected output is combined with the previous outputs. These norms give the splitting and collection ratios
\begin{equation}\label{eq:weights}
 \begin{aligned}
 u_t&=\left(\sum_{s=t}^{Q}v_s^2\right)^{1/2},&
 z_t&=\left(\sum_{s=0}^{t}v_s^2\right)^{1/2},\\
 \alpha_t&=v_t/u_t,&\zeta_t&=u_{t+1}/u_t,\\
 \xi_t&=v_t/z_t,&\gamma_t&=z_{t-1}/z_t,
 \end{aligned}
 \qquad u_{Q+1}=z_{-1}=0.
\end{equation}
All denominators used in these definitions are positive. In particular, $u_0=z_Q=1$, and
$\alpha_t^2+\zeta_t^2=\xi_t^2+\gamma_t^2=1$.
Writing the corresponding signal amplitudes as $h,s,y$, we use the update
\begin{equation}\label{eq:stage-components}
 \begin{aligned}
 h'&=\Aco_n h+\alpha_t\Bco_n s,\qquad &s'&=\zeta_t s,\\
 y'&=\xi_t(\Cco_n h+\alpha_t\Dco_n s)+\gamma_t y.
 \end{aligned}
\end{equation}
Thus $\alpha_ts$ enters the transfer update, $\zeta_ts$ remains pending, and the emitted amplitude joins the collected output. The next query sends $h'$ to $\Delta_m(A)h'$ while preserving $s',y'$. Suppressing the common signal identity, this update is the numerical matrix
\begin{equation}\label{eq:Mt}
 M_t=
 \begin{pmatrix}
 \Aco&\alpha_t\Bco&0_{bm\times1}\\
 0_{1\times am}&\zeta_t&0\\
 \xi_t\Cco&\xi_t\alpha_t\Dco&\gamma_t
 \end{pmatrix}.
\end{equation}
The block order is active, pending, collected, and each numerical entry multiplies $I_n$. Figure~\ref{fig:three-components} follows these components of the same quantum state through the update and the next query.

\begin{figure}[t]
\centering
\begin{tikzpicture}[x=1cm,y=1cm,>=stealth,every node/.style={font=\small}]
 \node (s) at (0,1.8) {$s$};
 \node (h) at (0,0) {$h$};
 \node (y) at (0,-1.8) {$y$};
 \node (sp) at (4.8,1.8) {$s'$};
 \node (hp) at (4.8,0) {$h'$};
 \node (yp) at (4.8,-1.8) {$y'$};
 \node (sq) at (8,1.8) {$s'$};
 \node (hq) at (8,0) {$\Delta_m(A)h'$};
 \node (yq) at (8,-1.8) {$y'$};
 \draw[->] (s) -- node[above] {$\zeta_t$} (sp);
 \draw[->] (h) -- node[above,pos=.66] {$\Aco_n$} (hp);
 \draw[->] (y) -- node[below] {$\gamma_t$} (yp);
 \draw[->] (s) -- node[above,sloped,pos=.38] {$\alpha_t\Bco_n$} (hp);
 \draw[->] (h) -- node[above,sloped,pos=.6] {$\xi_t\Cco_n$} (yp);
 \draw[->] (s) to[bend right=46] node[left,pos=.52] {$\xi_t\alpha_t\Dco_n$} (yp);
 \draw[->] (sp) -- node[above] {$I_n$} (sq);
 \draw[->] (hp) -- node[above] {$\Delta_m(A)$} (hq);
 \draw[->] (yp) -- node[above] {$I_n$} (yq);
 \node at (2.4,2.5) {known update $M_t\otimes I_n$};
 \node at (6.4,2.5) {one query};
\end{tikzpicture}
\caption{The three components in Eq.~\eqref{eq:stage-components}. The input can remain pending, enter the active sector, or contribute its constant term directly to the output. Active amplitude can be updated or collected; collected amplitude remains there. The three right-hand arrows are sectors of one supplied query, whose full unitary is called unconditionally. The final update $M_Q$ has no following query.}
\label{fig:three-components}
\end{figure}
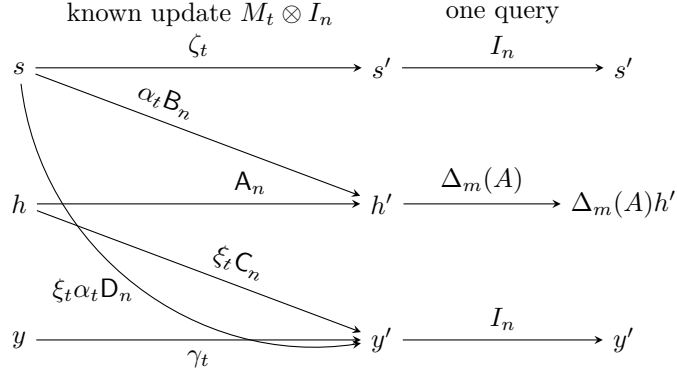

Each update is a contraction. To see this, we split the pending input with amplitudes $\alpha_t,\zeta_t$, apply $V$ to the active and injected components, and collect the emitted and previous outputs with amplitudes $\xi_t,\gamma_t$. The first map is an isometry and the last a coisometry, because both pairs have squared norm one. Appendix~\ref{app:circuit-packing} writes these maps explicitly as $M_t=T_tV_{\rm ext}S_t$ in Eq.~\eqref{eq:Mt-factors}.

In these three groups of coordinates, the selected query is
\begin{equation}\label{eq:logical-query}
 \mathcal B_m(A)=\begin{pmatrix}
 \Delta_m(A)&0&0\\0&I_n&0\\0&0&I_n
 \end{pmatrix}.
\end{equation}
Let $e_{am}$ be the column of length $am+2$ selecting the pending input, and let $e_{bm+1}^{\mathsf T}$ be the row of length $bm+2$ selecting the collected output, with indices starting at zero. The matrix whose selected action we require is
\begin{equation}\label{eq:logical-query-product}
 (e_{bm+1}^{\mathsf T}\otimes I_n)(M_Q\otimes I_n)
 \mathcal B_m(A)(M_{Q-1}\otimes I_n)\cdots
 \mathcal B_m(A)(M_0\otimes I_n)(e_{am}\otimes I_n).
\end{equation}
The right endpoint supplies the input, each $M_t$ updates the three components, and the left endpoint selects the collected output. We first evaluate this product; its embedding in the physical oracle will follow from the same coordinate assignments.

For these stage weights, we define the degree multipliers
\begin{equation}\label{eq:window-autocorrelation}
 w_d=\sum_{t=0}^{Q-d}v_tv_{t+d}\qquad(0\le d\le Q).
\end{equation}
Here $q_d(A)=\sum_{|w|=d}q_wA_w$, with $q_w$ from Eq.~\eqref{eq:coefficient-factorization}. The following matrix identity applies to any contractive transfer matrix; the polynomial realization computed in Section~\ref{sec:extract} is the case used by our algorithm.

\begin{theorem}[The polynomial implemented by the finite sequence]\label{thm:window-circuit}
For $Q\ge1$, a positive real unit weight vector, and the contraction $V$ in Eq.~\eqref{eq:V}, the matrix product in Eq.~\eqref{eq:logical-query-product} equals
\begin{equation}\label{eq:window-output}
                         \sum_{d=0}^{Q}w_dq_d(A).
\end{equation}
If the coefficients define a degree-$D$ polynomial with $D\le Q$, the sum ends at $D$.
\end{theorem}
\begin{proof}
We expand the selected product in Eq.~\eqref{eq:logical-query-product}. Its zero blocks allow either a direct constant contribution through $\Dco$, or one consecutive run of active queries between injection and collection. Before update $t$, the pending amplitude is $u_t$ times the original input, since each preceding update multiplies it by $\zeta_s=u_{s+1}/u_s$. Its injection weight is therefore $\alpha_tu_t=v_t$. An output emitted at update $r$ receives $\xi_r=v_r/z_r$ and is retained with total later factor
\[
 \prod_{s=r+1}^{Q}\gamma_s=z_r/z_Q=z_r.
\]
Its final collection weight is $v_r$, including $r=Q$. Thus an input injected at $t$ and collected at $t+d$ contributes $v_tv_{t+d}q_d(A)$. For a word $j_1\cdots j_d$, the active path supplies precisely
\[
 \Cco F_{j_1}^{(m)}\Aco F_{j_2}^{(m)}\cdots
       \Aco F_{j_d}^{(m)}\Bco=q_{j_1\cdots j_d},
\]
with the rightmost letter acting first. Summing the possible start times gives coefficient $w_dq_{j_1\cdots j_d}$ for this word; summing the words of length $d$ then gives $w_dq_d(A)$. At $d=0$ the same calculation gives $w_0\Dco_n=\Dco_n$, because $\sum_tv_t^2=1$.

If the represented function is a polynomial of degree $D$, its coefficient identity makes every $q_d$ with $d>D$ zero. Lemma~\ref{lem:extract} establishes this identity for the realization computed by our algorithm, proving the final assertion.
\end{proof}

The multipliers $w_d$ depend only on the degree. We can therefore cancel them in the classical target before finding its realization:
\begin{equation}\label{eq:mechanism-compensation}
 \widetilde p=\sum_{d=0}^{D}\frac{p_d}{w_d},\qquad
 q=\widetilde p/\beta,\qquad D\le Q.
\end{equation}
The resulting product equals $p(A)/\beta$. To keep $\beta\le(1+\tau)B$, we must choose weights for which $\|\widetilde p\|_{\SA,\Delta}$ remains close to $B$, leaving a strict margin for the certificate search. After the following fixed-length example, we prove that smooth weights achieve this with $Q=O(D/\sqrt\tau)$.

For the polynomial of Example~\ref{ex:commutator-thread}, two queries and three equal stage weights give $w_0=1$ and $w_2=1/3$. At normalization one, we thus realize
\[
 q(x)=\frac14+\frac{x_1x_2-x_2x_1}{4},
 \qquad \|q\|_{\SA,\Delta}=\frac34<1.
\]
The quadratic contribution can enter only at stage zero and leave at stage two, so its factor $1/3$ recovers Eq.~\eqref{eq:running-target}. The constant contribution enters and leaves at the same stage, and its weights sum to one. This illustrates compensation at a fixed query count; the smooth-weight construction will instead bring the normalization close to the target norm $5/12$.

\paragraph{Controlling the compensation.}
The compensation in Eq.~\eqref{eq:mechanism-compensation} is useful only if it preserves the target norm up to a small factor. Bounding its homogeneous terms separately could lose their interference. We instead use Fourier averaging of the whole polynomial~\cite{katznelson2004} along $A_j\mapsto e^{i\theta}A_j$. The identity $\Delta(e^{i\theta}A)=e^{i\theta}\Delta(A)$ ensures that this orbit stays in the same domain at every matrix size. With $N=Q+1$ and $D<N$, we denote the resulting degree-weighting map by $\mathcal W_vp=\sum_{d=0}^Dw_dp_d$. Smooth stage weights make this averaging map close to the identity in the same uniform norm.

\begin{lemma}[A rational choice of weights]\label{lem:weights}
For every integer $N\ge2$, a positive rational vector
$r=(r_0,\ldots,r_{N-1})$ can be computed in polynomial time in $N$
such that $v=r/\|r\|_2$ has positive rational autocorrelations $w_d$
for $0\le d<N$. For every integer $0\le D<N$, these weights satisfy,
on degree-at-most-$D$ scalar polynomials,
\begin{equation}\label{eq:window-bound}
 \|p-\mathcal W_vp\|_{\SA,\Delta}
 \le\rho_N\|p\|_{\SA,\Delta},
 \qquad \rho_N=\frac{64D^2}{(N+1)^2}.
\end{equation}
The entries of $r$ and the $w_d$ have bit length polynomial in $\log N$. If $\rho_N<1$, then
\begin{equation}\label{eq:inverse-window}
 \mathcal W_v^{-1}p=\sum_{d=0}^D p_d/w_d,
 \qquad
 \|\mathcal W_v^{-1}p\|_{\SA,\Delta}
                \le\frac{\|p\|_{\SA,\Delta}}{1-\rho_N}.
\end{equation}
\end{lemma}
\begin{proof}
We use the real unit vector
\[
 s_t=\sqrt{\frac2{N+1}}\sin\frac{(t+1)\pi}{N+1},\qquad 0\le t<N.
\]
With $s_{-1}=s_N=0$, the recurrence $s_{t-1}+s_{t+1}=2\cos(\pi/(N+1))s_t$ gives $\sum_ts_ts_{t+1}=\cos(\pi/(N+1))$. We compute dyadic approximations $r_t$ with error at most $[1000N(N+1)^2]^{-1}$, using convergent sine and square-root series. The bound $\sin x\ge2x/\pi$ on $[0,\pi/2]$ puts every $s_t$ above a fixed positive multiple of $(N+1)^{-3/2}$, so all $r_t$ are positive and require $O(\log N)$ bits.

The normalized rational approximation obeys
$\|v-s\|_2\le[100(N+1)^2]^{-1}$. The shift matrix with ones on its first superdiagonal has norm one, and therefore
\begin{equation}\label{eq:w1}
 1-w_1\le1-\cos\frac\pi{N+1}+2\|v-s\|_2
            <\frac5{(N+1)^2}.
\end{equation}
The normalization of $r$ cancels from its autocorrelations, leaving the rational ratios
$w_d=(\sum_t r_tr_{t+d})/(\sum_t r_t^2)$. Positivity of $r$ gives $w_d>0$ for $d<N$. A common dyadic denominator and sums of $N$ terms give the claimed logarithmic bit bounds.

For $D=0$, both claims are immediate from $w_0=1$. We therefore assume $D\ge1$ and prove a norm estimate independent of the number of words. For an admissible tuple $A$, let $g(\theta)=p(e^{i\theta}A)$. On the real line its norm is at most $B=\|p\|_{\SA,\Delta}$. More generally, the reversed polynomial $z^Dp(A/z)$ is analytic at zero and has norm at most $B$ on $|z|=1$. Applying the scalar maximum principle after left and right multiplication by arbitrary unit vectors gives $\|p(tA)\|\le|t|^DB$ when $|t|\ge1$; when $|t|\le1$, the tuple $tA$ is admissible. On the complex $\theta$-circle of radius $2/D$ around a real point, it follows that $\|g(\theta)\|\le e^2B$. For $D\ge1$, Cauchy's estimate consequently gives
\[
 \|g''(\theta)\|\le\frac{e^2D^2}{2}B\qquad(\theta\in\R).
\]

The nonnegative even kernel
$K_v(\theta)=|\sum_t v_te^{it\theta}|^2$ has normalized integral one and Fourier coefficients $w_d$. Thus its average of $g$ is exactly $\mathcal W_vp(A)$. Because the kernel is even, we may integrate the central second difference of $g$, whose norm is at most $\theta^2\sup\|g''\|/2$. This yields
\begin{align}
 \|p(A)-\mathcal W_vp(A)\|
 &\le \frac{e^2D^2B}{4}
       \int_{-\pi}^{\pi}\theta^2K_v(\theta)\frac{\dd\theta}{2\pi}
       \notag\\
 &\le\frac{e^2\pi^2D^2B}{8}(1-w_1)
       <\frac{64D^2B}{(N+1)^2}.\label{eq:window-integral}
\end{align}
The second inequality uses $\theta^2\le\pi^2(1-\cos\theta)/2$ on this interval, and the last follows from Eq.~\eqref{eq:w1}. Taking the supremum over the same input tuples and all finite sizes proves Eq.~\eqref{eq:window-bound}. The finite-dimensional normed polynomial space has a linear map $\mathcal W_v$ within $\rho_N<1$ of the identity, so the geometric-series bound for its inverse is $1/(1-\rho_N)$. Its explicit inverse is division of each degree by the nonzero rational $w_d$, proving Eq.~\eqref{eq:inverse-window}.
\end{proof}

The bound on $\mathcal W_v^{-1}$ controls the cost of compensation for the entire polynomial. Section~\ref{sec:coherent-use} chooses $Q$ and $\beta$ so that $q=\mathcal W_v^{-1}p/\beta$ has a certified strict margin. We next implement the weighted product with the original oracle; its value will remain $p(A)/\beta$.

\subsection{Unitary implementation with the supplied oracle}\label{sec:unitary-implementation}
We place the product in Eq.~\eqref{eq:logical-query-product} on the physical oracle and memory registers. We retain $\kappa=\min\{i,2\}$ identity labels and allocate
\begin{equation}\label{eq:R-general}
 R=\max\left\{m,\left\lceil\frac2\kappa\right\rceil\right\}
\end{equation}
memory labels. The inequality $R\ge m$ provides the active memory, while $\kappa R\ge2$ provides distinct pending and collected coordinates even when $i=1$. 

These choices give $d_-=\kappa+b$ retained input labels and $d_+=\kappa+a$ output labels, each with $R$ memory values. The zero--one matrices $J_-$ and $J_+$ have sizes $d_-R\times(bm+2)$ and $d_+R\times(am+2)$. They place the active block index and its multiplicity on the oracle and memory labels, and place the last two coordinates on flat idle indices zero and one. Appendix~\ref{app:circuit-packing} gives their entries and verifies the query identity in Eq.~\eqref{eq:packed-query}. The known matrices between queries are then
\begin{equation}\label{eq:Ct}
 \begin{aligned}
 C_0&=J_-M_0e_{am},\\
 C_t&=J_-M_tJ_+^\dagger\quad(1\le t<Q),\\
 C_Q&=e_{bm+1}^{\mathsf T}M_QJ_+^\dagger.
 \end{aligned}
\end{equation}
Their sizes are $d_-R\times1$, $d_-R\times d_+R$, and $1\times d_+R$, respectively. Since each packing matrix has orthonormal columns, the $C_t$ are contractions. Their selected product is Eq.~\eqref{eq:logical-query-product}, whose value is given by Theorem~\ref{thm:window-circuit}.

Each known update is a contraction, and each query contributes only a designated block of its unitary. The complementary amplitudes remain in the quantum state. We record their first exit in an extra register, so that the amplitude still carrying label zero follows exactly the product just calculated. This is the block-encoding product principle~\cite{gilyen2019}, implemented with one fresh label for each possible first exit.

For a general contraction product, let $1\le\kappa\le i$ be the number of retained identity labels. The retained query block is
\begin{equation}\label{eq:retained-block}
 L_\kappa(A)=
 \begin{pmatrix} I_{\kappa n}&0_{\kappa n\times bn}\\
             0_{an\times\kappa n}&\Delta(A)\end{pmatrix},
 \qquad d_-=\kappa+b,\quad d_+=\kappa+a.
\end{equation}
With $R$ memory labels, we write $\Lambda_R(A)$ for this block acting identically on each memory value. Its entries and the original physical label offsets are specified in Eqs.~\eqref{eq:query-entries} and~\eqref{eq:retained-labels}. Let $C_0$ be a contraction column of length $d_-R$, let $C_t$ have size $d_-R\times d_+R$ for $1\le t<Q$, and let $C_Q$ be a contraction row of length $d_+R$. The product to be implemented is
\begin{equation}\label{eq:contraction-product}
 (C_Q\otimes I_n)\Lambda_R(A)(C_{Q-1}\otimes I_n)\cdots
                    \Lambda_R(A)(C_0\otimes I_n).
\end{equation}
For the window construction, the packing in Eq.~\eqref{eq:packed-query} identifies this product with Eq.~\eqref{eq:logical-query-product}. The same implementation will also serve the row factors below.

We now construct a block encoding of each known auxiliary contraction $C_t$, so that its selected action becomes a gate in the quantum circuit. Because the entries of $C_t$ are computed classically, we can also compute its defect matrices and complete all columns of the gate. For a contraction $C$ with $h_2$ rows and $h_1$ columns, we use the explicit unitary
\begin{equation}\label{eq:Julia}
 J(C)=\begin{pmatrix}
 C&(I_{h_2}-CC^\dagger)^{1/2}\\
 (I_{h_1}-C^\dagger C)^{1/2}&-C^\dagger
 \end{pmatrix}.
\end{equation}
The row groups have lengths $h_2,h_1$, while the column groups have lengths $h_1,h_2$. The first $h_2$ output coordinates therefore contain $C$ applied to the selected $h_1$ input coordinates. This is a rectangular Julia completion, which in the square case is the Halmos dilation~\cite[Theorems~0--2]{robinson2018}. The defect-operator block encoding used by Low and Su~\cite[Section~2.4]{lowsu2026} differs only by a known sign on the complementary input columns. We use Eq.~\eqref{eq:Julia} because it supplies every column explicitly, accommodates rectangular updates, and is computable with the precision bounds established for our certificate factors.

It is useful to distinguish this single-step completion from a power dilation. A Sz.-Nagy unitary dilation of a square contraction preserves every nonnegative power after compression, and finite $N$-dilations preserve the first $N$ powers~\cite[Theorems~1.1 and~1.3]{levyshalit2014}. A Julia completion guarantees the chosen block for one use: for the scalar contraction $c=0$, $J(0)=\begin{psmallmatrix}0&1\\1&0\end{psmallmatrix}$, so $J(0)^2=I_2$ has selected entry one rather than $c^2=0$. The complementary amplitude has returned. Our circuit combines different, possibly rectangular contractions and intervening oracle blocks; recording each first exit will give their required product. After the row groups are embedded on a common padded register, the selected input columns form a Stinespring isometry~\cite{stinespring1955,watrous2018}. Its selected outcome implements $\rho\mapsto C\rho C^\dagger$, and the complementary outcome completes a trace-preserving operation on that register. We retain the outcome information until the final selection to obtain the coherent block product.

We store the processing label in $\mathsf M$, use one qubit $\mathsf d$ to distinguish this output from the Julia complement, and use a register $\mathsf E$ to record the first failed step. Their widths are
\[
 q_M=\lceil\log_2R\rceil,\qquad q_{\mathsf d}=1,\qquad
 q_E=\lceil\log_2(2Q+1)\rceil.
\]
All three registers and the supplied oracle register $\mathsf O$ start at zero; a zero-width memory is omitted. The input on $\mathsf S$ is arbitrary and may be entangled with an untouched reference. The event value zero denotes the continuing component, and values $1,\ldots,2Q$ record the possible first exits from it.

We embed each $J(C_t)$ on its specified input and output coordinate lists, obtaining a known full unitary $J_{C_t}$. The construction of these lists in Appendix~\ref{app:unitary-details} includes all reverse columns. The known gate $G_t$ applies $J_{C_t}$ on $\mathsf E=0$ and identity on the other event values. For $0\le t<Q$, $D_t$ exchanges event zero with $2t+1$ conditional on $\mathsf d=1$. For $1\le t\le Q$, $F_t$ exchanges zero with $2t$ conditional on $\mathsf d=0$ and an oracle output outside the retained list. Each marker fixes all other basis strings. We combine these known gates into
\[
 W_0=D_0G_0,\qquad W_t=D_tG_tF_t\ (1\le t<Q),\qquad W_Q=G_QF_Q.
\]
Using $\widehat U_A$ for $U_A$ on $\mathsf O,\mathsf S$ and identity on $\mathsf M,\mathsf d,\mathsf E$, the full circuit is
\begin{equation}\label{eq:full-circuit}
 \mathcal U_A=(W_Q\otimes I_n)\widehat U_A(W_{Q-1}\otimes I_n)
                \cdots (W_1\otimes I_n)\widehat U_A(W_0\otimes I_n).
\end{equation}
Every query acts on every event value, whereas the known Julia gates act trivially on values that already record a failure. The fresh labels keep those components separate from the continuing component, whose contribution is selected by requiring $\mathsf O,\mathsf M,\mathsf d,\mathsf E$ to be zero at the end. This also rejects the last Julia complement. The phase qubit and clean workspace for elementary-gate synthesis are counted in Section~\ref{sec:coherent-use}.

\begin{theorem}[Unitary implementation of a contraction product]\label{thm:unitary}
For the contractions and coordinate lists above, the all-zero ancillary block of Eq.~\eqref{eq:full-circuit} equals Eq.~\eqref{eq:contraction-product} for every allowed oracle completion. The circuit uses $Q$ original forward queries and $3Q+1$ known unitary primitives on $q_o+q_M+1+q_E$ non-signal qubits.
\end{theorem}
\begin{proof}
We place the known Julia gates on the designated coordinates and apply them only when the failure register is zero. Their action is identity on the other failure labels. After each of the first $Q$ Julia gates, a reversible marker swaps label zero with a new odd label on its complementary output. After each query, a second marker swaps zero with a new even label on outputs outside the retained oracle rows.

A marker is a permutation, so its reverse action must also be considered. Immediately before that marker, its new label is unoccupied: earlier markers used different labels, the known gates preserve all occupied nonzero labels, and the raw oracle acts as identity on the failure register. The reverse swap therefore acts on zero amplitude. Once a component is marked, subsequent markers use fresh labels and preserve it. Induction shows that the component with failure label zero follows precisely $C_0$, the retained query block, $C_1$, and so on. The final selection also checks the last dilation output. Thus its matrix is Eq.~\eqref{eq:contraction-product} for every oracle completion.

Counting the layers gives $Q+1$ known Julia gates, $Q$ markers for their earlier complements, and $Q$ query markers, for a total of $3Q+1$ known primitives. The circuit contains exactly $Q$ unconditional calls to the original unitary. The coordinate construction in Appendix~\ref{app:unitary-details} places all of them on the declared registers.
\end{proof}

For $Q=2$, Figure~\ref{fig:whole-circuit} places these gates on their physical registers. We denote the selected contraction product in Eq.~\eqref{eq:contraction-product} by $\mathcal K_2(A)$ in the diagram. The symbol $X_{0j}$ exchanges the two event values $0$ and $j$; an open control tests the entire zero string of $\mathsf E$. The gates $F_1,F_2$ perform the retained-output test defined above.

\begin{figure}[t]
\centering
\begin{quantikz}[row sep=.24cm,column sep=.24cm]
 \lstick{$\mathsf S:\ket\psi$} &\qw&\qw&\gate[2]{U_A}&\qw&\qw&\qw&\gate[2]{U_A}&\qw&\qw&\rstick{$\mathcal K_2(A)\ket\psi$}\qw\\
 \lstick{$\mathsf O:\ket0$} &\gate[3]{J_{C_0}}&\qw&\qw&\gate[4]{F_1}&\gate[3]{J_{C_1}}&\qw&\qw&\gate[4]{F_2}&\gate[3]{J_{C_2}}&\rstick{$\bra0$}\qw\\
 \lstick{$\mathsf M:\ket0$} &\qw&\qw&\qw&\qw&\qw&\qw&\qw&\qw&\qw&\rstick{$\bra0$}\qw\\
 \lstick{$\mathsf d:\ket0$} &\qw&\ctrl{1}&\qw&\qw&\qw&\ctrl{1}&\qw&\qw&\qw&\rstick{$\bra0$}\qw\\
 \lstick{$\mathsf E:\ket0$} &\octrl{-3}&\gate{X_{01}}&\qw&\qw&\octrl{-3}&\gate{X_{03}}&\qw&\qw&\octrl{-3}&\rstick{$\bra0$}\qw
\end{quantikz}
\caption{The complete two-query matrix-gate circuit. The four ancillary bras specify a joint final zero selection; the signal expression is its unnormalized output. The $\mathsf d$-controlled transpositions are $D_0,D_1$, and the open controls give $G_0,G_1,G_2$. Both oracle calls act on every failure component. When the $C_t$ are the synthesized target parameters, $\mathcal K_2(A)$ is the corresponding polynomial block. Elementary-gate synthesis adds the phase and clean work qubits counted in Theorem~\ref{thm:gates}.}
\label{fig:whole-circuit}
\end{figure}
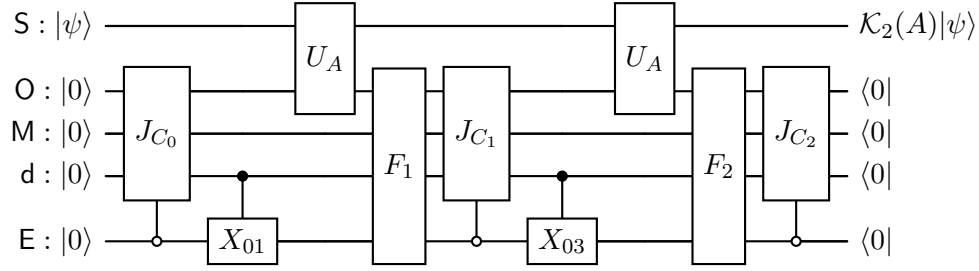

The numerical example in Eq.~\eqref{eq:example-transfer-blocks} now gives a complete instance of this construction. Its diagonal input has $i=1$, $a=b=2$ and a three-qubit oracle register. With three equal stage weights and $Q=2$, the realization of $q=(1+x_1x_2)/4$ produces $p_0(A)=I_n/4+A_1A_2/12$ at normalization one. Its multiplicity is $m=R=4$, so the known factors have sizes $12\times1$, $12\times12$ and $1\times12$. Selecting zero on the three oracle qubits, two memory qubits, dilation qubit and three event qubits gives exactly $p_0(A)$ in the ideal circuit. Appendix~\ref{app:worked-scalar} supplies the full coordinates and elementary-gate resource bound. This example fixes a convenient normalization to display the computation; the algorithm below also searches for a normalization close to the target norm.

Applying Theorem~\ref{thm:unitary} to the matrices in Eq.~\eqref{eq:Ct} places the output of Theorem~\ref{thm:window-circuit} in the all-zero ancillary block of a $Q$-query circuit. Its processing multiplicity is $R$ from Eq.~\eqref{eq:R-general}. For a fixed realization, this uses one pending and one collected coordinate independently of $Q$; the failure register and the list of known gates retain their stated dependence on the sequence length.

\subsection{The synthesis algorithm, resources, and error}\label{sec:coherent-use}
Algorithm~\ref{alg:synthesis} combines the classical routines with the query construction. Its only additional coefficient operation is degree reweighting. After obtaining $B\le\widehat B\le(1+\tau/8)B$, we use $\beta=(1+\tau/2)\widehat B$ and choose the least $N>D$ with $(N+1)^2\tau\ge1024D^2$. The weights of Lemma~\ref{lem:weights} then define $Q=N-1$ and an auxiliary target whose strict norm bound is proved below.

We represent the reweighted polynomial by adding a degree coordinate to the supplied recurrence. In block row and column indices $0,\ldots,D$, each of length $s$, its transitions and endpoints are
\begin{equation}\label{eq:degree-clock}
 \widetilde T_j=
 \begin{pmatrix}
 0&T_j&0&\cdots&0\\
 0&0&T_j&\ddots&\vdots\\
 \vdots&&\ddots&\ddots&0\\
 0&\cdots&0&0&T_j\\
 0&\cdots&\cdots&0&0
 \end{pmatrix},\quad
 \widetilde\ell=[\ell\ 0\ \cdots\ 0],\quad
 \widetilde r=
 \begin{pmatrix}r_0/(\beta w_0)\\r_0/(\beta w_1)\\\vdots\\r_0/(\beta w_D)\end{pmatrix}.
\end{equation}
The transition size is $(D+1)s\times (D+1)s$. Multiplication by the initial row advances along the superdiagonal, so
$\widetilde\ell\widetilde T_w\widetilde r=c_w/(\beta w_{|w|})$ for $|w|\le D$, and all longer coefficients vanish. We call the resulting polynomial $q$. Its degree index belongs to the classical coefficient representation, whose matrices we use to compute the circuit parameters.

Algorithm~\ref{alg:synthesis} applies the certificate and contraction routines to this compensated recurrence, then forms the known updates and their full unitary implementation. Only after these exact objects have been specified do we compile the known gates to the requested error.

\begin{algorithm}[htbp]
\caption{Classical synthesis for a scalar target under a general joint input}
\label{alg:synthesis}
\begin{algorithmic}[1]
\Require Rational finite-state coefficients $(\ell,T_1,\ldots,T_k,r_0)$, joint-oracle metadata, and rational $0<\tau,\epsilon<1$.
\Ensure A positive rational normalization $\beta$ and known layers of a forward-query circuit with selected-block error at most $\epsilon$.
\State We reduce the coefficient representation and test polynomiality as in Appendix~\ref{app:coefficient-inputs}.
\If{the representation is not polynomial}
  \State \Return rejection of the classical input promise
\ElsIf{the target is zero or constant}
  \State \Return the corresponding zero-query construction
\EndIf
\State We determine its true degree $D$ by Eq.~\eqref{eq:degree-test}.
\State We compute $B\le\widehat B\le(1+\tau/8)B$ by Theorem~\ref{thm:norm-search}; we set $\beta=(1+\tau/2)\widehat B$.
\State We choose the least $N>D$ with $(N+1)^2\tau\ge1024D^2$; we set $Q=N-1$.
\State We compute the rational vector $r$ and degree weights $w_d$ of Lemma~\ref{lem:weights}, then the representation of $q$ in Eq.~\eqref{eq:degree-clock}.
\State We compute residual coordinates for $q$ and its strict certificate by Proposition~\ref{prop:strict}, with $\gamma=1$ and $\nu=1-\tau/4$.
\State We compute $X,Y,T_0,V$ by Lemma~\ref{lem:extract}, then $M_t$ and $C_t$ by Eqs.~\eqref{eq:Mt} and~\eqref{eq:Ct}.
\State We form the complete known layers by Theorem~\ref{thm:unitary} and compile them with total operator-norm error at most $\epsilon$ by Theorem~\ref{thm:gates}.
\State \Return $\beta$ and the known-layer descriptions, interleaved with $Q$ calls to the supplied oracle
\end{algorithmic}
\end{algorithm}

\begin{proof}[Proof of the general-input assertion of Theorem~\ref{cor:scalar-synthesis}]
For $D\ge1$, Theorem~\ref{thm:norm-search} gives the stated rational norm estimate in polynomial bit time. Its normalization satisfies
\[
 B<\beta\le(1+\tau/2)(1+\tau/8)B\le(1+\tau)B.
\]
The chosen $N$ is $O(D/\sqrt\tau)$ and gives $\rho_N\le\tau/16$. Lemma~\ref{lem:weights} therefore proves the strict promise
\begin{equation}\label{eq:strict-q}
 \|q\|_{\SA,\Delta}
 \le\frac1{(1+\tau/2)(1-\tau/16)}\le1-\tau/4.
\end{equation}
For $0<\tau\le1$, the last inequality follows after multiplication by its positive denominator from $\tau(24-18\tau+\tau^2)\ge0$. This establishes the rational margin required by the strict feasibility subroutine.

The reweighted representation has at most $(D+1)s$ states, so its residual construction has $r\le(D+1)s+1$. Lemma~\ref{lem:extract} permits $m=br$ and provides polynomial-time certified entry evaluations of $V$. In particular, it proves Eq.~\eqref{eq:coefficient-factorization} for every word, including every coefficient above degree $D$. Theorem~\ref{thm:window-circuit} consequently implements
\[
 \sum_{d=0}^D w_d\sum_{|w|=d}\frac{c_w}{\beta w_d}A_w
             =p(A)/\beta
\]
exactly before the known gates are approximated. The selected output therefore acts on the original signal register.

Each classical matrix in Eq.~\eqref{eq:degree-clock} has polynomial dimension and rational bit length. The norm algorithm, strict search, and matrix extraction have polynomial bit complexity in those data and $\log(1/\tau)$. Computing the $Q+1$ updates and compiling them is polynomial also in $Q$. Theorem~\ref{thm:unitary} gives the full unitary circuit; Theorem~\ref{thm:gates} in Appendix~\ref{sec:finite-gates} supplies its known-gate count and uniform error. The ideal registers use $q_o+\lceil\log_2R\rceil+1+\lceil\log_2(2Q+1)\rceil$ qubits, and elementary-gate compilation adds at most a constant multiple of this width. All known gate parameters are computed from the target description and the joint-encoding data. If the oracle is implemented inexactly, Corollary~\ref{cor:oracle-error} adds its designated-block error separately.

For a scalar constant $c\ne0$, rational bisection of $|c|$ gives a rational $\beta$ in the required interval. The known $2\times2$ Julia matrix of the scalar $c/\beta$, acting on one ancillary qubit, has the desired selected block and requires no oracle call. For the zero polynomial, a Pauli $X$ on a selected ancillary qubit gives a zero selected block; any fixed positive normalization may be reported. These cases are separated before the residual construction, completing the general assertion.
\end{proof}

The register and supported-gate constructions give the ancillary bound
\begin{equation}\label{eq:main-memory}
 q_{\rm anc}=O\bigl(q_o+\log((Q+1)R)\bigr),\qquad
 R=\max\left\{m,\left\lceil2/i\right\rceil\right\},\qquad
 m\le b\bigl((D+1)s+1\bigr).
\end{equation}
Here $q_{\rm anc}$ includes the supplied oracle register and the clean workspace for elementary gates; $R$ counts processing labels. Addressing a physical oracle label costs polynomially in its $q_o$-bit width, rather than in $2^{q_o}$. Theorem~\ref{thm:gates} gives the support-dependent gate count. Its full-unitary error telescope also permits errors in the designated query block, as follows.

\begin{corollary}[Error in the encoded block]\label{cor:oracle-error}
If the $t$th implemented unitary query has a designated submatrix differing from Eq.~\eqref{eq:visible-block} by at most $\eta_t$, then the circuit with exact known gates differs from Eq.~\eqref{eq:contraction-product} by at most $\sum_{t=1}^Q\eta_t$. Known-gate synthesis adds $\epsilon_g$, so a uniform bound $\eta_t\le\eta_o$ gives total error at most $Q\eta_o+\epsilon_g$.
\end{corollary}
\begin{proof}
The markers in Theorem~\ref{thm:unitary} test membership in the designated coordinates. They therefore implement the product of the actual compressed query matrices, each of which is a contraction. A product telescope bounds the change caused by each designated submatrix by $\eta_t$. The same full-unitary telescope used in Theorem~\ref{thm:gates}, now with the actual queries, pays for the known-gate error. The retained corner is a further input and output compression of the
full designated corner in Eq.~\eqref{eq:visible-block}, so its error
is at most $\eta_t$. The estimate permits arbitrary compatible completions and includes the same stated error in the designated identity block. The separate constant and zero
constructions use no query and hence incur only their known-gate error.
\end{proof}

For a scalar circuit with total ancillary width $q_{\rm anc}$, the embedding $J_0=\ket{0^{q_{\rm anc}}}\otimes I_n$ specifies the selected input and output subspace. An exact block $M=J_0^\dagger\mathcal U_AJ_0=p(A)/\beta$ succeeds on a unit input $\psi$ with probability
\begin{equation}\label{eq:normalization-success}
 p_{\rm succ}(\psi)=\|p(A)\psi\|^2/\beta^2.
\end{equation}
At $\beta=1$, success is determined entirely by the norm of the target output on the chosen input. For example, the norm-one row polynomial in Example~\ref{ex:row} vanishes when both input matrices vanish, and its success probability is then zero. A unitary target behaves differently: because it preserves every input norm, accuracy of its selected block also controls the complete output state. The next lemma derives this relation directly from the columns of the circuit unitary.

\Needspace{9\baselineskip}
\begin{lemma}[Coherent use of a unitary target]\label{lem:coherent-output}
Let $\mathcal U$ be a unitary of order $2^{q_{\rm anc}}n$, with selected block $M=J_0^\dagger\mathcal UJ_0$, and let $V_{\rm tar}$ be an $n\times n$ unitary. If $\|M-V_{\rm tar}\|\le\delta\le1$, then
\begin{equation}\label{eq:coherent-output-bound}
 \|\mathcal UJ_0-J_0V_{\rm tar}\|\le\sqrt{2\delta},\qquad
 p_{\rm succ}(\psi)\ge(1-\delta)^2
\end{equation}
for every unit input, including any entangled reference. At $\delta=0$, the ancillary registers return to zero and no postselection is required.
\end{lemma}
\begin{proof}
Both $\mathcal UJ_0$ and $J_0V_{\rm tar}$ have orthonormal columns. Their difference $E_{\rm col}$ therefore satisfies
\[
 E_{\rm col}^\dagger E_{\rm col}
   =2I_n-M^\dagger V_{\rm tar}-V_{\rm tar}^\dagger M
   =-(M-V_{\rm tar})^\dagger V_{\rm tar}
       -V_{\rm tar}^\dagger(M-V_{\rm tar}),
\]
whose norm is at most $2\delta$. Also $\|M\psi\|\ge1-\delta$, which proves the success bound. The same identities hold after tensoring with the reference identity. At zero error the first norm vanishes, giving the asserted full output state.
\end{proof}

For Hamiltonian simulation, Martyn et al. formulate full coherence as accurate evolution with success approaching one on a single copy of the input~\cite[Definition~2, Section~II.3]{martyn2023coherent}. The same operational interpretation applies here through Eq.~\eqref{eq:coherent-output-bound}. If $B=1$ and $p(A)$ is unitary on the application inputs, the near-optimal normalization gives block distance at most $1-1/\beta+\epsilon$ from $p(A)$, which tends to zero with $\tau,\epsilon$. For the general input model, attaining this limit carries the stated $\tau^{-1/2}$ query dependence. A supplied exact boundary circuit attains the zero-error case directly. Appendix~\ref{sec:applications} treats a nonunitary target, where success is controlled by a separate lower singular-value estimate.

The polynomial circuit can also be applied to an effectively computed approximation of a matrix function. Appendix~\ref{sec:applications} gives a compact noncommutative example, including its truncation and success bounds. We turn next to the stronger factorization available for a row input.

\subsection{Synthesis with one row query per degree}\label{sec:exact}\label{sec:row-complement-discovery}\label{sec:row}
A row input $\Delta(A)=[A_1\ \cdots\ A_k]$ gives us more control over the combination of terms: one query applies the different letters on different input labels and adds their contributions in a common output label. We exploit this structure to replace the weighted sequence by a factorization with exactly $D$ queries, independently of the normalization margin. The factorization must preserve the joint norm as each degree is removed.

The certificate supplies the needed invariant by completing $p/\beta$ to a column of polynomials. Its coefficient correlations describe an isometry on word coefficients, as in the row-isometric models of Popescu~\cite{popescu1989,popescu1991}. Those correlations make the highest-degree output range orthogonal to the constant output. We can therefore separate these ranges, remove one initial letter from the former, and reconstruct the original column with one row query. We first prove this step, then compute its ranges from the finite-state description.

A complementary list for $p/\beta$ consists of scalar polynomials $\theta_1,\ldots,\theta_q$, with $\theta_1=p/\beta$. For each word $w$, their scalar coefficients form the numerical column
\begin{equation}\label{eq:theta-coeff}
 \theta_w=\begin{pmatrix}\theta_{1,w}\\\vdots\\\theta_{q,w}\end{pmatrix}.
\end{equation}
The required norm-preserving condition is the following set of finite coefficient correlations:
\begin{equation}\label{eq:inner-correlations}
 \Gamma_\emptyset=\sum_w\theta_w^\dagger\theta_w=1,
 \qquad
 \Gamma_u=\sum_w\theta_w^\dagger\theta_{wu}=0
                    \quad(u\ne\emptyset).
\end{equation}
We call such a finite list \emph{inner}, following the isometric inner-multiplier condition on free words~\cite{popescu1995}. The first identity fixes the total coefficient norm, and the others cancel overlaps between coefficients at words related by an appended suffix. The convolution matrix below makes this norm-preserving condition precise.

The correlation identities have a concrete matrix interpretation that will also explain how the SA certificate supplies them. If the degree is at most $d$, we give its convolution matrix columns indexed by words of length at most $d$ and rows indexed by a component $\ell$ and a word of length at most $2d$. At row $(\ell,v)$ and column $w$, its entry is $\theta_{\ell,u}$ when $v=uw$, and is zero if no such $u$ exists. The Gram of this matrix is identity exactly when Eq.~\eqref{eq:inner-correlations} holds: two suffix-incomparable column words have disjoint row supports, while for suffix-comparable words their inner product is $\Gamma_u$ or its conjugate. Every correlation is tested against the empty column word. The isometry here is a property of coefficient convolution; evaluation on an allowed tuple gives a contractive polynomial column. We will verify the correlations through state Grams, keeping the calculation in the compact representation.

For the row block, the target-derived certificate supplies a complementary list directly. We use its normalized form
\begin{equation}\label{eq:normalized-row-certificate}
 E^\dagger E-P^\dagger P/\beta^2=S+\Phi(T),
 \qquad S,T\succeq0,
\end{equation}
where $P$ still denotes the coordinates of the unnormalized polynomial $p$. We compute $S,T$ for the normalized target $p/\beta$ in this subsection. A full square factor $G=S^{1/2}$ gives the scalar polynomials
\[
 \theta_1(x)=p(x)/\beta,\qquad
 \theta_{1+\ell}(x)=\sum_t G_{\ell t}g_t(x).
\]
We show that these coefficients satisfy Eq.~\eqref{eq:inner-correlations} using a finite matrix argument. We use the explicit word matrices $S_j^{(L)}$ from the proof of Theorem~\ref{thm:certificate}, now with cutoff $L=2D$ and scale one. Their combined row is a contraction, and
$(S_j^{(L)})^\dagger S_\ell^{(L)}=\delta_{j\ell}P_{<L}$, where $P_{<L}$ is the diagonal matrix selecting words shorter than $L$. Each auxiliary $h_{\alpha,\ell}$ in Eq.~\eqref{eq:localizing-sum} has degree at most $D-1$. On an input word of length at most $D$, its output therefore has length at most $2D-1$, and the query defect in Eq.~\eqref{eq:localizing-sum} vanishes on that output. The certificate consequently says that the finite convolution matrix of the complementary list, restricted to these input columns, has Gram identity. The equivalence following Eq.~\eqref{eq:inner-correlations} proves the desired correlations. The compact Gram calculation in Appendix~\ref{app:row-compact} computes the same coefficient identities from the state recurrence.

We now factor this list so that each factor can be implemented by one row query. The highest-degree coefficient columns identify the part of the output that must contain at least one more input matrix. Projecting onto that part and removing an initial letter lowers the degree; the complementary part remains on the identity coordinate. The following update carries out this separation while preserving the coefficient correlations.

For a current complementary list of degree $d\ge1$, we choose a numerical $q\times e$ matrix $Z$ with orthonormal columns spanning its degree-$d$ coefficient columns, and put $\Pi=ZZ^\dagger$. The proposed reduced coefficients are
\begin{equation}\label{eq:peel-coeff}
 \psi_w=
 \begin{pmatrix}
 (I_q-\Pi)\theta_w\\
 Z^\dagger\theta_{1w}\\
 \vdots\\
 Z^\dagger\theta_{kw}
 \end{pmatrix},
 \qquad \psi_w\text{ has }q+ke\text{ entries}.
\end{equation}
To express their reconstruction, we use the $q\times(q+ke)$ numerical matrices
\[
 B_0=[I_q-\Pi\ \ 0\ \cdots\ 0],\qquad
 B_j=[0\ \cdots\ Z\ \cdots\ 0],
\]
where $Z$ occupies the $j$th $e$-column block after the first $q$ columns. This is a finite-degree factorization of an inner free multiplier, in the setting described by Popescu~\cite{popescu1995}. We prove the coefficient update below and then identify the corresponding input and output matrices for a single row query.

\begin{lemma}[One degree-reducing coefficient step]\label{lem:peel}
If the degree-$d$ list satisfies Eq.~\eqref{eq:inner-correlations}, the coefficients in Eq.~\eqref{eq:peel-coeff} satisfy the same correlations, have degree at most $d-1$, and obey
\begin{equation}\label{eq:peel-recover}
 \theta_\emptyset=B_0\psi_\emptyset,\qquad
 \theta_{jw}=B_0\psi_{jw}+B_j\psi_w.
\end{equation}
\end{lemma}
\begin{proof}
For a word $u$ of length $d$, the only nonzero term in $\Gamma_u$ can be $\theta_\emptyset^\dagger\theta_u$, so $\Pi\theta_\emptyset=0$. The first block of Eq.~\eqref{eq:peel-coeff} removes the entire degree-$d$ coefficient range, and each later block takes a left residual. Hence the new list has degree at most $d-1$. For the empty coefficient, $B_0\psi_\emptyset=(I-\Pi)\theta_\emptyset=\theta_\emptyset$. For $jw$, its two contributions are $(I-\Pi)\theta_{jw}$ and $\Pi\theta_{jw}$, proving Eq.~\eqref{eq:peel-recover}.

The inner correlations are preserved by a direct coefficient calculation. For every word $u$, including the empty word, we have
\begin{align*}
 \sum_w\psi_w^\dagger\psi_{wu}
 &=\sum_w\theta_w^\dagger(I-\Pi)\theta_{wu}
      +\sum_{j,w}\theta_{jw}^\dagger \Pi\theta_{jwu}\\
 &=\sum_w\theta_w^\dagger\theta_{wu}
                       -\theta_\emptyset^\dagger \Pi\theta_u
 =\Gamma_u.
\end{align*}
The second sum in the first line covers each nonempty prefix word exactly once, and the last equality uses $\Pi\theta_\emptyset=0$. The updated coefficients therefore satisfy the same correlation identities, allowing us to repeat the reduction until a constant numerical column of norm one remains.
\end{proof}

A mixed-degree example illustrates the separation. For two variables we take
\begin{equation}\label{eq:row-mixed-complement}
 \theta(x)=\frac12\begin{pmatrix}1+x_1x_2\\1-x_1x_2\end{pmatrix},\qquad
 Z=\frac1{\sqrt2}\begin{pmatrix}1\\-1\end{pmatrix},\qquad
 \Pi=\frac12\begin{pmatrix}1&-1\\-1&1\end{pmatrix}.
\end{equation}
Its constant and quadratic coefficient columns are orthogonal and each has squared norm $1/2$, which verifies the inner correlations. The degree-reducing step gives
\[
 \psi(x)=\begin{pmatrix}1/2\\1/2\\x_2/\sqrt2\\0\end{pmatrix}.
\]
The first two entries carry the constant contribution through the identity label. The third acquires $x_1$ in the active row query and is sent by $Z$ to the opposite signs of the two outputs. Thus the same factor reconstructs both the sum and difference in Eq.~\eqref{eq:row-mixed-complement}. A second reduction removes $x_2$, giving a two-query realization. In the circuit, the last factor found is applied first: the active contributions acquire $A_2$ and then $A_1$, yielding the product $A_1A_2$. At $A_1=A_2=0$ the evaluated column has norm $1/\sqrt2$, illustrating why the invariant concerns coefficient convolution rather than the norm of each matrix evaluation.

The coefficient recovery now identifies the matrix that one query must implement. Writing $\theta(A)=\sum_w\theta_w\otimes A_w$ and $\psi(A)=\sum_w\psi_w\otimes A_w$ for the two polynomial columns, we have
\begin{equation}\label{eq:row-one-query-recovery}
 \theta(A)=\left(B_0\otimes I_n+\sum_{j=1}^kB_j\otimes A_j\right)\psi(A).
\end{equation}
We route the first $q$ coordinates through the identity label and the $j$th residual block through active input $j$, after applying its frame $Z$. The active contributions reach one common output label, where they lie in the range of $\Pi$; the identity contribution lies in its orthogonal complement.

We use the first physical identity label and all $k$ active input labels, so $\kappa=1$, $d_-=k+1$, and $d_+=2$ in Section~\ref{sec:unitary-implementation}. The active labels remain at their original offsets $i,\ldots,i+k-1$. For any memory size $R_*\ge q$, the preparation and collection matrices described in Appendix~\ref{app:row-packing} have sizes
\[
 E_Z:\ (k+1)R_*\times(q+ke),\qquad R_Z:\ q\times2R_*.
\]
Their column and row Grams are $E_Z^\dagger E_Z=\operatorname{diag}(I_q-\Pi,I_{ke})$ and $R_ZR_Z^\dagger=I_q$. The block entries in Eq.~\eqref{eq:row-packing} give
\begin{equation}\label{eq:row-one-query-identity}
 B_0\otimes I_n+\sum_{j=1}^kB_j\otimes A_j
       =(R_Z\otimes I_n)\Lambda_{R_*}(A)(E_Z\otimes I_n).
\end{equation}
Both endpoint matrices are contractions, so Theorem~\ref{thm:unitary} realizes the recovery step with one call to the original row oracle. Repeated reduction leaves a constant unit column after at most $D$ steps. Starting with that column, the circuit applies the factors in reverse reduction order and selects the original first polynomial at the end; Eq.~\eqref{eq:row-query-matrices} records the chronological matrices. It remains to compute the frames $Z$ efficiently and bound a common processing size $R_*$.

We next compute the highest-degree coefficient range from the compact description. This range is all that the coefficient step needs from the degree-$d$ terms; listing those terms individually is unnecessary. We keep the transition coordinates fixed and identify the output directions removed at each step through a Gram matrix of their future responses. We assume that the current input list is specified by rational matrices $N_j$ of order $s_0$, rational columns $b_j$ of length $s_0$, a $q_0\times s_0$ output matrix $O$, and a column $d_0$ of length $q_0$:
\begin{equation}\label{eq:row-state-coeff}
 \theta_\emptyset=d_0,\qquad
 \theta_{j_1\cdots j_t}=ON_{j_1}\cdots N_{j_{t-1}}b_{j_t}
 \quad(t\ge1).
\end{equation}
For a word $w=j_1\cdots j_t$, we write $N_w=N_{j_1}\cdots N_{j_t}$ and $N_\emptyset=I_{s_0}$. We require entry evaluators for $O,d_0$ and their exact rational joint Gram $[O\ d_0]^\dagger[O\ d_0]$. This permits exact support decisions even when the output entries contain square roots. Rational reachability and observability reduction, detailed in Appendix~\ref{app:arithmetic}, makes the columns $N_wb_j$ span the state space and removes the common kernel of all $ON_w$. For a reduced degree-$d$ polynomial list, the transitions are then jointly nilpotent of index $d$.

For a state-coordinate column $x$, the total squared output response is $\sum_w\|ON_wx\|^2$. We collect this quantity in the observability Gram
\begin{equation}\label{eq:row-Gram}
 H=O^\dagger O+\sum_jN_j^\dagger HN_j
   =\sum_{|w|<d}N_w^\dagger O^\dagger ON_w\succ0.
\end{equation}
Iteration from zero computes this rational $s_0\times s_0$ matrix in at most $d$ updates. We also compute rational bases $U_\ell$ for
$\operatorname{span}\{N_wx:|w|=\ell-1\}$, with $U_1=I_{s_0}$ and $U_{d+1}$ empty. The ranges are nested because an additional transition reduces the set of reachable directions, and nilpotence makes the last range zero. At stage $t$, the directions in $U_{d-t+1}$ have already been removed from the output; the next coefficient range comes from $U_{d-t}$. We choose $R_t$ to complete the former basis to the latter and write $e_t$ for its number of columns. To evaluate the response of a remaining direction independently of representatives that differ by an eliminated one, we minimize the original quadratic form:
\[
 x^\dagger H_tx=\min_z(x-U_{d-t+1}z)^\dagger H(x-U_{d-t+1}z).
\]
Completing the square gives the rational shorted matrix~\cite{andersontrapp1975}
\begin{equation}\label{eq:short-Gram}
 H_t=H-HU_{d-t+1}
       (U_{d-t+1}^\dagger HU_{d-t+1})^{-1}U_{d-t+1}^\dagger H,
\end{equation}
where the subtracted term is zero if $U_{d-t+1}$ has no columns. In particular, $H_t\succeq0$ and its kernel is $\operatorname{ran}U_{d-t+1}$. Starting with $O_0=O$ and the supplied $d_0$, we translate Eq.~\eqref{eq:peel-coeff} into the following updates for $0\le t<d$:
\begin{equation}\label{eq:row-updates}
 \begin{aligned}
 K_t&=R_t^\dagger H_tR_t,\qquad
 Z_t=O_tR_tK_t^{-1/2},\qquad \Pi_t=Z_tZ_t^\dagger,\\
 O_{t+1}&=\begin{pmatrix}
 (I-\Pi_t)O_t\\ Z_t^\dagger O_tN_1\\\vdots\\Z_t^\dagger O_tN_k
 \end{pmatrix},\qquad
 d_{t+1}=\begin{pmatrix}
 (I-\Pi_t)d_t\\ Z_t^\dagger O_tb_1\\\vdots\\Z_t^\dagger O_tb_k
 \end{pmatrix}.
 \end{aligned}
\end{equation}
The transition matrices $N_j$ and final-letter columns $b_j$ stay fixed. The first block retains the part of the current output orthogonal to the range of $\Pi_t$, while each later block removes an initial letter by multiplying $O_t$ by its transition or final-letter column. Thus the decrease in polynomial degree is carried by the changing output pair. If $O_t$ has $q_t$ rows, then $Z_t$ has $q_t$ rows and $e_t$ columns, and both updated outputs have $q_{t+1}=q_t+ke_t$ rows. For $e_t=0$, the inverse-root and appended variable blocks are absent. Otherwise the next proof establishes $K_t\succ0$ and identifies $H_t$ with the current response Gram; these facts justify both the frame and the use of the original rational data to compute it.

\begin{theorem}[Efficient synthesis from a complementary list]\label{thm:row-compiler}
A degree-$d$ inner list, with $d\ge1$, satisfying the rational-Gram input assumptions of Eq.~\eqref{eq:row-state-coeff} can be synthesized with $d$ forward row queries and processing multiplicity $R_*=q_0+ks_0$. Its selected block is the first polynomial of the list. The preparation and collection matrices determined by Eqs.~\eqref{eq:short-Gram} and~\eqref{eq:row-updates} have polynomial-time entry evaluators; finite-gate synthesis for rational $0<\epsilon<1$ is polynomial in the supplied description, oracle width, and $\log(1/\epsilon)$.
\end{theorem}
\begin{proof}
We prove that the remaining metric is the Gram of the current output responses:
\begin{equation}\label{eq:metric-invariant}
 H_t=\sum_wN_w^\dagger O_t^\dagger O_tN_w.
\end{equation}
At $t=0$ this is Eq.~\eqref{eq:row-Gram}. Assuming it at stage $t$, every $N_jR_t$ lies in $\operatorname{ran}U_{d-t+1}=\ker H_t$. The positive terms of the sum then give
\[
 H_tR_t=O_t^\dagger O_tR_t,
 \qquad K_t=(O_tR_t)^\dagger(O_tR_t)\succ0.
\]
The last inequality follows because $R_t$ is independent modulo the kernel, so $Z_t$ has orthonormal columns.

Its range is precisely the highest-degree coefficient range. Indeed, each state column is a linear combination of $N_vb_j$ by reachability. In $O_tN_wx$ with $|w|=d-t-1$, every term with nonempty $v$ lies in the already removed range and vanishes. The remaining columns $O_tN_wb_j$ are exactly the coefficients of degree $d-t$. Thus $O_tU_{d-t}$ spans those coefficients, and $O_tR_t$ spans the same space because $O_t$ kills $U_{d-t+1}$. Lemma~\ref{lem:peel} consequently applies to the computed frame.

To propagate the invariant, we express the removed output energy as $B_t=O_t^\dagger \Pi_tO_t$; the coefficient update then gives
\[
 O_{t+1}^\dagger O_{t+1}
 =O_t^\dagger O_t-B_t+\sum_jN_j^\dagger B_tN_j.
\]
Summing its congruences by all $N_w$ cancels every nonempty-word contribution of $B_t$, leaving $H_t-B_t$. The identity $H_tR_t=O_t^\dagger O_tR_t$ also gives
\[
 B_t=H_tR_t(R_t^\dagger H_tR_t)^{-1}R_t^\dagger H_t.
\]
Subtracting this matrix minimizes the current quadratic form over the additional directions $R_t$. Successive minimization over the nested spaces therefore equals minimization over $U_{d-t}$, which is the definition of $H_{t+1}$. This proves Eq.~\eqref{eq:metric-invariant} at the next stage.

After $d$ stages, $H_d=0$, hence $O_d=0$ and innerness gives $d_d^\dagger d_d=1$. The removed ranks telescope:
\[
 \sum_{t=0}^{d-1}e_t=s_0,\qquad q_t=q_0+k\sum_{j<t}e_j\le q_0+ks_0.
\]
A single memory of this size accommodates every row factor. Applying those factors in reverse reduction order gives the $d$-query circuit in Eq.~\eqref{eq:row-query-matrices}.

All bases and metrics are computed from the original rational state data. Appendix~\ref{app:row-compact} proves polynomial bit bounds for their entries, the positive eigenvalues of $K_t$, and the guarded evaluations of $O_t,d_t,Z_t$. Theorem~\ref{thm:gates} then supplies the finite gate list. Thus neither the coefficient spans nor the output factors require a word table.
\end{proof}

We apply this theorem to the complement already obtained from the certificate. The numerical matrix $T_j^g$ of order $h$ has its block row $(j,1),\ldots,(j,r)$ equal to $J$ and every other row zero. With $e_0$ the first coordinate column of length $h$, the coefficients of the list $g$ are $T_w^ge_0$. This follows directly from Eq.~\eqref{eq:g-list} and the residual identities. The matrices $T_j^g$ are jointly nilpotent. Writing $P$ for the target coordinate row in Eq.~\eqref{eq:EJP}, the numerical output matrix for the complementary coefficients is
\begin{equation}\label{eq:complement-output}
 O_\Theta=\begin{pmatrix}P/\beta\\S^{1/2}\end{pmatrix},
 \qquad O_\Theta^\dagger O_\Theta=P^\dagger P/\beta^2+S.
\end{equation}
Thus in Eq.~\eqref{eq:row-state-coeff} we may use $N_j=T_j^g$, $b_j=T_j^ge_0$, $O=O_\Theta$, and $d_0=O_\Theta e_0$. Their joint output Gram is rational, and their entries have the square-root evaluators already proved. The reachable nonconstant states lie in the last $kr$ coordinates, so rational reduction gives $s_0\le kr$. The number of initial complementary rows is at most $q_0=1+h$.

For the complement obtained from a target certificate, the bounds $q_0\le1+h$ and $s_0\le kr$ give the processing multiplicity
\begin{equation}\label{eq:row-memory}
 R_*=1+h+k^2r,\qquad h=1+kr,\qquad r\le s+1.
\end{equation}

\begin{corollary}[Row synthesis from the coefficients]\label{cor:row-main}
For a rational finite-state scalar target of degree $D\ge1$ on the row domain and rational $0<\tau,\epsilon<1$, the algorithm computes $B<\beta\le(1+\tau)B$ and a $D$-query circuit for $p(A)/\beta$. Its processing multiplicity is given by Eq.~\eqref{eq:row-memory}, and classical and finite-gate synthesis to error $\epsilon$ are polynomial in the input length, oracle width, $\log(1/\tau)$, and $\log(1/\epsilon)$. The selected block is exact with ideal known gates.
\end{corollary}
\begin{proof}
We compute $B\le\widehat B\le(1+\tau/8)B$ using
Theorem~\ref{thm:norm-search}, and set
$\beta=(1+\tau/2)\widehat B$. Then
$B<\beta\le(1+\tau)B$, and $B/\beta\le\nu$ for the rational
$\nu=(1+\tau/2)^{-1}<1$. We apply Proposition~\ref{prop:strict} to the original target $p$ with $\gamma=\beta$ and $\nu=(1+\tau/2)^{-1}$. This returns a certificate of $1-p^*p/\beta^2$. The preceding construction supplies a complementary list of degree at most $D$ in the state representation required by Theorem~\ref{thm:row-compiler}, whose joint output Gram is rational. Its memory bound is $q_0+ks_0\le1+h+k^2r$. The first polynomial has degree $D$, so the list also has degree exactly $D$ and the circuit uses that many queries. The strict-search dependence is polynomial in $\log(1/\tau)$, and the factorization uses $D$ queries at every such margin. This proves the row assertion of Theorem~\ref{cor:scalar-synthesis}. The exact degree lower bound is established in Corollary~\ref{cor:degree-lower} for the same completion-independent forward model.
\end{proof}

The normalization benefit is visible in the two-query example
\[
 p(A)=\tfrac35A_1A_2+\tfrac45A_2A_1,\qquad
 \|p\|_{\SA,\mathrm{row}}=1.
\]
Example~\ref{ex:row} gives its exact circuit at normalization one. A termwise LCU using the same row oracle keeps the two products on separate labels and prepares and selects amplitudes $\sqrt{3/7},\sqrt{4/7}$, yielding $5p/7$ with two queries. The two circuits therefore have success probabilities $\|p(A)\psi\|^2$ and $25\|p(A)\psi\|^2/49$ on the same unit input. The row factorization improves this particular termwise construction by combining the products under their joint norm bound. Access costs can change the comparison: if separate unit-normalized encodings are first loaded with equal weights, the resulting row is $[A_1\ A_2]/\sqrt2$, and the same degree-two circuit returns $p(A)/2$. The supplied row, its uniform norm domain, and the cost of obtaining that row must therefore be fixed before comparing normalizations.

The same algorithms give normalization one whenever a strict bound below one is supplied. After computing a normalization smaller than one, we multiply the final known row by that normalization; the row remains contractive and its output becomes the target itself.

\begin{corollary}[Synthesis at normalization one]\label{cor:normalization-one}
Let a rational finite-state scalar polynomial of degree $D\ge1$ satisfy $\|p\|_{\SA,\Delta}\le1-\sigma$, with rational $0<\sigma<1$, and let $0<\epsilon<1$ be rational. A circuit with exact ideal selected block $p(A)$ uses $O(D/\sqrt\sigma)$ queries, or $D$ for a row input. Classical and finite-gate synthesis to error $\epsilon$ are polynomial in the input length, oracle width, $\sigma^{-1/2}$, and $\log(1/\epsilon)$; the row case depends polynomially on $\log(1/\sigma)$.
\end{corollary}
\begin{proof}
We apply the corresponding target algorithm with $\tau=\sigma/2$. Its rational normalization satisfies
\[
 0<\beta_0\le(1+\sigma/2)\|p\|_{\SA,\Delta}
           \le(1+\sigma/2)(1-\sigma)<1.
\]
Replacing its final numerical row $C_Q$ by $\beta_0C_Q$ preserves contractivity and changes the ideal block from $p/\beta_0$ to $p$. This multiplication is known rational arithmetic and adds no query. We compile the modified known matrices with total error at most $\epsilon$ using Theorem~\ref{thm:gates}; the dimensions and precision bounds are unchanged. Substitution of $\tau=\sigma/2$ gives the stated costs.
\end{proof}

When an inner completion at normalization one is supplied in the representation of Theorem~\ref{thm:row-compiler}, that theorem also handles the boundary $B=1$. Its rational joint Gram and certified entry algorithms provide the data needed for exact support decisions and finite-gate synthesis. The target-only algorithm instead uses the strict margin to discover its completion.

The row construction uses one query for each degree removed, and this count is optimal for exact synthesis in our access model. The usual polynomial method relates query count to polynomial degree~\cite{beals2001}; we apply it to the designated blocks by averaging over allowed oracle completions. Completion independence makes that average equal to the required target.

\begin{corollary}[Exact degree lower bound]\label{cor:degree-lower}
A nonzero scalar polynomial of true degree $D$ requires at least $D$ forward queries for exact completion-independent implementation at any finite positive normalization in the model of Section~\ref{sec:unitary}. The row algorithm of Corollary~\ref{cor:row-main} attains this bound.
\end{corollary}
\begin{proof}
We fix an admissible tuple and write $P_-$ and $P_+$ for the projections onto its designated oracle input and output labels. For any $|z|=1$, multiplying the original oracle on the left by $P_++z(I-P_+)$ and on the right by $P_-+z(I-P_-)$ gives another unitary with the same designated block. The four corners acquire powers $1,z,z,z^2$, respectively. Since the circuit uses only forward queries, its selected output after $Q$ calls is a polynomial in $z$ of degree at most $2Q$. Averaging over the $(2Q+1)$st roots of unity therefore retains only paths that use a designated corner at every call. Independence from the completion makes this average equal to the original target.

The averaged circuit is a product of $Q$ blocks linear in $A_1,\ldots,A_k$, interleaved with known numerical contractions, so its degree is at most $Q$. Equality on the whole input domain implies equality of free coefficients: a sufficiently small matrix tuple supported on a directed path isolates any prescribed word coefficient in one matrix entry, as detailed in Appendix~\ref{app:fixed-query}. Hence a nonzero degree-$D$ target requires $Q\ge D$, while a zero-query circuit is constant. Corollary~\ref{cor:row-main} supplies the matching construction.
\end{proof}

Appendix~\ref{app:fixed-query} complements these algorithms with a necessary and sufficient Gram condition at each prescribed query count and normalization, including singular boundary certificates. Its variables include every query history and can be exponential in the query count. The efficient synthesis above uses the smaller residual representation and a strict margin. We now extend that synthesis to a joint matrix of polynomial outputs, which will provide the data labels for quantum operations.

\section{Joint polynomial outputs and quantum channel transformations}\label{sec:channels}
The scalar circuit has a single selected output. We now retain a data register so that one query sequence can produce several polynomials coherently. This extension serves two tasks: an input label may choose which polynomial acts on the signal, or an output label may become the environment of a quantum operation. Both require us to preserve the joint amplitudes between labels. For scalar-coefficient polynomial blocks $P_{\mu\nu}(A)$, the desired successful block acts as
\begin{equation}\label{eq:block-state-action}
 \sum_{\nu=0}^{u-1}\ket\nu_{\mathsf L}\ket{\psi_\nu}_{\mathsf S}
 \ \longmapsto\ \frac1\beta\sum_{\mu=0}^{v-1}\ket\mu_{\mathsf L}
                    \sum_{\nu=0}^{u-1}P_{\mu\nu}(A)\ket{\psi_\nu}_{\mathsf S}.
\end{equation}
Here the $\psi_\nu$ can be arbitrary signal columns, including their correlations with an untouched reference. A diagonal choice of the polynomial blocks preserves the input label; the case $u=1$ produces a column of outputs. Tracing that output label after success gives Kraus operators proportional to the component polynomials. We will construct the joint transformation first, because its norm condition determines the common normalization and all of the known gates.

\subsection{Joint polynomial targets}\label{sec:block-register}\label{sec:block-certificate}
The norm of each entry alone cannot certify Eq.~\eqref{eq:block-state-action}. For example, both entries of $[I_n\ I_n]$ have norm one, but the row sends $(\ket0+\ket1)\ket\psi/\sqrt2$ to $\sqrt2\ket\psi$. More generally, the input-label blocks of the squared target are
\[
 [\mathsf P(A)^\dagger\mathsf P(A)]_{\nu\nu'}
       =\sum_\mu P_{\mu\nu}(A)^\dagger P_{\mu\nu'}(A).
\]
Because these off-diagonal terms govern interference between input labels, we seek a certificate for the whole matrix and retain the corresponding amplitudes until the final selection.

We specify the target by rectangular coefficient matrices,
\begin{equation}\label{eq:block-target}
 \mathsf P(A)=\left[\sum_w(P_w)_{\mu\nu}A_w\right]_{
                  \substack{0\le\mu<v\\0\le\nu<u}},
 \qquad B_{\mathsf P}=\sup_{n,A:\|\Delta(A)\|\le1}\|\mathsf P(A)\|.
\end{equation}
Its finite-state input has rational real and imaginary parts and takes the form
\begin{equation}\label{eq:block-representation}
 P_w=LT_wR_0,\qquad L:\ v\times s,\quad
 T_j:\ s\times s,\quad R_0:\ s\times u.
\end{equation}
The target has $un$ input and $vn$ output coordinates. We use a separate data register $\mathsf L$ of $q_L=\lceil\log_2\max\{u,v\}\rceil$ qubits, with supported input values $\nu<u$ and retained output values $\mu<v$. Its input is supplied by the application and may be entangled with the signal. The oracle register and all added processing registers start at zero. We denote by $\mathcal B_{u,v}(\mathcal U_A)$ the matrix block retaining the specified data endpoints and zero on the other auxiliary outputs. The recurrence determines the degree $D=\max\{|w|:P_w\ne0\}$, polynomiality, and the zero case by the joint tests in Appendix~\ref{app:joint-coordinates}.

\begin{theorem}[Multivariable polynomial transformation]\label{thm:matrix-synthesis}
Let $\mathsf P$ have a rational finite-state representation of degree $D\ge1$, and let $0<\tau,\epsilon<1$ be rational. Under the joint-oracle assumptions of Section~\ref{sec:unitary}, a classical algorithm computes a rational $\beta$ and a quantum circuit such that
\[
 B_{\mathsf P}<\beta\le(1+\tau)B_{\mathsf P},\qquad
 \|\mathcal B_{u,v}(\mathcal U_A)-\mathsf P(A)/\beta\|\le\epsilon.
\]
It uses $O(D/\sqrt\tau)$ unconditional forward queries, or $D$ for a row input. Classical computation and known-gate count are polynomial in the rational input length, $q_o$, $\tau^{-1/2}$, and $\log(1/\epsilon)$; the row case depends polynomially on $\log(1/\tau)$. With exact known gates, the selected block equals $\mathsf P(A)/\beta$.
\end{theorem}
The guarantee is uniform over signal sizes, admissible tuples, oracle completions, and reference systems, and reduces to scalar synthesis when $u=v=1$. Lemma~\ref{lem:block-constant} supplies the zero-query constructions for constant targets, including an always-failed outcome for the zero matrix.

To preserve these input correlations, we perform the residual calculation on polynomial rows with $u$ entries. Rational closure of the $u$ constant rows and the nonempty left residuals of the $v$ target rows gives a basis $f_1,\ldots,f_r$ with $r\le s+u$. Appending one initial letter to these residuals gives the successor matrix $g$, whose $h=u+kr$ independent rows consist of the constants and the rows $x_jf_\ell$. Its numerical coordinate maps satisfy
\[
 Eg=I_u,\qquad Jg=f,\qquad Pg=\mathsf P,\qquad L_jg=x_jf.
\]
We form $E_\alpha,M_\mu$ and $\Phi$ by the same formulas as Eqs.~\eqref{eq:EM}--\eqref{eq:Phi}. Each of the first two matrices now has $br$ rows and $h$ columns. Appendix~\ref{app:joint-coordinates} gives the coefficient arrays explicitly. For Hermitian $Q_0$ of order $h$, the hereditary matrix $q_{Q_0}=g^*Q_0g$ evaluates to $g(A)^\dagger(Q_0\otimes I_n)g(A)$ on all $un$ input coordinates together.

\begin{theorem}[Exact joint rectangular certificate]\label{lem:block-certificate}
The matrix $q_{Q_0}$ is positive semidefinite at every allowed tuple if and only if $Q_0=S+\Phi(T)$ for positive semidefinite matrices of orders $h=u+kr$ and $br$. In particular, for $\beta\ge0$, the condition $B_{\mathsf P}\le\beta$ is equivalent to
\begin{equation}\label{eq:block-SDP}
 \beta^2E^\dagger E-P^\dagger P=S+\Phi(T),\qquad S,T\succeq0.
\end{equation}
\end{theorem}
\begin{proof}
A feasible pair expresses $q_{Q_0}$ as a polynomial square and the query defect acting on residual rows, which proves positivity on all $un$ input coordinates together. For necessity, we use the separation proof of Theorem~\ref{thm:certificate} with coefficient positions indexed by a word and an input label. Summing the coercive functional over the $u$ constant inputs closes the cone; a separator then produces an admissible tuple and a column $\Omega$ whose quadratic form is $\operatorname{Tr}(HQ_0)<0$. Appendix~\ref{app:proof-lem:block-certificate} proves this joint evaluation identity, including the ordinary transpose in the complex Gram factorization and the kernel argument for singular separators. Thus the same certificate controls all input-label cross terms.
\end{proof}

\begin{proposition}[Finite-bit joint certificate computation]\label{prop:block-classical}
The rational affine system of Theorem~\ref{lem:block-certificate} is computable in polynomial bit time from Eq.~\eqref{eq:block-representation}. Given rational $\gamma>0$ and $0\le\nu<1$ with $B_{\mathsf P}/\gamma\le\nu$, the capped rational search returns positive definite rational certificate matrices for $I_u-\mathsf P^*\mathsf P/\gamma^2$ in time polynomial in the input, the bit lengths of $\gamma,\nu$, and $\log(1/(1-\nu))$. For nonzero $\mathsf P$ and rational $0<\eta<1$, it also computes $B_{\mathsf P}\le\widehat B\le(1+\eta)B_{\mathsf P}$ in polynomial bit time with dependence polynomial in $\log(1/\eta)$.
\end{proposition}
\begin{proof}
The residual constant is now an $r\times u$ matrix, and summing the coercive functional over all constant inputs gives $\operatorname{Tr}S+\operatorname{Tr}T\le\beta^2u/\chi$. The general joint certificate also supplies a strict certificate for $I_u$. These two changes give the inner and outer radii for the same capped search and norm bisection. Appendix~\ref{app:proof-prop:block-classical} computes the rational constants and bounds their bit lengths.
\end{proof}

The numerical extraction now retains $u$ input columns and $v$ output rows. For a strict certificate of a normalized target $\mathsf Q$, we set $m=br$, use $K=T^{1/2}$, and form the stacks $X=[KM_1;\ldots;KM_a;E]$ and $Y=[KE_1;\ldots;KE_b;P_{\mathsf Q}]$. Their sizes are $(am+u)\times h$ and $(bm+v)\times h$, and their Gram difference is $S\succ0$. With $T_0=X^\dagger X$, the matrix we compute is
\begin{equation}\label{eq:block-V}
 V=YT_0^{-1}X^\dagger=
 \begin{pmatrix}\Aco&\Bco\\\Cco&\Dco\end{pmatrix},\qquad
 \begin{array}{ll}
 \Aco:\ bm\times am,&\Bco:\ bm\times u,\\
 \Cco:\ v\times am,&\Dco:\ v\times u.
 \end{array}
\end{equation}
The same Gram argument proves $VX=Y$ and $\|V\|<1$. Thus $\Bco$ injects all input labels, $\Cco$ collects all output labels, and the coefficient recurrence preserves their cross terms. In the finite query construction, the pending and collected components consequently have $u$ and $v$ coordinates. Appendix~\ref{sec:block-queries} gives their complete stage matrices and packing.

The resulting query product has an initial matrix of size $d_-R\times u$, internal matrices of size $d_-R\times d_+R$, and a final matrix of size $v\times d_+R$. We embed their endpoints on
\[
 \ket0_{\mathsf O}\ket0_{\mathsf M}\ket\nu_{\mathsf L}\ket0_{\mathsf d}
 \quad\text{and}\quad
 \ket0_{\mathsf O}\ket0_{\mathsf M}\ket\mu_{\mathsf L}\ket0_{\mathsf d},
\]
respectively. Between those endpoints, the continuing component has $\mathsf L=0$ and carries the data dependence in the processing amplitudes. The full Julia completion retains its complementary amplitudes under the failure labels, so this prescription is valid even when the initial numerical matrix is only contractive. Appendix~\ref{app:joint-endpoints} specifies the remaining columns and all binary padding. The final coherent range test $\mathsf L<v$ retains the entire output column and rejects unused labels.

\begin{proof}[Proof of Theorem~\ref{thm:matrix-synthesis}]
We extract a contraction from a strict joint certificate using the Gram factorization of Section~\ref{sec:extract}, now with $u$ injection columns and $v$ collection rows. The equality $VX=Y$ generates every coefficient matrix simultaneously, as verified in Appendix~\ref{sec:block-queries}. For a general oracle we apply this extraction to $\mathsf Q=\mathcal W_v^{-1}\mathsf P/\beta$, whose degree weights are then canceled by the query sequence. Taking matrix elements between arbitrary unit input and output columns extends the scalar weight estimate to the joint operator norm with the same constant, so the normalization choice and $Q=O(D/\sqrt\tau)$ bound of Section~\ref{sec:general} apply.

On the row domain, we instead certify $\mathsf P/\beta$ directly. Its complementary polynomial satisfies the full $u\times u$ coefficient-correlation identities, which the compact degree reduction preserves until it reaches a constant isometry after $D$ steps. Appendix~\ref{sec:block-row} derives this factorization and its Gram bounds. For either product, we use the endpoint lists just specified and the same fresh-failure-label invariant as Theorem~\ref{thm:unitary}. This proves the joint block identity for every input data column and every oracle completion, as illustrated in Figure~\ref{fig:block-data-circuit}. A total known-gate operator error of $\epsilon$ therefore bounds the selected-block error by $\epsilon$, also with a reference; Proposition~\ref{prop:block-resources} gives the register and gate counts.
\end{proof}

\begin{figure}[t]
\centering
\begin{quantikz}[column sep=0.45cm,row sep={0.5cm,between origins}]
 \lstick{$\mathsf S$}&\qw&\gate[2]{U_A}&\qw&\gate{\cdots}&\gate[2]{U_A}&\qw&\qw\\
 \lstick{$\mathsf O:\ket0$}&\gate[3]{W_0}&\qw&\gate[3]{W_1}&\qw&\qw&\gate[3]{W_Q}&\rstick{$0$}\qw\\
 \lstick{$\mathsf L$}&\qw&\qw&\qw&\qw&\qw&\qw&\rstick{$\mu<v$}\qw\\
 \lstick{$\text{work}:\ket0$}&\qw&\qw&\qw&\qw&\qw&\qw&\rstick{$0$}\qw
\end{quantikz}
\caption{One query sequence with coherent data endpoints. The known layers $W_t$ act on the oracle, data, and processing registers. The bundled bottom wire contains $\mathsf M,\mathsf d,\mathsf E$ and the phase and clean work bits used in gate synthesis; their widths are given in Proposition~\ref{prop:block-resources}. The initial layer encodes an arbitrary supported input label $\nu<u$, and the final layer decodes $\mu<v$. The output range is retained coherently, with every other auxiliary output selected at zero. Only the supplied queries act on the signal.}
\label{fig:block-data-circuit}
\end{figure}
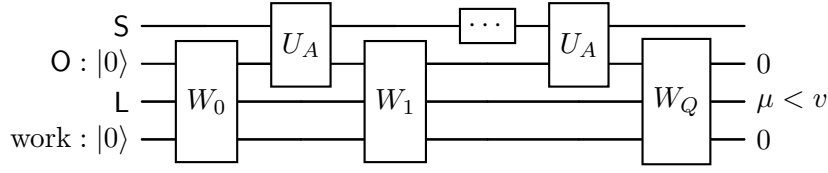

\paragraph{Retaining the polynomial label.}\label{sec:indexed-family}
An indexed family illustrates why it is useful to retain a quantum input label. To apply a chosen polynomial coherently while preserving that choice, we take a diagonal target with members $p_0,\ldots,p_{N_{\rm f}-1}$, represented by rational recurrences of sizes $s_\ell$ on the same domain. Allowing zero members and assuming at least one nonconstant member, we have
\[
 \mathsf P(A)=\operatorname{diag}(p_0(A),\ldots,p_{N_{\rm f}-1}(A)),\qquad
 D=\max_\ell\deg p_\ell,\qquad
 B_{\max}=\max_\ell\|p_\ell\|_{\SA,\Delta},\qquad s=\sum_\ell s_\ell,
\]
where zero members contribute no degree. The selected action is
\begin{equation}\label{eq:intro-indexed-action}
 \sum_{\ell=0}^{N_{\rm f}-1}\ket\ell\ket{\psi_\ell}
 \ \longmapsto\ \frac1\beta\sum_{\ell=0}^{N_{\rm f}-1}
                   \ket\ell p_\ell(A)\ket{\psi_\ell}.
\end{equation}
The retained label preserves the choice's coherence with the signal and any reference.

\begin{proposition}[Coherent indexed polynomial synthesis]\label{prop:indexed-family}
For rational $0<\tau,\epsilon<1$, one circuit implements Eq.~\eqref{eq:intro-indexed-action} to selected-block error $\epsilon$, with $B_{\max}<\beta\le(1+\tau)B_{\max}$. It uses $O(D/\sqrt\tau)$ queries, or $D$ on a row input, and a $\lceil\log_2N_{\rm f}\rceil$-qubit data label. Its classical and known-gate bounds are those of Proposition~\ref{prop:block-resources} with $u=v=N_{\rm f}$ and total state count $s$.
\end{proposition}
\begin{proof}
For representations $(\ell_\ell,T_{\ell,j},r_\ell)$, we put $T_{\ell,j}$ in diagonal state block $\ell$ of each transition. Output row $\ell$ contains $\ell_\ell$ in that block, and input column $\ell$ contains $r_\ell$ there, with all other entries zero. Thus $(P_w)_{\ell m}=\delta_{\ell m}c_{\ell,w}$. Its degree is the largest member degree, and its joint norm is $B_{\max}$ because a finite maximum commutes with the supremum over the common domain. Theorem~\ref{thm:matrix-synthesis} gives the result.
\end{proof}
The largest degree determines the query count, whereas the number of members contributes to the classical input and known-gate costs. The next application uses output labels as an environment. Its target is a column rather than a diagonal matrix, so the component squares add in its normalization.

\subsection{Operations on a coherent Kraus implementation}\label{sec:kraus}\label{sec:Kraus-gauge}
A coherent environment lets us combine Kraus histories before selecting an outcome, and hence design a successful operation with a desired action. One example is a state-preserving branch obtained by adding two histories whose sum is proportional to identity. For two Kraus operators, the relevant sum and difference are
\[
 F_+(K)=K_1K_2+K_2K_1,\qquad F_-(K)=K_1K_2-K_2K_1
\]
on separate output labels. Example~\ref{ex:Kraus-histories} exhibits a channel for which the sum gives a state-preserving successful branch. We construct these and more general polynomial operations by controlling their joint normalization.

Let a channel on the signal have a specified Kraus representation~\cite{watrous2018}
\begin{equation}\label{eq:Kraus-channel}
 \mathcal E(\rho)=\sum_{a=1}^kK_a\rho K_a^\dagger,
 \qquad \sum_aK_a^\dagger K_a=I_n.
\end{equation}
We are supplied a coherent implementation with a fixed environment basis. Its active column is the corresponding Stinespring isometry~\cite{stinespring1955}
\begin{equation}\label{eq:Kraus-column}
 \Delta(K)=\begin{pmatrix}K_1\\\vdots\\K_k\end{pmatrix},\qquad
 U_K(\ket i_{\mathsf O}\ket\psi_{\mathsf S})
       =\sum_{a=1}^k\ket{i+a-1}_{\mathsf O}K_a\ket\psi_{\mathsf S}.
\end{equation}
Our supplied resource is the complete unitary $U_K$ of Section~\ref{sec:unitary}, with $a=k$, $b=1$, $2^{q_o}\ge i+k+1$, and an identity sector that provides a phase reference for combining different degrees. The returned environment remains coherent in its specified basis, and the column structure uses the general query construction of Section~\ref{sec:general}. Appendix~\ref{app:Kraus-oracle} gives the full matrix, and Appendix~\ref{app:loaders} constructs the idle sector from a coherent selector or a controllable implementation circuit; the latter requires the additional control access discussed by Ara\'ujo et al.~\cite{araujo2014}.

The desired operation is specified by scalar-coefficient polynomials arranged in one column:
\begin{equation}\label{eq:Kraus-target-column}
 F_b(K)=\sum_w c_{b,w}K_w,\qquad
 \mathsf F(K)=\begin{pmatrix}F_1(K)\\\vdots\\F_v(K)\end{pmatrix}.
\end{equation}
We combine separate representations $(\ell_b,T_{b,j},r_b)$ using block-diagonal transitions, stacked final columns $r_b$, and output row $b$ supported by $\ell_b$ in its state block. This gives Eq.~\eqref{eq:block-representation} with $s=\sum_b s_b$ and $u=1$, so the initially zero data register becomes the output environment.

For this input domain, the joint certificate enforces
\begin{equation}\label{eq:Kraus-joint-bound}
 \sum_bF_b(A)^\dagger F_b(A)\preceq\beta^2I_n
 \qquad\text{whenever }\sum_aA_a^\dagger A_a\preceq I_n,
\end{equation}
at every matrix size. Thus $B_{\mathsf F}$ is the supremum of $\|\mathsf F(A)\|$ on the whole column ball, whose isometric boundary contains the physical Kraus column. The inequality for the sum over all output labels is the trace-nonincreasing condition for the successful operation; it is therefore the quantity controlled by the joint certificate.

With $T=\mathsf F(K)/\beta$, the selected output before discarding its label is
\[
 T\rho T^\dagger
 =\frac1{\beta^2}\sum_{b,c=1}^v\ket{b-1}\!\bra{c-1}\otimes
                              F_b(K)\rho F_c(K)^\dagger.
\]
The range test on the output label preserves these coherences. Tracing the label after success gives the desired completely positive operation,
\begin{equation}\label{eq:Kraus-success}
 \mathcal M_\beta(\rho)=\frac1{\beta^2}
                      \sum_{b=1}^v F_b(K)\rho F_b(K)^\dagger,
 \qquad p_\beta(\rho)=\operatorname{Tr}(T^\dagger T\rho).
\end{equation}
The joint bound makes $\mathcal M_\beta$ trace-nonincreasing, while the input-dependent quadratic form $p_\beta(\rho)$ determines its success probability and may vanish. We measure the error of the operation before conditioning by the diamond norm,
\[
 \|\Phi\|_\diamond=\sup_{m\ge1}\sup_{X\ne0}
 \frac{\|(\Phi\otimes\operatorname{id}_m)(X)\|_1}{\|X\|_1}.
\]
This is the completely bounded trace norm; the reference dimension can be restricted to the input dimension~\cite{watrous2018}.

\begin{corollary}[Polynomial Kraus synthesis]\label{cor:Kraus-operation}
Given the nonzero rational finite-state column in Eq.~\eqref{eq:Kraus-target-column}, complete coherent $U_K$ access, and rational $0<\tau,\eta<1$, the algorithm computes $B_{\mathsf F}<\beta\le(1+\tau)B_{\mathsf F}$ and implements Eq.~\eqref{eq:Kraus-success}. The operation is exact with ideal known gates; selected-column error $\eta$ gives diamond-norm error at most $2\eta$. For degree $D\ge1$, the column-input construction uses $O(D/\sqrt\tau)$ queries and the resources of Proposition~\ref{prop:block-resources}. Constants are handled by Lemma~\ref{lem:block-constant}.
\end{corollary}
\begin{proof}
Theorem~\ref{thm:matrix-synthesis} supplies a contractive selected column $T$ under the general-input query bound, since the Kraus oracle is a column. We retain the valid data-label range and zero on the oracle, processing, dilation, failure, phase, and clean work outputs; tracing the data label then gives Eq.~\eqref{eq:Kraus-success}.

Let $\widetilde T$ be the implemented selected column, with $\|\widetilde T-T\|\le\eta$. Both columns are contractions, and their pre-trace difference on any operator $X$ is
\[
 (\widetilde T-T)X\widetilde T^\dagger
       +TX(\widetilde T-T)^\dagger.
\]
After tensoring a reference, each term has trace norm at most $\eta\|X\|_1$. Partial trace is a channel, which proves the $2\eta$ diamond bound. An oracle error $\eta_O$ in the entire designated block $\operatorname{diag}(I_{in},\Delta(K))$, including its identity sector, and total known-gate error $\epsilon_g$ give $\eta\le Q\eta_O+\epsilon_g$. If the ideal success probability is at least $p_{\min}>2\eta$, normalizing the successful outputs gives trace-norm error at most $4\eta/p_{\min}$.
\end{proof}
An identically zero coefficient column uses an always-failed zero-query outcome. To obtain a trace-preserving channel with a retained outcome flag, we replace failure by a fixed signal state; Proposition~\ref{prop:Kraus-erasure} gives the additional registers and error bound for this postprocessing, which uses no further oracle call.

\begin{example}[Symmetric and antisymmetric Kraus histories]\label{ex:Kraus-histories}
For the degree-two targets $F_+,F_-$ introduced above, we can evaluate this normalization explicitly.
We assign the two combinations of ordered histories to orthogonal labels $\ket+$ and $\ket-$. Writing $X=A_1A_2$ and $Y=A_2A_1$ for a column contraction, we find
\[
 \begin{aligned}
 F_+(A)^\dagger F_+(A)+F_-(A)^\dagger F_-(A)
 &=2(X^\dagger X+Y^\dagger Y)\\
 &\preceq2\sum_jA_j^\dagger
                    \left(\sum_iA_i^\dagger A_i\right)A_j
 \preceq2I_n.
 \end{aligned}
\]
The cross terms cancel in the joint norm, although each labeled output retains interference between the histories. The bound is attained by the qubit Kraus operators $K_1=\ket0\!\bra1$ and $K_2=\ket1\!\bra0$, for which $F_+=I$ and $F_-=Z$, so $B_{\mathsf F}=\sqrt2$ and the coherent selected output is
\[
 \ket\psi\longmapsto
 \frac1\beta\bigl(\ket+\ket\psi+\ket- Z\ket\psi\bigr).
\]
Discarding the label gives $(\rho+Z\rho Z)/\beta^2$ with probability $2/\beta^2$. Selecting the $+$ label instead has probability $1/\beta^2$ and conditional output $\rho$, also for inputs entangled with a reference. The theorem chooses a rational $\beta>\sqrt2$ within the stated relative margin; the boundary value $\sqrt2$ gives probability-one success for the joint output in this example.

These targets also have a direct two-call construction. We retain both returned environments and order their four histories by the products $K_1K_2,K_2K_1,K_1K_1,K_2K_2$. We apply the selected environment matrix
\[
 C=\frac1\beta\begin{pmatrix}1&1&0&0\\1&-1&0&0\end{pmatrix},
 \qquad CC^\dagger=\frac2{\beta^2}I_2\preceq I_2.
\]
Its Julia completion gives the two outputs using clean-input Stinespring access and fresh environments. This fixed-length construction illustrates the operation directly; the joint-oracle compiler also accepts finite-state targets that combine different history lengths.
\end{example}

\begin{remark}[When the successful operation is a channel]
If an exact realization has $\beta=1$ and $\sum_bF_b(K)^\dagger F_b(K)=I_n$, its selected column is an isometry and success has probability one. Tracing its output label then produces the target channel. For an isometric target column at $\beta>1$, success is $\beta^{-2}$; tracing every auxiliary output also includes the failed Kraus operators. Appendix~\ref{app:joint-resources} gives the coherent isometry error bound for the rectangular circuit.
\end{remark}

The operation is specified relative to the coherent implementation. Kraus lists for the same channel, padded to a common length, obey $\widetilde K_b=\sum_aV_{ba}K_a$ for an environment unitary $V$~\cite[Corollary~2.23]{watrous2018}. Tracing the environment removes this change of basis, while evaluating fixed polynomials in its amplitudes generally retains it. In the history example, the change $(K_1,K_2)\mapsto((K_1+K_2)/\sqrt2,(K_1-K_2)/\sqrt2)$ gives $F_+=0$ and $F_-=-Z$ for the displayed qubit channel, so it changes the successful operation. This dependence on implementation is the additional information studied in coherent channel control~\cite{abbott2020} and coherent-trajectory models~\cite{chiribellakristjansson2019}.

The coherent identity sector also fixes a phase reference when different degrees are combined. For example, the one-Kraus lists $K=I_n$ and $K=-I_n$ represent the same identity channel, but $(I_n+K)/2$ is respectively $I_n$ and zero. A classical instrument specifies only the maps $K_a\rho K_a^\dagger$, which are unchanged by independent Kraus phases. The algorithm therefore receives the full coherent unitary and its environment basis as input. Example~\ref{ex:Kraus-filter} in Appendix~\ref{app:Kraus-idle-example} gives a two-output filter using the interference with the identity sector.

The direct two-call construction extends to $d$ clean-input Stinespring calls by retaining fresh environments and applying a supplied contraction $C$ to their histories. Appendix~\ref{app:Kraus-histories} derives the condition $C^\dagger C\preceq I$ and the cross-history operation, with $d q_{\rm env}$ environment qubits and $vk^d$ explicit matrix entries. The polynomial compiler uses a compact recurrence and the coherent identity sector to combine different lengths, with normalization determined by the associated column ball. These different input constraints account for the different normalization problems.

\subsection{Channel-only access and causal controllers}\label{sec:channel-only}
For channel-only access, the object to be specified is a controller interleaved with the channel calls. One source of such a specification is optimization over causal Choi matrices~\cite{chiribellaebler2016}. As a supplementary construction, we give the classical step that converts supplied rational controller data into the processing gates and retained memory. Such networks are described by quantum strategies~\cite{gutoskiwatrous2007} and combs~\cite{chiribella2009}, extending the single-slot supermap setting~\cite{chiribella2008}. If $\mathcal E_1,\ldots,\mathcal E_Q$ act on the signal and $\mathcal V_j(X)=U_jXU_j^\dagger$ are known processing channels on signal and memory, the required mapping is
\begin{equation}\label{eq:intro-comb-map}
 \begin{aligned}
 \mathfrak C[\mathcal E_1,\ldots,\mathcal E_Q](\rho)
  =\operatorname{Tr}_{\mathsf M_{\rm c}}\bigl[
  \bigl(\mathcal V_{Q+1}\circ(\mathcal E_Q\otimes\operatorname{id})
       \circ\mathcal V_Q\circ\cdots
  {}\circ(\mathcal E_1\otimes\operatorname{id})\circ\mathcal V_1\bigr)
       (\rho\otimes\ket0\!\bra0_{\mathsf M_{\rm c}})\bigr].
 \end{aligned}
\end{equation}
The identities act on $\mathsf M_{\rm c}$, which starts at zero and is discarded at the end. This formulation specifies a function of the channel slots, because their environments are traced by the channel calls themselves. The known processing gates may act on the signal as well as the retained memory.

To compute these gates, we receive exact rational Choi matrices for the controller prefixes. Positivity specifies complete positivity through the Choi correspondence~\cite{choi1975}; a partial-trace identity at each prefix specifies causal normalization. The realization theorem for combs~\cite{chiribella2009} and its prefix-rank construction~\cite{bisio2010} supply the mathematical basis. Appendix~\ref{app:comb} adapts them to numerical Gram factors, full square unitaries, and finite-bit gate synthesis. Proposition~\ref{prop:comb-compilation} gives a polynomial bound in the explicit Choi input length and $\log(1/\epsilon)$, including an optional successful outcome with a supplied deterministic completion. A dense terminal Choi matrix has $n^{4(Q+1)}$ entries, so this bound is measured in that explicit description. The polynomial algorithm supplies the compact-coefficient construction of this paper; compilation from the supplied Choi matrices describes the additional data and operations needed when the target is instead a function of channel slots alone.

\section{Discussion}\label{sec:discussion}
We have constructed multivariable polynomial transformations by connecting a compact coefficient description to a finite Schur--Agler certificate and then to a quantum circuit. Its completeness in the target's residual space connects the classical input size to the norm calculation and gate parameters. The general sequence and row factorization then show how the input layout affects the query count.

A natural next question is which other layouts permit synthesis at the degree. The row factorization removes the normalization-margin dependence from the query count, while our general construction costs $O(D/\sqrt\tau)$ queries. Understanding when this dependence can be reduced will require a suitable invariant for the coefficients or a different way to combine their products. Exact boundary discovery is another issue: the finite-history criterion gives algebraic-input decidability at a prescribed query count, while compact certificates at the norm boundary would extend the efficient synthesis algorithm. We would also like to optimize the coefficient representation~\cite{berstelreutenauer2011}, positive certificate and processing memory together, since minimizing any one of them need not minimize the others.

For matrix-function applications, synthesis should be paired with approximations that are compact and controlled in the same joint norm. Free holomorphic approximation~\cite{aglermccarthy2015} and noncommutative function theory~\cite{kaliuzhnyivinnikov2014} provide the setting; Appendix~\ref{sec:applications} supplies one computable example with a tail and success bound. In single-variable QSP, constrained minimax approximation explicitly couples approximation quality to the global implementation bound~\cite{dong2026minimax}. An analogous design process on a joint matrix domain could help identify useful multivariable filters, with loading and input-state costs included in their analysis.

The joint-output construction makes polynomial coefficients available as design variables for heralded filters and instruments. We can ask for a desired family of Kraus polynomials, impose a common trace bound, and use the certificate to calculate the corresponding gates. For quantum-information applications, the next steps are to find families with useful success probabilities and to understand their dependence on the chosen coherent implementation. Representation-independent families would connect this approach more directly to channel-level tasks, while robustness must distinguish fixed environment coordinates from the aligned dilations of Stinespring continuity~\cite{kretschmann2008}. For channel-only controllers, causal-network optimization~\cite{chiribellaebler2016} can supply the Choi data for the separate compiler in Appendix~\ref{app:comb}. Compact descriptions preserving the intermediate causal identities, memory tradeoffs~\cite{bisio2012memory}, and stability under adaptive tests~\cite{gutoski2012} are the immediate questions for that route; other higher-order maps require their own realization conditions~\cite{bisioperinotti2019}.

\section*{Note added}
\phantomsection
\addcontentsline{toc}{section}{Note added}
In concurrent work, Low~\cite{low2026} develops multivariate quantum signal processing for noncommuting matrices through Schur--Agler realizations and analytic transfer-function design. Our emphasis is on complete synthesis from compact rational coefficient descriptions under joint block access, with near-optimal uniform normalization and polynomial bit complexity. We also treat quantum channel transformations, including polynomial operations on coherent Kraus implementations and gate synthesis from causal Choi data. Jing, Wang, Yu, and Zhang~\cite{jing2026} give multivariate polynomial transformations with optimal per-variable query counts and extend their construction to commuting unitaries. Jiang and An~\cite{jiangan2026} construct multivariable quantum eigenvalue transformations for commuting matrices using contour-integral and Fourier-transform methods.

\Needspace{12\baselineskip}
\section*{AI disclosure}
\phantomsection
\addcontentsline{toc}{section}{AI disclosure}
The central ideas and main components of this work were conceived by the authors. In developing the results and preparing the manuscript, the authors engaged in multi-round discussions with OpenAI's models, beginning with GPT-5.3 Pro and continuing through GPT-6 Astra. These interactions covered methodological details, alternative formulations and derivations, proof development and checking, searches for gaps and counterexamples, literature searches, diagnostic code and numerical checks, and the organization, exposition, and revision of the manuscript. The authors directed the research and made the final decisions about the arguments and presentation. They take full responsibility for the results, proofs, citations, and final text.

\section*{Acknowledgments}
\phantomsection
\addcontentsline{toc}{section}{Acknowledgments}
The authors thank Zhong-Xia Shang and Dong An for helpful discussion. This work is supported by the National Natural Science Foundation of China Grant No.~12361161602 and NSAF Grant No.~U2330201, the Quantum Science and Technology--National Science and Technology Major Project No.~2023ZD0300200, Beijing Natural Science Foundation Z250004, and Beijing Science and Technology Planning Project Grant No.~Z25110100810000. Y.~Wu acknowledges support from NSFC Grants No.~62501060, No.~62461160263, and No.~62371050.

\titleformat{\section}{\normalfont\Large\bfseries}{\thesection}{1em}{}
\phantomsection
\addcontentsline{toc}{section}{References}
\bibliographystyle{alpha}
\bibliography{ref}

\medskip

\appendix
\hypersetup{bookmarksdepth=1}
\titleformat{\section}{\normalfont\Large\bfseries}{Appendix~\thesection.}{0.6em}{}
\addtocontents{toc}{\protect\setcounter{tocdepth}{1}}
\addtocontents{toc}{\protect\vspace{6pt}\protect\noindent\protect\textbf{Appendices}\protect\par}

\section{Finite-bit matrix computation}\label{app:arithmetic}
The classical synthesis algorithm must make exact decisions about coefficient identities and matrix supports, and it must approximate the numerical gate entries with certified error. We first give the quantitative semidefinite search and norm estimation, then the rational coefficient and support calculations on which the constructions depend, and finally the matrix-function evaluators. Rank and feasibility are decided from exact rational data throughout. Gate compilation uses the certified intervals and overlapping tolerance tests of Appendix~\ref{app:gates}.

\subsection{Strict certificates and finite norm computation}\label{app:classical-margins}\label{app:classical-norm-search}
The feasibility routine needs an outer bound on its positive matrices and an inner ball that is guaranteed by the strict norm promise. A scaled coefficient Gram provides the outer bound; a uniform estimate on the residual polynomials gives a positive definite certificate of the constant polynomial and hence the inner ball. Both estimates can be computed in residual coordinates, without constructing the exponentially large word-evaluation tuple used to prove completeness.

A rational left inverse of the $ab\times k$ matrix whose columns are the entries of $F_j$ gives constants $c_j^{\Delta}\ge1$ such that $\|A_j\|\le c_j^{\Delta}$ whenever $\|\Delta(A)\|\le1$. Indeed, the same entry formula expresses $A_j$ as a scalar linear combination of the $n\times n$ blocks of $\Delta(A)$, each of norm at most one. The sum of absolute real and imaginary parts of those scalar coefficients is a rational admissible bound. We also choose a positive rational $\rho\le1$ satisfying
\[
 \rho^2\left\|\sum_jF_j^\dagger F_j\right\|\le\tfrac14;
\]
a rational entrywise upper bound on the matrix norm suffices for this choice.

The weighted coefficient Gram of the residual list is computed by the finite recurrence
\begin{equation}\label{eq:Stein}
 W_f=a_0a_0^\dagger+
          \rho^2\sum_j A_j^{\rm res}W_f(A_j^{\rm res})^\dagger.
\end{equation}
Iteration from zero terminates after at most $D$ updates, because all length-$D$ residual transition products vanish. If $f_{\ell,w}$ denotes the scalar coefficient of $x_w$ in $f_\ell$, then its entries are
$[W_f]_{\ell t}=\sum_w\rho^{2|w|}f_{\ell,w}\overline{f_{t,w}}$.
The independence of the residuals makes $W_f$ positive definite. The Gram of the list $g$ is the explicit block matrix
\begin{equation}\label{eq:Wg}
 W_g=\begin{pmatrix}
 1&0&\cdots&0\\
 0&\rho^2W_f&&0\\
 \vdots&&\ddots&\\
 0&0&&\rho^2W_f
 \end{pmatrix},
\end{equation}
with $k$ copies of the lower diagonal block, so its order is $h=1+kr$. There are no off-diagonal blocks because different initial letters produce different words.

We compute a rational spectral lower bound
\begin{equation}\label{eq:chi}
 \lambda=\left(1+\sum_{u,v}
       (|\operatorname{Re}(W_f^{-1})_{uv}|+
        |\operatorname{Im}(W_f^{-1})_{uv}|)\right)^{-1},
 \qquad \chi=\tfrac34\rho^2\lambda.
\end{equation}
Here $0<\lambda\le\min\{1,\lambda_{\min}(W_f)\}$. At the finite tuple used in Eq.~\eqref{eq:coercive}, $I-\Delta^\dagger\Delta\succeq3I/4$. Each polynomial $h_{\alpha,\ell}$ has degree at most $D-1$, so its coefficient norm on the empty-word input is evaluated exactly by the truncated tuple. The first square term is bounded below by the smallest eigenvalue of $W_g$, and the query-defect term by $3\lambda/4$. Thus the chosen $\chi$ verifies Eq.~\eqref{eq:coercive}. In particular, every certificate of $\beta^2-p^*p$ obeys
\begin{equation}\label{eq:trace-cap}
 \tr S+\tr T\le\beta^2/\chi,
\end{equation}
because its value on the empty-word input is at most $\beta^2$.

A uniform bound on the residual polynomials can also be computed from their recurrence. Replacing numerical matrix norms by rational majorants, we choose
\begin{equation}\label{eq:Uf}
 U_f\ge\|a_0\|\sum_{d=0}^{D-1}
               \left(\sum_jc_j^{\Delta}\|A_j^{\rm res}\|\right)^d.
\end{equation}
For every unit signal vector, the square root of $\sum_\ell\|f_\ell(A)\ket\psi\|^2$ is at most $U_f$, by the Euclidean triangle inequality applied to the coefficient expansion. Consequently $\sum_\ell f_\ell(A)^\dagger f_\ell(A)\preceq U_f^2I_n$. This yields a rational bound
\begin{equation}\label{eq:sigma}
 M_0=1+\left(b+\sum_j(c_j^{\Delta})^2\right)U_f^2,
 \qquad \sigma_0=(2M_0)^{-1}
\end{equation}
on the positive expression represented by $I_h+\Phi(I_{br})$. Its square term is bounded by $I_n+(\sum_j(c_j^{\Delta})^2)U_f^2I_n$, and its query-defect term lies between zero and $bU_f^2I_n$. Therefore the constant polynomial minus $\sigma_0$ times this expression is at least $I_n/2$ at every admissible tuple. Applying the generalized certificate of Theorem~\ref{thm:certificate} to that difference gives the following quantitative consequence. We prove its constants explicitly, since the classical stopping bound depends on their encoding lengths.

\begin{lemma}[Explicit strict margins]\label{lem:margin}
The constant polynomial $1$ has a certificate in the same coordinates with $S\succeq\sigma_0I_h$ and $T\succeq\sigma_0I_{br}$. For real $\gamma>0$ and $0\le\nu<1$, if
$\|p/\gamma\|_{\SA,\Delta}\le\nu$, then $1-p^*p/\gamma^2$
has a certificate with
\begin{equation}\label{eq:mu}
 S\succeq\mu_*I_h,\qquad T\succeq\mu_*I_{br},
 \qquad \mu_*=(1-\nu^2)\sigma_0,
 \qquad \tr S+\tr T\le1/\chi.
\end{equation}
When $\gamma>0$ and $0\le\nu<1$ are rational input parameters,
all margin constants above can be computed rationally. Their logarithmic
magnitudes and inverse magnitudes are polynomially bounded in the rational
input lengths and $\log(1/(1-\nu))$.
\end{lemma}
\begin{proof}
The positive difference has coefficient matrix
\[
 Q_c=E^\dagger E-\sigma_0\bigl(I_h+\Phi(I_{br})\bigr).
\]
The uniform estimate above proves $q_{Q_c}(A)\succeq I_n/2$ at every admissible tuple. The general hereditary assertion of Theorem~\ref{thm:certificate} therefore gives $Q_c=S_c+\Phi(T_c)$ with $S_c,T_c\succeq0$. Rearranging yields the constant certificate
\[
 E^\dagger E=(S_c+\sigma_0I_h)
                 +\Phi(T_c+\sigma_0I_{br}),
\]
whose two matrices have the required positive margins. The hereditary polynomial with coefficient matrix $\nu^2E^\dagger E-P^\dagger P/\gamma^2$ is positive on the domain by the norm promise, so the same theorem supplies its certificate. Adding $1-\nu^2$ times the displayed constant certificate gives the exact coefficient matrix $E^\dagger E-P^\dagger P/\gamma^2$ with the lower bounds in Eq.~\eqref{eq:mu}, and Eq.~\eqref{eq:trace-cap} bounds its trace. The constants have controlled bit length because Eq.~\eqref{eq:Stein} is a recurrence of length at most $D$ on rational matrices of order $r$: common denominators bound its entries, while inverses and coordinate solutions are ratios of rational minors. The finitely many sums and products in Eqs.~\eqref{eq:Uf}--\eqref{eq:sigma} also have polynomial logarithmic size, so all spectral bounds are computable from the rational residual data and the supplied strict margin with the asserted encoding lengths.
\end{proof}

We now eliminate $S$ from Eq.~\eqref{eq:certificate} for $p/\gamma$ and search over the independent real entries of $T$. There are $d_0=(br)^2$ real coordinates, counting a complex off-diagonal entry by its real and imaginary parts. The search region is
\begin{equation}\label{eq:search-region}
 \begin{aligned}
 T&\succeq\mu_*I_{br}/4,\\
 E^\dagger E-P^\dagger P/\gamma^2-\Phi(T)&\succeq\mu_*I_h/4,\\
 |\operatorname{Re}T_{uv}|,\ |\operatorname{Im}T_{uv}|&\le U_0,
 \qquad U_0=1+2/\chi.
 \end{aligned}
\end{equation}
The rational number $L_\Phi=1+\sum_\alpha\|E_\alpha\|_F^2+
\sum_\mu\|M_\mu\|_F^2$ is a norm majorant for the affine map. A perturbation of Euclidean coordinate norm $d$ changes $T$ by operator norm at most $2d$ and $\Phi(T)$ by at most $2L_\Phi d$. The strict certificate from Lemma~\ref{lem:margin} is therefore the center of a ball of radius
\begin{equation}\label{eq:radii}
 r_{\rm in}=\min\{1,\mu_*/(8L_\Phi)\},
 \qquad
 R_{\rm out}=U_0\max\{1,d_0\}
\end{equation}
contained in the search region, whereas the entire region lies in the ball of radius $R_{\rm out}$ about the origin. The center of the inner ball need not be known.

\paragraph{Rational separation and stopping.}
For a rational Hermitian matrix $M$, elimination either proves $M\succeq0$ or finds a rational column $v$ with $v^\dagger Mv<0$. A negative diagonal entry is already such a witness. A positive pivot permits a Schur complement, and a witness for that complement is lifted by rational triangular substitution. If a diagonal pivot is zero and its row has a nonzero off-diagonal entry, we choose a two-coordinate column with a rational multiple and phase of that entry; its quadratic form has a linear negative term, which dominates the quadratic term for a sufficiently small positive rational multiplier. If the row is zero it is deleted. These cases give an exact rational test, with bit bounds controlled by minors as in fraction-free elimination~\cite{bareiss1968}. Applying $v^\dagger M(z)v$ to an affine matrix gives a separating linear inequality in its real coordinates, and box violations give immediate coordinate cuts.

The ellipsoid method uses the inner and outer radii in Eq.~\eqref{eq:radii} as follows. In real dimension $d\ge2$, a central cut with normal $a$ updates a center $c$ and a positive shape matrix $H$ by
\begin{equation}\label{eq:ellipsoid}
 c'=c-\frac{Ha}{(d+1)\sqrt{a^{\mathsf T}Ha}},
 \qquad
 H'=\frac{d^2}{d^2-1}
      \left(H-\frac{2}{d+1}
                    \frac{Haa^{\mathsf T}H}{a^{\mathsf T}Ha}\right).
\end{equation}
The new ellipsoid contains the retained half of the old one. Its logarithmic volume ratio is
\[
 -\frac d2\log(1-d^{-2})-\operatorname{arctanh}(d^{-1})
 =-\sum_{j\ge1}\frac{d^{-(2j-1)}}{2j(2j-1)}
 \le-\frac1{2d}.
\]
A separator at an infeasible center can be weakened to the parallel central cut through that center while preserving every feasible point. If its normal is zero, the violated affine inequality is constant, so the search returns failure immediately. These central cuts implement separation-to-feasibility~\cite{groetschell1981,groetschell1988}; the rounding choices below preserve their volume decrease with rational arithmetic.

We take $\theta=1/(100d^2)$. If $H_t$ is the shape after $t$ rounded cuts, we maintain the rational spectral lower bound $\lambda_{\min}(H_t)\ge R_{\rm out}^2 4^{-t}$. The exact shape update obeys
\[
 \frac{d^2}{(d+1)^2}H_t\preceq H'\preceq
                  \frac{d^2}{d^2-1}H_t,
\]
so $\ell_t=R_{\rm out}^2 4^{-(t+1)}$ is a valid lower bound for $H'$. We round its entries symmetrically to a dyadic grid of spacing $g_t\le\theta\ell_t/(8d)$, then add $dg_tI$. The resulting rational $H_q$ satisfies $H'\preceq H_q\preceq(1+\theta)H'$, since the operator-norm entry-rounding error is at most $dg_t/2$. We also compute a rational center $c_q$ with $\|c_q-c'\|_2\le\theta\sqrt{\ell_t}$ and use the new shape $(1+\theta)^2H_q$. The center tolerance is effective because $\sqrt{\ell_t}=R_{\rm out}2^{-(t+1)}$ is rational. The triangle inequality in the $H'^{-1}$ norm proves that this rounded ellipsoid contains the exact updated ellipsoid.

The rounding multiplies volume by at most $(1+\theta)^{3d/2}$, so a rounded cut reduces volume by at least the factor $\exp(-1/(4d))$. Its shape remains bounded below by $H_t/4$ and above by $2H_t$, which proves the stated lower bound inductively and also gives $\lambda_{\max}(H_t)\le 2^tR_{\rm out}^2$. Hence all grid sizes, center-square-root evaluations, and rational separators require polynomially many bits in $t$, the input length, and the radius logarithms. Rounding after every cut prevents the denominator growth of an exact iteration.

For an entirely rational stopping rule, we compute an integer $b_r\ge0$ with $2^{b_r}\ge\max\{1,R_{\rm out}/r_{\rm in}\}$ and use at most
\begin{equation}\label{eq:ellipsoid-cap}
                    N_{\rm ell}=8d^2(1+b_r)
\end{equation}
cuts. If no center has been feasible by then, the volume bound is smaller than the volume of a ball of radius $r_{\rm in}$, contradicting the promised retained ball. Dimension one uses midpoint interval bisection for at most $b_r+2$ cuts. Thus the cap is computable before the search starts; failure disproves the specified interior-ball promise and does not decide arbitrary exact boundary feasibility. This is exactly the distinction required by the norm search of Theorem~\ref{thm:norm-search}.

\begin{proof}[Quantitative proof of Proposition~\ref{prop:strict}]
The radii in Eq.~\eqref{eq:radii} allow us to run the rational separation and rounded ellipsoid procedure on Eq.~\eqref{eq:search-region}. Exact Hermitian elimination supplies a rational separator for a violated matrix inequality, and the box constraints give coordinate cuts. Minor bounds control the separators' encoding lengths, while Lemma~\ref{lem:margin} bounds the logarithm of the radius ratio. Together with the rounding estimates, these bounds give polynomial bit cost through the predetermined cap in Eq.~\eqref{eq:ellipsoid-cap}, by which a feasible rational center must have been found if the strict promise holds. Checking every returned $T$ against the original rational inequalities ensures that $S=E^\dagger E-P^\dagger P/\gamma^2-\Phi(T)$ satisfies the coefficient identity exactly with the reserved positive margins. Conversely, reaching the cap without a verified point contradicts the inner-ball guarantee and therefore disproves the stated norm promise, which is the lower-bound implication needed by norm search; no conclusion about arbitrary boundary feasibility is required.
\end{proof}

\paragraph{Relative norm estimation.}
The capped search can estimate an unknown norm because its two outcomes give verified upper and lower bounds at nearby thresholds. It starts from rational bounds computed from the coefficient recurrence. For a nonzero nonconstant target, evaluation at the scaled word tuple gives
\begin{equation}\label{eq:norm-bracket}
 t_p=PW_gP^\dagger>0,\qquad
 \ell_0=\frac{t_p}{1+t_p}\le\sqrt{t_p}\le B.
\end{equation}
A rational upper bound follows from Eq.~\eqref{eq:residual-identities}, the $c_j^{\Delta}$, and $U_f$, for example by bounding
$|c_\emptyset|+\sum_j\|p_j\|c_j^{\Delta}U_f$. Both endpoints have polynomial bit length. The following proof gives the stopping bound for the bracket updates, including the case of a constant target.

\begin{proof}[Detailed proof of Theorem~\ref{thm:norm-search}]
For a nonzero constant polynomial, $B=|c_\emptyset|$ is obtained to relative precision by rational bisection using the exact square $|c_\emptyset|^2$; the positive bracket $|c_\emptyset|^2/(1+|c_\emptyset|^2)\le B\le1+|c_\emptyset|^2$ gives the stated polynomial bit bound. We therefore assume the target is nonconstant for the remaining argument. For a rational threshold $b_0>0$, we run the strict-search procedure for the normalization $(1+\eta/8)b_0$, with the inner-ball promise that follows from $B\le b_0$. The number of iterations is capped using that promise. If the promise holds, the algorithm must find a feasible rational point. Whenever it does find such a point, exact verification proves $B\le(1+\eta/8)b_0$, whether or not the initial promise held. If the cap is reached without a feasible point, we conclude only that $B>b_0$. The proved iteration bound is what makes this failed search a valid lower-bound test.

Starting with the bracket $[\ell,U]$ provided by Eq.~\eqref{eq:norm-bracket}, we use $b_0=(\ell+U)/2$. On failure we update the lower endpoint to $b_0$; on success we update the upper endpoint to $\min\{U,(1+\eta/8)b_0\}$. The bracket continues to contain $B$. While $U>(1+\eta)\ell$, a success gives length at most
\[
 \tfrac12(U-\ell)+\tfrac\eta8 b_0
 \le\left(\tfrac58+\tfrac\eta{16}\right)(U-\ell)
 \le\tfrac{11}{16}(U-\ell),
\]
and a failure halves its length. The number of iterations before $U\le(1+\eta)\ell$ is therefore polynomial in the input bit length and $\log(1/\eta)$. At that point the rational upper endpoint is the desired estimate. The same strict subroutine can subsequently produce a certificate at any normalization with an explicitly supplied margin above this bound.
\end{proof}

\subsection{Finite matrix witnesses}

\begin{corollary}[A finite matrix size for the norm]\label{cor:finite-norm}
For the nonconstant scalar polynomial in Theorem~\ref{thm:certificate}, the norm $B$ is attained on a tuple of matrices of order at most $h=1+kr$. The norm can equivalently be written as
\[
 B=\max_{\|\Delta(Y)\|\le1}\|p(Y)\|,
 \qquad Y_1,\ldots,Y_k\text{ are }h\times h\text{ matrices}.
\]
For every $0\le\beta<B$, the explicit formulas in Eq.~\eqref{eq:GNS-matrices} construct such a finite tuple from a separator of $\beta^2E^\dagger E-P^\dagger P$ and give $\|p(Y)\|>\beta$.
\end{corollary}
\begin{proof}
We first explain directly why these finite coordinates suffice to witness the norm. For an admissible tuple $A$ and a unit vector $\psi$, we take $\mathcal H_0$ to be the span of the $h$ vectors $g_t(A)\psi$ and write $P_0$ for its orthogonal projection. We compress each input to $Y_j=P_0A_j|_{\mathcal H_0}$. The block matrix $\Delta(Y)$ is a compression of $\Delta(A)$ on its input and output copies of $\mathcal H_0$, so it is contractive. The space contains $\psi$, every $f_\ell(A)\psi$, and $p(A)\psi$ by the coordinate identities.

Induction on residual degree now gives $f(Y)\psi=f(A)\psi$ for every residual $f$. Indeed, this is immediate for constants; once it holds for the left residuals $\partial_jf$, the identity $f=f_\emptyset+\sum_jx_j\partial_jf$ gives $f(Y)\psi=P_0f(A)\psi=f(A)\psi$. The same argument then gives $p(Y)\psi=p(A)\psi$. Thus each value contributing to the supremum of $\|p(A)\|$ is already attained on a space of dimension at most $h$. Padding by zero blocks places all these witnesses at size $h$ without changing the original polynomial action or the input bound.

To obtain an attained maximum, we use independence of the $F_j$. A numerical left inverse of their vectorized coefficient matrix expresses each $Y_j$ as a fixed linear combination of the blocks of $\Delta(Y)$, so the set of admissible size-$h$ tuples is bounded as well as closed. It is therefore compact, and continuity of $p$ proves the displayed maximum. The separator construction in Eq.~\eqref{eq:GNS-matrices} supplies the stated explicit witness whenever $\beta<B$, because its negative quadratic form for $\beta^2-p^*p$ implies $\|p(Y)\|>\beta$. This dimension reduction is consistent with finite-space positivity tests~\cite[Section~4.4]{heltonklepnelson2015}; the dual certificate below gives a numerical form of the same norm witness.
\end{proof}

For applications requiring an explicit input witnessing a lower bound, the primal certificate has the exact numerical dual
\begin{equation}\label{eq:dual}
 B^2=\max\{\tr(HP^\dagger P):
      H\succeq0,\ \Phi^*(H)\succeq0,\ \tr(HE^\dagger E)=1\}.
\end{equation}
We prove the dual equality and then use a feasible dual matrix to construct an input witnessing its lower bound. Pairing Eq.~\eqref{eq:certificate} with feasible $H$ gives weak duality. The strict constant certificate from Lemma~\ref{lem:margin} gives $1\ge\sigma_0\tr H$, so the dual feasible set is bounded and closed. It is nonempty because the scaled word-evaluation functional used in Eq.~\eqref{eq:coercive} gives such a matrix. For every $0\le\beta<B$, separation of the closed primal cone gives a nonzero $H\succeq0$ with $\Phi^*(H)\succeq0$ and $\tr(H(\beta^2E^\dagger E-P^\dagger P))<0$. The constant certificate ensures $\tr(HE^\dagger E)>0$, permitting normalization. These points approach $B^2$, and compactness proves equality and attainment in Eq.~\eqref{eq:dual}.

For any dual feasible point, the explicit matrices in Eq.~\eqref{eq:GNS-matrices} satisfy $\|\Delta(Y)\|\le1$, $\|\omega\|=1$, and
$\|p(Y)\omega\|^2=\tr(HP^\dagger P)$ by Eq.~\eqref{eq:GNS-reproduction}. Thus a dual point gives actual finite matrix inputs witnessing its lower bound. The ordinary transpose in $W^\dagger W=H^{\mathsf T}$ is necessary: factoring $H$ instead changes the conjugation in Eq.~\eqref{eq:GNS-localizing} for complex joint-block coefficients. The construction applies to each supplied feasible dual point; the efficient norm-search algorithm uses the strict primal certificates described in the main text.

\subsection{Exact coefficient and support calculations}\label{app:coefficient-inputs}
Reachable and observable reduction of Eq.~\eqref{eq:finite-state} gives a minimal representation in which polynomiality is equivalent to joint nilpotence of the transitions. Indeed, if all sufficiently long coefficients vanish, then $\ell T_vT_wT_zr_0=0$ for all $v,z$ and sufficiently long $w$; reachability and observability force $T_w=0$. The converse follows directly from the coefficient formula. To decide nilpotence, we compute successive column bases of $I_s$, $[T_1\ \cdots\ T_k]$, and $[T_1N\ \cdots\ T_kN]$, where $N$ is the preceding retained basis. These spans form a descending sequence that either reaches zero within $s$ steps or stabilizes at a nonzero span. Equality of consecutive spans forces equality at every later step, so in the nilpotent case each nonzero span strictly decreases in dimension. The nilpotence index is at most the reduced state count, and a nonzero $s$-state polynomial has degree $D\le s-1$. Its true degree and the zero case then follow from the finite recurrence
\begin{equation}\label{eq:degree-test}
 H_0=r_0r_0^\dagger,\qquad H_{d+1}=\sum_jT_jH_dT_j^\dagger,
 \qquad \ell H_d\ell^\dagger=\sum_{|w|=d}|c_w|^2.
\end{equation}
We decide all these equalities by exact rational arithmetic, so the reported degree is the true polynomial degree.

For a sparse word list, we make one directed path for each supplied word, label its edges by the successive letters, and place the word coefficient on one edge; a direct edge records the constant coefficient. The total number of edges is at most the sum of the supplied word lengths plus one, so this gives an explicit acyclic branching program of polynomial size. The branching-program input mentioned in the introduction admits the same representation directly. With vertices in topological order, we write its adjacency matrix as $Z_0+\sum_j Z_jx_j$, where every numerical $Z_j$ is strictly upper triangular and $Z_0$ contains the constant edge terms. If there are $s$ vertices, then $C_0=(I_s-Z_0)^{-1}=\sum_{t=0}^{s-1}Z_0^t$ is rational and
\[
 \ell\bigl(I_s-Z_0-\textstyle\sum_jZ_jx_j\bigr)^{-1}r_0
 =\ell C_0\bigl(I_s-\textstyle\sum_j Z_jC_0x_j\bigr)^{-1}r_0.
\]
Consequently the initial row $\ell C_0$, transitions $Z_jC_0$, and final column $r_0$ give Eq.~\eqref{eq:finite-state}. The transitions remain strictly upper triangular, so only finitely many terms occur. Rational matrix multiplication and elimination have polynomial bit cost in the explicit program, as required for the conversion.

The row factorization also needs a reachable and observable representation when the output coefficients contain square roots. Its exact rational joint Gram contains the information needed to determine the support.
For the complementary representation in Eq.~\eqref{eq:row-state-coeff}, we retain a rational column-basis matrix $B$ of the reachable coordinates. A rational left inverse $B^{\ell}$ gives the restricted matrices
$N'_j=B^{\ell}N_jB$, $b'_j=B^{\ell}b_j$, and $O'=OB$.
The restricted joint Gram is a rational congruence of the original one. We compute a rational row-basis matrix $Q_{\rm obs}$ by closing the rows of $O'^\dagger O'$ under right multiplication by each $N'_j$. Its kernel is exactly the common kernel of the matrices $O'N'_w$: the initial Gram has the same kernel as $O'$, and closure retains every transition required for observability. With a rational right inverse $J_{\rm obs}$, the observable quotient has the explicit numerical matrices
\[
 N''_j=Q_{\rm obs}N'_jJ_{\rm obs},\qquad
 b''_j=Q_{\rm obs}b'_j,\qquad O''=O'J_{\rm obs}.
\]
Invariance of $\ker Q_{\rm obs}$ under every $N'_j$ proves that these arrays preserve each coefficient in Eq.~\eqref{eq:row-state-coeff}. Their joint output Gram is again rational, while the entry evaluator for $O''$ is obtained by multiplying the supplied evaluator by the known rational matrices. Retaining basis rows and columns from transition products of bounded length controls all encoding lengths by exact elimination and minor bounds. This gives the reduction required before Eq.~\eqref{eq:row-Gram}, including an output matrix whose entries contain square roots.

For a rational feasible point of Theorem~\ref{thm:fixed-query}, we factor each $G_t$ by its full positive square root. Its post-query Gram $B_t^\dagger B_t$ is a rational linear function of $G_{t-1}$. We compute its Moore--Penrose inverse on a rationally determined range by exact elimination and use
$B_t^+=(B_t^\dagger B_t)^+B_t^\dagger$.
Nonzero rational minors bound all inverse singular values with polynomial logarithms in the explicit certificate size, and the continuous square-root procedure below evaluates the remaining factors. Thus singular supplied certificates have polynomial-bit extraction in their actual input length; this bound does not provide an efficient way to discover them from a compact target.

\subsection{Certified evaluation of numerical matrices}
For positive semidefinite matrices, the elementary continuity bound
\begin{equation}\label{eq:root-continuity}
 \|A^{1/2}-B^{1/2}\|\le\|A-B\|^{1/2}
\end{equation}
requires no positive eigenvalue gap. To verify it, we set $\eta=\|A-B\|$ and use $A\preceq B+\eta I$ and monotonicity of the square root. The latter follows from its positive integral representation
$A^{1/2}=(2/\pi)\int_0^\infty A(t^2I+A)^{-1}\,\dd t$
and order reversal for positive inverses. Since $(B+\eta I)^{1/2}\preceq B^{1/2}+\sqrt\eta I$, we obtain an upper bound on $A^{1/2}-B^{1/2}$; interchanging $A,B$ gives the lower bound. This is the usual matrix-function estimate~\cite{higham2008}, included here because the Julia defects may have zero eigenvalues even when the original certificate is strict.

A Hermitian rational approximation $A_q$ within $\eta$ of $A\succeq0$ can be made positive by adding $2\eta I$. Its square root differs from the ideal one by at most $\sqrt{3\eta}$. We choose this regularization of order the square of the target root accuracy, so it pays only for numerical evaluation; the ideal unitary is still the Julia matrix of the original contraction. For a positive definite rational matrix with a certified spectrum in $[\lambda,M]$, we use the standard Newton--Schulz inverse-square-root iteration, related to the matrix-sign iterations discussed in Higham~\cite[Chapter~6, Section~6.7]{higham2008}. We increase $M$ to a rational square and set
\begin{equation}\label{eq:Newton}
 Y_0=M^{-1/2}I,
 \qquad Y_{t+1}=\tfrac12Y_t(3I-AY_t^2).
\end{equation}
The exact iterates commute with $A$. At an eigenvalue $a$, the scalar error $e_t=1-ay_t^2$ obeys $e_{t+1}=e_t^2(3+e_t)/4$ and remains in $[0,1)$. While $ay_t^2\le1/2$, the quantity $ay_t^2$ increases by at least a fixed factor greater than one. After $O(\log(M/\lambda))$ iterations, the scalar error is bounded away from one, and the subsequent convergence is quadratic. Multiplication of an inverse-root approximation by $A$ gives its positive root. Thus only polynomially many operations in logarithmic condition and precision parameters are needed.

The finite-bit adaptation needed here is an explicit guard-precision bound. Dyadic rounding of Eq.~\eqref{eq:Newton} is controlled in operator norm, even though rounded iterates need not commute with $A$. The exact iterates obey $\|Y_t\|\le\lambda^{-1/2}$. On a fixed neighborhood of this ball, the matrix polynomial update has a Lipschitz constant polynomial in $M$ and $\lambda^{-1}$. A guard precision equal to the number of iterations times the logarithm of this bound, plus the required accuracy bits and matrix-size conversion factors, bounds the total rounding error. This is polynomial in the stated logarithmic parameters. Combined with Eq.~\eqref{eq:root-continuity}, it supplies certified root evaluators at singular matrices as well. For inverse roots, the algorithm uses only the explicit positive lower bounds proved for $T_0$, $T$, or $K_t$; it never attempts an inverse on an uncertified numerical support.

We apply these estimates to the expressions occurring in our constructions. Roots of supplied rational Gram matrices can be evaluated directly; the inverse roots in the output recurrences act on positive definite rational matrices with the minor bounds proved above. The recurrences then use a polynomial number of products, whose magnitude and Lipschitz bounds give polynomial guard precision. The possibly singular roots in the resulting Julia gates are evaluated after those contractions have been formed: Eq.~\eqref{eq:root-continuity} requires the radicands at squared accuracy, which doubles the requested precision bits. The terminal effect and its Julia completion in the comb construction add only a fixed number of such root evaluations. Allocating accuracy in this order, and using the product and inverse perturbation inequalities at the intervening operations, gives the entry algorithms used in Theorems~\ref{thm:gates} and~\ref{thm:row-compiler}.

\section{Circuit implementation and finite computation}\label{app:unitary-details}

\subsection{The supplied query unitary}
We specify the coordinates used by the contraction-product implementation and then give its packing and finite-gate calculations. The signal register $\mathsf S$ has dimension $n=2^{n_s}$ and the oracle register $\mathsf O$ has dimension $d_o=2^{q_o}$. We use the full query matrix in Eq.~\eqref{eq:oracle-full}.
Its three block rows have lengths $in$, $an$, and $(d_o-i-a)n$, and its three block columns have lengths $in$, $bn$, and $(d_o-i-b)n$. The off-diagonal blocks consequently have dimensions
\[
 B_A:\ an\times (d_o-i-b)n,\qquad
 C_A:\ (d_o-i-a)n\times bn,
\]
while $D_A$ has $(d_o-i-a)n$ rows and $(d_o-i-b)n$ columns. The identity columns already have norm one, forcing the adjoining zero blocks by unitarity. The construction must work for every compatible choice of $B_A,C_A,D_A$.

For this promise to hold throughout the input domain, the oracle needs $d_o\ge i+a+b$ labels. Indeed, at the zero tuple all $bn$ active input columns must be accommodated in the $(d_o-i-a)n$ unspecified output rows, which gives this lower bound. To see that it suffices, we write $d_{\mathrm{pad}}=d_o-i-a-b$ and complete the designated block by the square matrix
\begin{equation}\label{eq:oracle-existence}
 \begin{pmatrix}
 I_{in}&0&0&0\\
 0&\Delta(A)&(I_{an}-\Delta(A)\Delta(A)^\dagger)^{1/2}&0\\
 0&(I_{bn}-\Delta(A)^\dagger\Delta(A))^{1/2}&-\Delta(A)^\dagger&0\\
 0&0&0&I_{d_{\mathrm{pad}}n}
 \end{pmatrix}.
\end{equation}
The block row lengths are $in,an,bn,d_{\mathrm{pad}}n$, and the block column lengths are $in,bn,an,d_{\mathrm{pad}}n$; if $d_{\mathrm{pad}}=0$, we omit the final row and column. The middle block is a rectangular Julia completion and is unitary by the defect identities in Robinson~\cite[Theorems~0--1]{robinson2018}. This verifies that a valid input oracle exists at every allowed tuple. The query supplied to our algorithm remains the full $U_A$ in Eq.~\eqref{eq:oracle-full}, whose action includes the unspecified blocks.

The input coordinates $0,\ldots,i+b-1$ and output coordinates $0,\ldots,i+a-1$ are the designated coordinates of Eq.~\eqref{eq:visible-block}. Numerical matrices act on the signal by
\begin{equation}\label{eq:amplify}
 C\otimes I_n=[C_{\mu\nu}I_n]_{\mu,\nu}.
\end{equation}
The rightmost factor in a matrix product acts first. In particular, the first physical query producing $A_{j_1}\cdots A_{j_d}$ supplies the rightmost letter $A_{j_d}$.

Arbitrary finite signal sizes are covered by embedding a tuple into larger zero matrices of power-of-two order. Its polynomial block and contraction promise are retained, so the uniform norm in Eq.~\eqref{eq:SA-norm} agrees with the norm needed by the qubit model.

\subsection{Registers and reversible failure markers}
Theorem~\ref{thm:unitary} proves the selected-path invariant. Here we specify all columns of its gates, including those reached after failure. A retained query coordinate is converted to its physical oracle label by
\begin{equation}\label{eq:retained-labels}
 \begin{aligned}
 \sigma_-(u)&=u\quad(0\le u<\kappa),&
 \sigma_-(\kappa+\nu)&=i+\nu\quad(0\le\nu<b),\\
 \sigma_+(u)&=u\quad(0\le u<\kappa),&
 \sigma_+(\kappa+\mu)&=i+\mu\quad(0\le\mu<a).
 \end{aligned}
\end{equation}
The scalar algorithms retain only $\kappa\le2$ identity labels; the fixed-query characterization may retain all $i$. Thus the synthesis arrays depend on the selected lists and their $q_o$-bit labels, even if the full identity sector is exponentially large in $q_o$.

We use the Julia matrix in Eq.~\eqref{eq:Julia}, whose unitarity follows from the defect-intertwining identity~\cite[Theorems~0--1]{robinson2018}. Its input and output dimensions may differ, so we place its two ordered coordinate lists on a common physical support before defining the gate.

We store the processing index in a register $\mathsf M$ of $q_M=\lceil\log_2R\rceil$ qubits, with $R_b=2^{q_M}$ basis states, and add one dilation qubit $\mathsf d$. Both start at zero; $\ket0$ on a named register denotes its full all-zero string, and a register of width zero is omitted. The value $\mathsf d=0$ identifies the required block of a Julia matrix, while $\mathsf d=1$ holds its complementary output. For a contraction $C$, the input list consists of its $h_1$ designated oracle--memory labels with $\mathsf d=0$, followed by the first $h_2$ oracle--memory labels with $\mathsf d=1$; the output list uses its $h_2$ designated output labels followed by the first $h_1$ labels with $\mathsf d=1$. The designated coordinates are converted to physical strings using $\sigma_+$ or $\sigma_-$ as appropriate. Each list therefore has $h_1+h_2$ distinct strings of length $q_o+q_M+1$, and their union has size $s_C\le2(h_1+h_2)$. The complementary strings label the reverse columns of the known unitary on the same signal register.

We order this union once, complete each list by the unused strings in that order, and let $P_{\rm in},P_{\rm out}$ be the resulting $s_C\times s_C$ permutation matrices. Their entries are zeros and ones, determined entirely by the lists. The numerical matrix to be applied on this support is
\begin{equation}\label{eq:support-unitary}
 \widetilde J_C=P_{\rm out}
 \begin{pmatrix}J(C)&0\\0&I_{s_C-h_1-h_2}\end{pmatrix}
 P_{\rm in}^\dagger.
\end{equation}
The gate $J_C$ acts by $\widetilde J_C$ on the union and by identity on its orthogonal complement. This specifies every column using only the supported matrix and its list of bit strings.

For $C_0$, the selected input is the single ancillary label with oracle and memory both zero. For an internal $C_t$, its inputs are the designated output labels of the preceding query with memory $0,\ldots,R-1$. For $t<Q$, its selected outputs are the designated next-query input labels with the same live memory range. The final $C_Q$ outputs the single oracle--memory label zero. Consequently every support in Eq.~\eqref{eq:support-unitary} has size at most
\begin{equation}\label{eq:support-bound}
 s_C\le S:=2(d_-+d_+)R=2(2\kappa+a+b)R.
\end{equation}
These scalar endpoints will be replaced by the data-register endpoints in Section~\ref{sec:block-register}.

To keep every unwanted output distinguishable from the successful component at all later stages, we add a failure register $\mathsf E$ of $q_E=\lceil\log_2(2Q+1)\rceil$ qubits, initialized at zero. Its nonzero used values are $1,\ldots,2Q$; the remaining binary values are fixed by all known markers. A known gate $G_t$ applies $J_{C_t}$ when $\mathsf E=0$ and otherwise applies identity. We also use the following explicitly reversible basis permutations. The gate $D_t$, for $0\le t<Q$, exchanges $\mathsf E=0$ and $\mathsf E=2t+1$ precisely when $\mathsf d=1$. The gate $F_t$, for $1\le t\le Q$, exchanges $\mathsf E=0$ and $\mathsf E=2t$ precisely when $\mathsf d=0$ and the oracle output label is not in $\{\sigma_+(0),\ldots,\sigma_+(d_+-1)\}$. Every other basis state is fixed by the respective permutation. These definitions give the complete matrices of $G_t,D_t,F_t$: for $D_t$ and $F_t$, a matrix entry is one exactly when its row string is the stated permutation of its column string, and it is zero otherwise.

These conditional transpositions are reversible permutations of auxiliary basis states. They do not control the unknown query.

The query $\widehat U_A$ acts on $\mathsf O,\mathsf S$ for every value of $\mathsf M,\mathsf d,\mathsf E$. For matrix indexing we order these registers as $(\mathsf O,\mathsf M,\mathsf d,\mathsf E,\mathsf S)$ and use $(\alpha,\mu,d,e,x)$ for their basis values, where $0\le\alpha<d_o$, $0\le\mu<R_b$, $d\in\{0,1\}$, $0\le e<2^{q_E}$, and $0\le x<n$. Its full matrix entry is
\begin{equation}\label{eq:raw-query-entries}
 [\widehat U_A]_{(\alpha,\mu,d,e,x),(\beta,\nu,d',e',y)}
 =\delta_{\mu\nu}\delta_{dd'}\delta_{ee'}
                         [U_A]_{(\alpha,x),(\beta,y)}.
\end{equation}
These coordinate definitions give the known layers $W_0,\ldots,W_Q$ of Eq.~\eqref{eq:full-circuit}.

The chronological order is $G_0,D_0,\widehat U_A,F_1,G_1,D_1,\ldots,
\widehat U_A,F_Q,G_Q$, as shown for two queries in Figure~\ref{fig:whole-circuit}. The final selection requires $\mathsf d=0$, so a final marker $D_Q$ is unnecessary. The fresh-label argument in Theorem~\ref{thm:unitary} applies to the coordinate lists just specified.

The selected columns of each $J_C$ have exactly the prescribed contraction block. Eq.~\eqref{eq:raw-query-entries} preserves all failure labels, and the explicit transpositions use a fresh label at every step. These are the hypotheses of the induction in Theorem~\ref{thm:unitary}; applying it to the displayed coordinate lists gives Eq.~\eqref{eq:contraction-product}. The final selection rejects the last Julia complement without an additional marker. There are $Q+1$ Julia gates and $2Q$ markers, as counted in that theorem.

\subsection{Query coordinates and compact factor computation}\label{app:circuit-details}
We now place the updates of Section~\ref{sec:general} on the physical coordinates above. The same packing convention applies to both the finite sequence and the row factors.

\subsubsection{The window updates and their oracle coordinates}\label{app:circuit-packing}
We retain the $\kappa$ identity labels and the active labels of Eq.~\eqref{eq:retained-block}, at the physical offsets in Eq.~\eqref{eq:retained-labels}. On $R$ processing labels, the query acts in the order (oracle label, memory label, signal coordinate) with entries
\begin{equation}\label{eq:query-entries}
 [\Lambda_R(A)]_{(\alpha,\mu,x),(\beta,\nu,y)}
       =\delta_{\mu\nu}[L_\kappa(A)]_{(\alpha,x),(\beta,y)}.
\end{equation}
Here $0\le\alpha<d_+$, $0\le\beta<d_-$, and $0\le\mu,\nu<R$. For the window construction, $\kappa=\min\{i,2\}$ and $R$ is given in Eq.~\eqref{eq:R-general}.

We prove that $M_t$ is contractive by splitting its action into input preparation, application of $V$, and output collection. The corresponding matrices are
\begin{equation}\label{eq:Mt-factors}
 \begin{aligned}
 S_t&=\begin{pmatrix}
 I_{am}&0&0\\0&\alpha_t&0\\0&\zeta_t&0\\0&0&1
 \end{pmatrix},\qquad
 T_t=\begin{pmatrix}
 I_{bm}&0&0&0\\0&0&1&0\\0&\xi_t&0&\gamma_t
 \end{pmatrix},\\
 V_{\rm ext}&=\begin{pmatrix}
 \Aco&\Bco&0&0\\
 \Cco&\Dco&0&0\\
 0&0&1&0\\
 0&0&0&1
 \end{pmatrix},\qquad M_t=T_tV_{\rm ext}S_t.
 \end{aligned}
\end{equation}
The sizes of $S_t,V_{\rm ext},T_t$ are $(am+3)\times (am+2)$, $(bm+3)\times (am+3)$, and $(bm+2)\times (bm+3)$. The first factor on an input splits the pending component with weights $\alpha_t,\zeta_t$; the middle matrix applies $V$ to the active and injected components; the last matrix combines the new output with the previous accumulated output. Since $S_t^\dagger S_t=I_{am+2}$, $T_tT_t^\dagger=I_{bm+2}$, and $\|V_{\rm ext}\|\le1$, their product is a contraction. We use this factorization to prove contractivity, then apply the Julia construction to the resulting query matrix.

The endpoint selectors in Eq.~\eqref{eq:logical-query-product} specify the initial pending input and final collected output. We now verify their embedding on the oracle and memory labels.

The choice of $R$ provides at least $m$ labels for each active input and at least two coordinates in the identity part. We specify their placement by zero--one numerical matrices. The matrix $J_-$ has $(\kappa+b)R$ rows and $bm+2$ columns. Its nonzero entries are
\begin{equation}\label{eq:packing-minus}
 [J_-]_{(\kappa+\nu)R+\mu,\,\nu m+\mu}=1
 \quad(0\le\nu<b,\ 0\le\mu<m),
 \qquad
 [J_-]_{q,\,bm+q}=1\quad(q=0,1).
\end{equation}
The latter row numbers $q$ lie among the first $\kappa R$ idle coordinates, because $\kappa R\ge2$. With $a$ replacing $b$, the matrix $J_+$ has $(\kappa+a)R$ rows and $am+2$ columns and nonzero entries
\begin{equation}\label{eq:packing-plus}
 [J_+]_{(\kappa+\nu)R+\mu,\,\nu m+\mu}=1,
 \qquad [J_+]_{q,\,am+q}=1\quad(q=0,1).
\end{equation}
Every unlisted entry is zero. Distinct columns select distinct ancillary labels, so $J_-^\dagger J_-=I_{bm+2}$ and $J_+^\dagger J_+=I_{am+2}$.

In the block--memory--signal order, the repeated input of Eq.~\eqref{eq:delta-amplification} has entries
\begin{equation}\label{eq:delta-memory}
 [\Delta_m(A)]_{(\mu,r,x),(\nu,s,y)}
 =\delta_{rs}\sum_j f_{\mu\nu,j}(A_j)_{xy},\qquad
 \Delta_m(A):\ amn\times bmn.
\end{equation}
Here $1\le\mu\le a$, $1\le\nu\le b$, $1\le r,s\le m$, and $0\le x,y<n$. Substituting these entries into Eq.~\eqref{eq:query-entries} gives
\begin{equation}\label{eq:packed-query}
 (J_+^\dagger\otimes I_n)\Lambda_R(A)(J_-\otimes I_n)
  =\mathcal B_m(A)=
  \begin{pmatrix}
  \Delta_m(A)&0&0\\0&I_n&0\\0&0&I_n
  \end{pmatrix}.
\end{equation}
The right side has $(am+2)n$ rows and $(bm+2)n$ columns. Thus the active calculation and both identity components are contained in the same designated block of a single unconditional query.

Substitution of these packing matrices into Eq.~\eqref{eq:Ct} gives the query parameters used in the main construction. Their respective sizes are $((\kappa+b)R)\times1$, $(\kappa+b)R\times (\kappa+a)R$, and $1\times((\kappa+a)R)$. All are contractions. Eqs.~\eqref{eq:packing-minus} and~\eqref{eq:packing-plus}, together with Eq.~\eqref{eq:Ct}, determine every numerical matrix and every designated label needed by the complete unitary circuit Eq.~\eqref{eq:full-circuit}.

\subsubsection{Elementary gates and their error budget}
\label{sec:finite-gates}
The classical routines supply certified evaluators for every known contraction. Their supported Julia matrices act on bit strings of ideal width
\begin{equation}\label{eq:n0}
 n_0=q_o+\lceil\log_2 R\rceil+1+\lceil\log_2(2Q+1)\rceil.
\end{equation}
This counts the oracle, processing, dilation and failure registers. Compilation adds one phase qubit $\mathsf P_{\rm ph}$ and at most $n_0+1$ clean work qubits $\mathsf W$, initialized and selected at zero. The same work is reused after each predicate is uncomputed, and the oracle is identity on it. Thus $q_{\rm anc}$ includes the supplied $q_o$ oracle qubits. The retained output test uses $d_+$ comparisons of $q_o$-bit strings; since $\kappa\le2$, no exponentially large identity sector is enumerated.

\begin{theorem}[Elementary-gate synthesis]\label{thm:gates}
For $Q\ge1$, let the contractions in Eq.~\eqref{eq:contraction-product}
have entry evaluators satisfying the convention in Section~\ref{sec:classical}. For $0<\epsilon_g<1$, their full unitary circuit has a Clifford+$T$ implementation with the same $Q$ original queries, at most $2n_0+2$ ancillary qubits other than the signal, and
\begin{equation}\label{eq:gate-cost}
 O\left(QS^2n_0^3\log^4\left(2+\frac{QS^2n_0^3}{\epsilon_g}\right)\right)
\end{equation}
known elementary gates. The selected-block error is at most $\epsilon_g$ uniformly over all signal sizes and allowed oracle completions. Classical compilation is polynomial in the evaluator bounds, $Q,S,q_o$, and $\log(1/\epsilon_g)$.
\end{theorem}
\begin{proof}
Appendix~\ref{app:gates} compiles a unitary on $S$ listed states into $O(S^2)$ two-level rotations, preserving its phase through a determinant-one extension. Known control patterns cost $O(n_0^3)$ gates per rotation and linear clean space. Phase-sensitive Solovay--Kitaev synthesis costs $O(\log^4(1/\zeta))$ for each one-qubit tolerance $\zeta$~\cite[Theorem~1]{dawson2006}; the appendix verifies certified pivots and polynomial bit bounds for our evaluator input. Applying this to the $Q+1$ Julia gates and compiling the markers exactly gives Eq.~\eqref{eq:gate-cost} when the tolerances sum to $\epsilon_g$. The phase and work qubits fit within $2n_0+2$ qubits. Each replacement compares full unitaries with the same fixed oracle calls, so the operator-norm telescope bounds the circuit error, including work-register leakage, by $\epsilon_g$. Taking the selected block preserves that bound.
\end{proof}

\subsubsection{A degree-reducing factor as one row query}\label{app:row-packing}
The polynomial equations in Eq.~\eqref{eq:peel-recover} determine a physical row-query factor. We derive its numerical input and output matrices with their distinct label counts, so that the factor fits the square-unitary compiler of Section~\ref{sec:unitary-implementation}. Let $R_*$ be a common memory size at least the current $q$, and let
\[
 B_Y=\begin{pmatrix}I_q\\0_{(R_*-q)\times q}\end{pmatrix}.
\]
The packing matrix $E_Z$ has $(k+1)R_*$ rows and $q+ke$ columns. Its block rows are the designated oracle input labels, each of length $R_*$. Its only nonzero blocks are $B_Y(I-\Pi)$ in block row zero and the first $q$ columns, and $B_YZ$ in block row $j$ and the $j$th $e$-column block, for $1\le j\le k$. The output matrix has $q$ rows and $2R_*$ columns and is
\begin{equation}\label{eq:row-packing}
 E_Z=
 \begin{pmatrix}
 B_Y(I-\Pi)&0&0&\cdots&0\\
 0&B_YZ&0&\cdots&0\\
 0&0&B_YZ&\cdots&0\\
 \vdots&&&\ddots&\vdots\\
 0&0&0&\cdots&B_YZ
 \end{pmatrix},
 \qquad
 R_Z=[\,(I_q-\Pi)B_Y^\dagger\ \ \Pi B_Y^\dagger\,].
\end{equation}
The block rows are exactly the retained identity label followed by the $k$ active labels; each row has $R_*$ coordinates. Both matrices therefore use a single identity label. Matrix multiplication gives
\[
 E_Z^\dagger E_Z=\begin{pmatrix}I_q-\Pi&0\\0&I_{ke}\end{pmatrix},
 \qquad R_ZR_Z^\dagger=I_q.
\]
Thus both are contractions. Let $\widehat F_0$ be the $2\times(k+1)$ matrix whose only nonzero entry is one at row $0$, column $0$, and let $\widehat F_j$ have its only nonzero entry one at row $1$, column $j$, using zero-based indices. Multiplication of the known matrices gives
\begin{equation}\label{eq:row-factor-coeff}
 R_Z\widehat F_0^{(R_*)}E_Z=B_0,
 \qquad R_Z\widehat F_j^{(R_*)}E_Z=B_j\quad(1\le j\le k).
\end{equation}
The superscript denotes replacement of each numerical entry by that
entry times $I_{R_*}$. Therefore one designated row query between $E_Z$ and $R_Z$ implements the coefficient step Eq.~\eqref{eq:peel-recover}. Its full unitary realization is already supplied by Theorem~\ref{thm:unitary}.

After $d$ reductions, let $E_t,R_t^{\rm out}$ denote the packing and output matrices at reduction stage $t$, and let $d_d$ be the final constant column. With $e_0^{\mathsf T}$ selecting the first polynomial of the original list, the actual query matrices in chronological order are
\begin{equation}\label{eq:row-query-matrices}
 C_0=E_{d-1}d_d,\qquad
 C_t=E_{d-1-t}R_{d-t}^{\rm out}\quad(1\le t<d),\qquad
 C_d=e_0^{\mathsf T}R_0^{\rm out}.
\end{equation}
They have sizes $((k+1)R_*)\times1$, $(k+1)R_*\times 2R_*$, and $1\times(2R_*)$. Eq.~\eqref{eq:peel-recover} and Eq.~\eqref{eq:row-factor-coeff} show coefficientwise that their designated product equals $p(A)/\beta$. In particular, the last factor found by the classical reduction is the first applied in the quantum circuit. The full unitary construction preserves this chronological ordering coherently.

\subsubsection{Computing the row factorization from the target}\label{app:row-compact}
We give the exact state calculations and precision bounds used in Theorem~\ref{thm:row-compiler}. The input convention is Eq.~\eqref{eq:row-state-coeff}; its rational joint Gram allows the reachable and observable reduction of Appendix~\ref{app:arithmetic}. After reduction, nilpotence bounds every recurrence by the degree $d$. In particular, Eq.~\eqref{eq:row-Gram} computes $H$ without enumerating words. Direct expansion of the correlations gives the rational innerness test
\begin{equation}\label{eq:row-inner-test}
 d_0^\dagger d_0+\sum_jb_j^\dagger Hb_j=1,
 \qquad O^\dagger d_0+\sum_jN_j^\dagger Hb_j=0.
\end{equation}
For clarity, the correlation at a nonempty word $vj$ is
$(d_0^\dagger O+\sum_\ell b_\ell^\dagger HN_\ell)N_vb_j$.
Reachability shows that all these correlations vanish exactly when the second equality in Eq.~\eqref{eq:row-inner-test} holds; the first is the empty correlation. The input condition is therefore decided by exact rational matrix equalities in the reduced coordinates.

To identify the successive highest-degree parts, we use the state coordinates that remain reachable after a given number of transitions. Their spans shrink with that number because the transitions are jointly nilpotent. We initialize $U_1=I_{s_0}$ and compute each rational column-basis matrix $U_{\ell+1}$ by retaining independent columns from $[N_1U_\ell\ \cdots\ N_kU_\ell]$. Their column spans are nested, and $U_{d+1}$ has no columns. At stage $t$, we choose a rational matrix $R_t$ so that $[U_{d-t+1}\ R_t]$ is a column basis of $U_{d-t}$, and write $e_t$ for the number of added columns. The directions in $U_{d-t+1}$ have already been removed. We obtain the remaining response energy by minimizing the original $H$-quadratic form over those directions, which gives $H_t$ in Eq.~\eqref{eq:short-Gram}; the subtracted term is zero when the eliminated basis is empty. This is the finite-dimensional shorted matrix associated with Anderson and Trapp's construction~\cite{andersontrapp1975}. We need its variational identity, with $U=U_{d-t+1}$,
\[
 x^\dagger H_tx=\min_z (x-Uz)^\dagger H(x-Uz),\qquad
 z_{\min}=(U^\dagger HU)^{-1}U^\dagger Hx.
\]
It follows by completing the square; $U^\dagger HU$ is positive definite whenever $U$ is nonempty. In particular, $H_t\succeq0$, its kernel is the column span of $U$, and minimizing successively over nested spans agrees with minimizing directly over the larger span. These are the properties used in the degree-reduction recurrence. This formula computes each metric directly from the original rational data.

Starting with the output pair $O_0=O,d_0$, we apply the updates in Eq.~\eqref{eq:row-updates}, keeping $N_j,b_j$ fixed.
If $e_t=0$, the corresponding inverse-root and variable blocks are absent. In the nonempty case $K_t$ is positive definite by Theorem~\ref{thm:row-compiler}. If $O_t$ has $q_t$ rows, then $Z_t$ has $q_t$ rows and $e_t$ columns, and $O_{t+1},d_{t+1}$ have $q_{t+1}=q_t+ke_t$ rows. The state dimension $s_0$ is unchanged.

The invariant and range argument in Theorem~\ref{thm:row-compiler} prove that these are the coefficient factors of Lemma~\ref{lem:peel}, with $\sum_te_t=s_0$ and common processing multiplicity $q_0+ks_0$. It remains to justify their numerical evaluation.

\begin{proof}[Bit bounds for Theorem~\ref{thm:row-compiler}]
To bound the precision required by these formulas, we use their rational state coordinates. The basis columns can be retained as products of the original rational transitions of length at most $d$, and all $H_t$ are computed directly from the original rational $H$ in Eq.~\eqref{eq:short-Gram}. Their entries, and those of the positive definite $K_t$, are rational with polynomial bit length by common-denominator and minor bounds. Their positive eigenvalues are therefore bounded below by $2^{-\poly(L)}$, with a computable polynomial in the rational input length $L$. The inverse roots in Eq.~\eqref{eq:row-updates} have certified polynomial-time evaluators. The rank choices are determined by the rational basis calculations before we evaluate $O_t$ and $Z_t$.

To control the precision needed in this recurrence, we express the joint output Grams in terms of rational data. The observability equation and the preserved coefficient correlations give
\begin{equation}\label{eq:joint-row-Grams}
 \begin{aligned}
 O_t^\dagger O_t&=H_t-\sum_jN_j^\dagger H_tN_j,\\
 O_t^\dagger d_t&=-\sum_jN_j^\dagger H_tb_j,\\
 d_t^\dagger d_t&=1-\sum_jb_j^\dagger H_tb_j
 \end{aligned}
\end{equation}
These identities give rational formulas for all joint output Grams and bound the sizes of the matrices evaluated during the recurrence. Moreover, $O_t^\dagger O_t\preceq H_t\preceq H$ and $\|d_t\|\le1$. A polynomial number of guarded matrix operations consequently evaluates Eq.~\eqref{eq:row-updates} to the required precision. The finite-gate compiler then implements Eq.~\eqref{eq:row-query-matrices} with its actual numerical support sizes and uniformly bounded selected error.
\end{proof}

\subsubsection{An explicit row circuit}
\begin{example}[An explicit two-query scalar transformation]\label{ex:row}
For $i=1$, $a=1$, $b=k=2$, the original oracle can have $d_o=4$ ancillary labels. In known input and output orders, its prescribed entries are
\begin{equation}\label{eq:row-oracle-example}
 U_A=\begin{pmatrix}
 I_n&0&0&0\\
 0&A_1&A_2&B\\
 0&C_{11}&C_{12}&D_1\\
 0&C_{21}&C_{22}&D_2
 \end{pmatrix},
 \qquad \|[A_1\ A_2]\|\le1,
\end{equation}
where each displayed entry is an $n\times n$ matrix and all unspecified blocks are constrained only by unitarity. We implement
\[
 p(A_1,A_2)=\tfrac35 A_1A_2+\tfrac45 A_2A_1.
\]
With two memory labels, the designated query input has six numerical coordinates, ordered $(0,0),(0,1),(1,0),(1,1),(2,0),(2,1)$, and its output has four, ordered $(0,0),(0,1),(1,0),(1,1)$. The query matrices are explicitly
\begin{equation}\label{eq:row-example-matrices}
 C_0=\begin{pmatrix}0\\0\\0\\4/5\\3/5\\0\end{pmatrix},\qquad
 C_1=\begin{pmatrix}
 0&0&0&0\\0&0&0&0\\0&0&1&0\\0&0&0&0\\0&0&0&1\\0&0&0&0
 \end{pmatrix},\qquad
 C_2=\begin{pmatrix}0&0&1&0\end{pmatrix}.
\end{equation}
Their nonzero singular values are one, so their Julia defect roots are just projections. Direct multiplication of the designated blocks gives
\[
 (C_2\otimes I_n)\Lambda_2(A)(C_1\otimes I_n)\Lambda_2(A)(C_0\otimes I_n)
                  =\tfrac35 A_1A_2+\tfrac45 A_2A_1.
\]
On one branch the first query contributes $A_2$ and the second contributes $A_1$; the other branch gives the reverse order. The full circuit uses the same two unconditional instances of Eq.~\eqref{eq:row-oracle-example}, with the Julia gates and failure permutations specified in Section~\ref{sec:unitary-implementation}. The preparation uses the known rotation
$\begin{pmatrix}3/5&-4/5\\4/5&3/5\end{pmatrix}$ and basis permutations, while $C_1$ is the displayed routing matrix. The selected target has norm at most one by the contraction product, and norm exactly one by evaluation on the word matrices truncated at degree two, which send the empty word to orthogonal words $12$ and $21$ with coefficients $3/5,4/5$. Two exact queries are therefore optimal in this access model. If the row is obtained from equally weighted selection of two separate block encodings, it is divided by $\sqrt2$, and the same circuit returns $p/2$; the loading normalization is therefore part of the resulting block.
\end{example}

\subsection{Finite-gate synthesis of the known unitaries}\label{app:gates}
Theorem~\ref{thm:gates} requires a phase-sensitive compiler for a known unitary supported on at most $S$ listed basis states, whose entries have certified evaluators. Standard two-level synthesis and Solovay--Kitaev supply the gate decompositions~\cite{barenco1995,shende2006,dawson2006}; the points specific to our input are the support size, the phase of a controlled supported matrix, and the bit cost of numerical choices. We verify these points here. Our elementary set is $H,T,T^\dagger$, and CNOT, with exact NOT and Toffoli decompositions~\cite{selinger2013}.

A phase error in a gate applied only at failure label zero is a relative phase in the full circuit. We therefore add one phase qubit and one unused basis state with phase value one. For a supported unitary $U$ of order $s\le S$, we compile
\begin{equation}\label{eq:det-comp}
 \widehat U=\begin{pmatrix}U&0\\0&\overline{\det U}\end{pmatrix}
 \in\operatorname{SU}(s+1).
\end{equation}
It acts as $U$ on the phase-zero input and fixes every unlisted state. Two-level elimination reduces it to $O(S^2)$ matrices in $\operatorname{SU}(2)$. The elementary elimination is
\begin{equation}\label{eq:Givens}
 \frac1{\sqrt{|a|^2+|b|^2}}
 \begin{pmatrix}\overline a&\overline b\\-b&a\end{pmatrix}.
\end{equation}
After earlier columns have become coordinate columns, the remaining part of the next column has norm one. Certified intervals identify a pivot of magnitude at least $1/(2\sqrt{s+1})$; a signed two-level permutation moves it to the pivot position, and subsequent eliminations only increase its accumulated magnitude. Thus every denominator has a polynomial lower bound. The first $s$ diagonal entries become one, and determinant one fixes the last. This gives entry evaluators for all factors without deciding whether an approximately known entry is exactly zero.

For two listed bit strings, a Gray path conjugates a controlled one-qubit operation into the corresponding two-level gate~\cite{barenco1995}. Each path transposition and the failure-zero condition use known Boolean predicates: an AND chain computes the required bit pattern in clean workspace, the gate is applied, and the computation is reversed. Visibility markers additionally compare the oracle label with the retained list of $d_+$ strings. These operations use polynomial work in the label width, with the conservative bound $O(n_0^3)$ per two-level factor and linear clean space used in Theorem~\ref{thm:gates}. They never control the unknown oracle.

Barenco et al.'s Lemmas~4.3 and~5.1~\cite{barenco1995} give a controlled $W\in\operatorname{SU}(2)$ using two CNOTs and three one-qubit matrices satisfying
\begin{equation}\label{eq:ABC}
 ABC=I,\qquad AXBXC=W.
\end{equation}
We use a stable coordinate construction of these matrices instead of choosing Euler angles near a singular chart. With the Pauli matrices $\boldsymbol\sigma=(X,Y,Z)$, define
\begin{equation}\label{eq:quaternion-matrix}
 \mathcal Q(q_0,q_x,q_y,q_z)=
 \begin{pmatrix}q_0-iq_z&-q_y-iq_x\\q_y-iq_x&q_0+iq_z\end{pmatrix}.
\end{equation}
For a real unit vector $(q_0,q_x,q_y,q_z)$ this is an $\operatorname{SU}(2)$ matrix. Writing $W=cI-i\mathbf v\cdot\boldsymbol\sigma$, we first test whether it lies within the assigned tolerance of $I$ or $-I$. In that case we use the corresponding central matrix and charge the error; otherwise certified overlapping tests give a lower bound on $\|\mathbf v\|$ in terms of the tolerance. For real unit vectors $x,y$ with $1+x\cdot y>0$, we use the alignment matrix
\[
 H(x,y)=\mathcal Q\left(\frac{1+x\cdot y}{\sqrt{2(1+x\cdot y)}},
                        \frac{x\times y}{\sqrt{2(1+x\cdot y)}}\right).
\]
Pauli multiplication gives $H(x,y)(x\cdot\boldsymbol\sigma)H(x,y)^\dagger=y\cdot\boldsymbol\sigma$. Choosing a coordinate axis $e_j$ with $|e_j\cdot\mathbf v/\|\mathbf v\||\le3/4$ therefore gives
\[
 A=H(e_j,\mathbf v/\|\mathbf v\|)H(e_z,e_j),\quad
 B=\diag\left(\sqrt{\frac{1+c}{2}}+i\sqrt{\frac{1-c}{2}},
              \sqrt{\frac{1+c}{2}}-i\sqrt{\frac{1-c}{2}}\right),
 \quad C=B^\dagger A^\dagger.
\]
Both alignment denominators are bounded away from zero, and $1\pm c\ge\|\mathbf v\|^2/2$ bounds the root conditions. Since $XBX=B^\dagger$ and $(B^\dagger)^2=cI-i\|\mathbf v\|Z$, these factors obey Eq.~\eqref{eq:ABC}. At $W=-I$ we use $A=I,B=iZ,C=-iZ$. All choices use overlapping tolerance intervals, so no exact comparison of evaluator-defined entries is required.

For the one-qubit approximations we use the inverse-closed alphabet
\[
 \mathcal G=\{V_0,V_1,V_0^\dagger,V_1^\dagger\},\qquad
 V_0=HTHT^\dagger=R_x(\pi/4)R_z(-\pi/4),\quad V_1=TV_0T^\dagger.
\]
The subgroup generated by $\mathcal G$ is dense in $\operatorname{SU}(2)$. Indeed, $\operatorname{tr}V_0=1+1/\sqrt2$ is not an algebraic integer, since its monic minimal polynomial has a noninteger coefficient. Hence $V_0$ has infinite order, and the closure of its powers is a rotation circle. Conjugation by $T$ gives a second circle with a nonparallel axis; their infinitesimal generators and the corresponding Lie bracket span $\mathrm{su}(2)$, which proves density. To construct a fixed coarse net of radius $r$, we first compute a finite rational $\delta$-net of the unit sphere in four real coordinates using stereographic charts, with $0<\delta<r/4$. We enumerate gate words until every grid point has a certified approximation within $r-2\delta$. The triangle inequality then gives a strict covering margin for the whole sphere, and density guarantees that the enumeration terminates. Since $r$ and the grid are fixed independently of the target, this cost is fixed; overlapping acceptance radii make subsequent net choices certifiable.

Dawson and Nielsen's Theorem~1 and Sections~3--4~\cite{dawson2006} then give an operator-norm approximation of accuracy $\zeta$ with $O(\log^4(1/\zeta))$ gates. To justify finite-bit evaluation as well, we record the conditioning of the balanced-commutator step. For a residual $W_{\rm res}=\mathcal Q(c,\mathbf v)$ near the identity, we construct the balanced commutator from
\[
 \begin{gathered}
 h=\|W_{\rm res}-I\|=\sqrt{2(1-c)}\le1,\qquad
 s=((1-c)/2)^{1/4},\qquad t=\sqrt{1-s^2},\\
 V_{\rm b}=\mathcal Q(t,s,0,0),\qquad W_{\rm b}=\mathcal Q(t,0,s,0).
 \end{gathered}
\]
Their exact commutator is
\begin{equation}\label{eq:SK-commutator}
 V_{\rm b}W_{\rm b}V_{\rm b}^\dagger W_{\rm b}^\dagger
 =\mathcal Q(1-2s^4,\,2ts^3,-2ts^3,2t^2s^2).
\end{equation}
Its scalar coordinate is $c$ and its unit axis is $u=(s,-s,t)^{\mathsf T}/\sqrt{1+s^2}$. For $n=\mathbf v/\|\mathbf v\|$, an axis $e_j$ satisfying both $|e_j\cdot u|,|e_j\cdot n|\le4/5$ can be certified with a fixed margin, since each unit vector has at most one coordinate exceeding $3/4$ in magnitude. Conjugation by $A_{\rm c}=H(e_j,n)H(u,e_j)$ aligns the commutator axis with $n$. The resulting factors $V_{\rm c},W_{\rm c}$ satisfy
\[
 V_{\rm c}W_{\rm c}V_{\rm c}^\dagger W_{\rm c}^\dagger=W_{\rm res},
 \qquad \|V_{\rm c}-I\|^2=\|W_{\rm c}-I\|^2=2(1-t)\le2s^2=h.
\]
The alignment denominators are at least $\sqrt{2/5}$ and $t\ge1/\sqrt2$.

At a recursion level with tolerance $\eta$, an overlapping test either certifies that the residual is at most $\eta/10$, which we charge directly, or certifies $h\ge\eta/20$ before using the formula. In the latter case
\[
 1-c\ge\eta^2/800,\qquad \|\mathbf v\|\ge h/\sqrt2.
\]
These lower bounds make the inverse factors in the root and normalization calculations polynomially bounded in $\eta^{-1}$. The standard recursion, with $\eta_{r+1}=C\eta_r^{3/2}$, has depth $O(\log\log(1/\zeta))$ and uses polynomially many calls in $\log(1/\zeta)$; moreover, each local tolerance satisfies $\log(1/\eta)=O(\log(1/\zeta))$. The root evaluators of Appendix~\ref{app:arithmetic} therefore evaluate the recursion in polynomial bit time after allocating guard precision across these conditioned operations. The arithmetic approximations certify the choices and evaluate the exact unitary factors, whose gate words form the circuit. Reserving a fixed fraction of $\zeta$ for these evaluations preserves the final error guarantee.

We apply this procedure to the $O((Q+1)S^2)$ supported factors and their controlled decompositions, allocating tolerances with sum at most $\epsilon_g$. The known predicates are compiled exactly, and the phase and clean work registers fit the bound of Theorem~\ref{thm:gates}. Each elementary factor approximates a full ideal unitary, including its action on nonzero workspace sectors. Replacing the factors one at a time therefore bounds the full-circuit error by the sum of their tolerances, even in the presence of earlier leakage; selecting the ancillary block cannot increase it. This gives Eq.~\eqref{eq:gate-cost} while keeping all original queries unchanged.

\section{Joint access from separate block encodings}\label{app:loaders}
The joint query used by our synthesis theorem can be supplied directly or constructed from coherently selectable block encodings. We give the latter construction because it explains how the input model relates to familiar selection and preparation methods~\cite{childswiebe2012,gilyen2019}, and because its scaling changes the norm that the synthesis algorithm must use. Throughout this appendix we assume either a supplied coherent selector or controlled implementations of the individual inputs; these are additional access capabilities for unknown unitaries~\cite{araujo2014}.

For each $j\in\{1,\ldots,k\}$, we assume that a supplied unitary $U_j$ acts on an encoding register and the $n$-dimensional signal register, and that, in its stated computational basis,
\[
 [U_j]_{(0,x),(0,y)}=(A_j)_{xy}/\alpha_j,
 \qquad 0\le x,y<n,
\]
where $\alpha_j>0$ is known. By padding with identity qubits and applying known coordinate permutations, we give all encodings a common width $q_e\ge1$ and write $d_e=2^{q_e}$. This preserves their designated corners at no query cost, while the extra encoding dimension ensures that the joint oracle meets the dimension requirement in Eq.~\eqref{eq:oracle-full}.

The selector register has width $q_J=\lceil\log_2(k+1)\rceil$ and dimension $d_J=2^{q_J}$. Its label zero supplies a genuine identity operation; labels $1,\ldots,k$ select the inputs, and the remaining labels also act as identity. The complete selected unitary is the $d_Jd_en\times d_Jd_en$ matrix with entries
\begin{equation}\label{eq:selector}
 [U_{\rm sel}]_{(j,e,x),(\ell,f,y)}
 =\delta_{j\ell}
 \begin{cases}
  [U_j]_{(e,x),(f,y)},&1\le j\le k,\\
  \delta_{ef}\delta_{xy},&j=0\text{ or }j>k.
 \end{cases}
\end{equation}
The displayed action is unitary on every selector label. When it is supplied directly we charge its stated query cost; when controlled implementations of the $U_j$ are supplied instead, we realize it by applying each $U_j$ under the known condition $J=j$, so one selector call costs one controlled call to each input and the corresponding label comparisons.

We designate the single ancillary label $(J,e)=(0,0)$ as the idle coordinate, and $(j,0)$ for $1\le j\le k$ as the active input and output coordinates. Their selected matrix is
\[
 \begin{pmatrix}
 I_n&0&0&\cdots&0\\
 0&A_1/\alpha_1&0&\cdots&0\\
 0&0&A_2/\alpha_2&\ddots&\vdots\\
 \vdots&\vdots&\ddots&\ddots&0\\
 0&0&\cdots&0&A_k/\alpha_k
 \end{pmatrix}.
\]
The original oracle width is $q_o=q_J+q_e$, and its dimension obeys $d_o=d_Jd_e\ge2(k+1)\ge1+2k$. We place the designated lists in the order of Eq.~\eqref{eq:oracle-full} by exchanging their known bit strings, at a cost polynomial in the number and width of those strings. This gives $i=1$ and $a=b=k$; the additional identity labels belong to the unspecified part of the oracle and are not included in the numerical arrays used for synthesis.

For a known unit column $g=(g_1,\ldots,g_k)^{\mathsf T}$ with a certified entry evaluator, two-level elimination computes an index rotation $P_g$ satisfying $P_g\sum_jg_j\ket j=\ket1$. It fixes zero and unused labels. The pivot and precision bounds are those of Appendix~\ref{app:gates}, and its synthesis error is charged as a known-gate error. The square joint unitary
\[
 U_{\rm row}=(P_g\otimes I_{d_en})U_{\rm sel}
\]
has idle input and output coordinate $(0,0)$, active input coordinates $(j,0)$ for $1\le j\le k$, and active output coordinate $(1,0)$. Its selected matrix is
\begin{equation}\label{eq:selected-row}
 \begin{pmatrix}
 I_n&0&\cdots&0\\
 0&\overline g_1A_1/\alpha_1&\cdots&\overline g_kA_k/\alpha_k
 \end{pmatrix}.
\end{equation}
The oracle dimension satisfies $d_o\ge k+2$, so this is an instance of Eq.~\eqref{eq:oracle-full} with $i=1$, $a=1$, and $b=k$. Unitarity proves contractivity of the loaded row. With equal weights it is the row $[A_1/\alpha_1\ \cdots\ A_k/\alpha_k]/\sqrt k$; omitting this factor would incorrectly infer joint row contractivity from the separate input bounds.

The row algorithm acts on the variables $Z_j=\overline g_j A_j/\alpha_j$ that actually occur in Eq.~\eqref{eq:selected-row}. If all weights needed by the target are nonzero, its scalar polynomial is specified in those variables by
\begin{equation}\label{eq:row-coefficient-rescale}
 \widetilde c_{j_1\cdots j_d}
 =c_{j_1\cdots j_d}\prod_{t=1}^d
                              \frac{\alpha_{j_t}}{\overline g_{j_t}}.
\end{equation}
Consequently $\widetilde p(Z)=p(A)$, and we implement the change in Eq.~\eqref{eq:row-coefficient-rescale} by multiplying transition $T_j$ by $\alpha_j/\overline g_j$, preserving the number of states. This letter-dependent rescaling is essential when terms have different degrees. For the rational-input theorem, the multipliers must have rational real and imaginary parts; when the $\alpha_j$ are rational, one suitable choice for $k>1$ is
\[
 g_j=\frac{2t_j}{1+\|t\|_2^2}\ (j<k),\qquad
 g_k=\frac{1-\|t\|_2^2}{1+\|t\|_2^2},\qquad t_j=\frac1{2k},
\]
whose entries are positive and rational, while for $k=1$ we take $g_1=1$. Equal weights may introduce algebraic multipliers and therefore require a separate check that the resulting coefficients remain rational; for example, with $\alpha_1=\alpha_2=1$, the degree-two rescaling for two equal weights is the rational factor two. In either case the compiler estimates the norm in the loaded row coordinates, including the chosen weights, which determines the normalization of the resulting circuit.

We next obtain Eq.~\eqref{eq:pencil} with its known $a\times b$ matrices $F_j$. We choose positive known rational numbers $\beta_j\ge\|F_j\|$ and set $d_F=2^{q_F}$, where $q_F=\lceil\log_2(a+b)\rceil$. On a numerical register with $d_F$ labels, we use the explicit square unitary
\[
 W_j=
 \begin{pmatrix}
 F_j/\beta_j&(I_a-F_jF_j^\dagger/\beta_j^2)^{1/2}&0\\
 (I_b-F_j^\dagger F_j/\beta_j^2)^{1/2}&-F_j^\dagger/\beta_j&0\\
 0&0&I_{d_F-a-b}
 \end{pmatrix}.
\]
Its row groups have lengths $a,b,d_F-a-b$ and its column groups have lengths $b,a,d_F-a-b$, omitting an empty final group. The designated input labels are $0,\ldots,b-1$ and output labels $0,\ldots,a-1$, with corner $F_j/\beta_j$. This is a known numerical Julia matrix.

The joint registers are now $(J,F,e,\mathrm{signal})$. We combine the known conditional $W_j$ with Eq.~\eqref{eq:selector} to obtain the full square unitary $\mathcal S$ of order $d_Jd_Fd_en$, whose entries are
\[
 [\mathcal S]_{(j,\mu,e,x),(\ell,\nu,f,y)}
 =\delta_{j\ell}
 \begin{cases}
 (W_j)_{\mu\nu}[U_j]_{(e,x),(f,y)},&1\le j\le k,\\
 \delta_{\mu\nu}\delta_{ef}\delta_{xy},&j=0\text{ or }j>k.
 \end{cases}
\]
Since the two factors in each active selector block are unitary, $\mathcal S$ is unitary on the stated registers.

Assuming the $\alpha_j$ are positive rationals for the rational-input application, we define
\[
 \Lambda=\sum_j\beta_j\alpha_j,\qquad
 a_j=\sqrt{\beta_j\alpha_j/\Lambda}.
\]
A known selector preparation $P$ fixes label zero and maps $\ket1$ to $\sum_ja_j\ket j$, while acting as identity outside the active selector labels. We use the full unitary
\[
 U_\Delta=(P^\dagger\otimes I_{d_Fd_en})
                  \mathcal S(P\otimes I_{d_Fd_en}).
\]
The idle coordinate is $(0,0,0)$; the active input coordinates are $(1,\nu,0)$ for $0\le\nu<b$, and the active output coordinates are $(1,\mu,0)$ for $0\le\mu<a$. Directly multiplying the stated matrix entries gives
\[
 [U_\Delta]_{(1,\mu,0,x),(1,\nu,0,y)}
 =\sum_j a_j^2\frac{f_{\mu\nu,j}}{\beta_j}
                      \frac{(A_j)_{xy}}{\alpha_j}
 =\frac1\Lambda\sum_jf_{\mu\nu,j}(A_j)_{xy}.
\]
The idle corner is $I_n$ and its off-diagonal blocks vanish because both preparations fix selector label zero. With the known coordinate permutations, this is precisely Eq.~\eqref{eq:oracle-full} for the loaded block $\Delta(A)/\Lambda$, with $i=1$. Its width is $q_o=q_J+q_F+q_e$ and $d_o=d_Jd_Fd_e\ge1+a+b$. This is the explicit linear-combination construction underlying the block-encoding rule of Gily\'en et al.~\cite[Lemma~52]{gilyen2019}, with its rectangular coordinate lists and identity resource retained.

The compiler now receives $Z_j=A_j/\Lambda$, so to preserve the original target we give it coefficients
\begin{equation}\label{eq:coefficient-rescale}
 \widetilde c_w=\Lambda^{|w|}c_w.
\end{equation}
We obtain these coefficients by scaling every transition by $\Lambda$, again preserving the number of states, and compute the uniform norm in the resulting coordinates. The preparations and numerical $W_j$ matrices are known gates, whose synthesis error enters Corollary~\ref{cor:oracle-error}.

Each forward call to $U_{\rm row}$ or $U_\Delta$ uses one forward selector call and the stated known gates; the inverse preparation is known, so it does not require an inverse input query. If we construct the selector from individually controlled inputs, a $Q$-query synthesis makes $Q$ calls to each such input, including on amplitudes already marked as failures. These costs explain the distinction between receiving a joint block directly and first loading it from separate encodings. The upper synthesis guarantee applies to the loaded unitary, whereas a converse over all compatible completions retains that larger quantifier: a particular loader, or additional promises on its constituent inputs, can restrict the available completions or matrix tuples.

\section{From causal Choi matrices to the gates of a comb}\label{app:comb}
A channel-level target specifies how a controller connects unknown channel calls while retaining its own memory. We compute the controller's gates from rational causal Choi data, using the sequential characterization of Chiribella, D'Ariano, and Perinotti~\cite[Theorems~5--6]{chiribella2009} and the prefix-rank memory construction of Bisio et al.~\cite[Theorem~2]{bisio2010}. Our task is to make the numerical factors, full unitary gates, and precision costs explicit.

\subsection{Causal data and sequential isometries}\label{sec:comb-application}
We construct the processing gates in Eq.~\eqref{eq:intro-comb-map} from their causal Choi data. There are $Q$ channel slots and $J=Q+1$ known processing steps.

We describe the known network before connecting the channel slots, with input $X_j$ and output $Y_j$ each of dimension $n=2^{n_s}$. In the unnormalized Choi convention~\cite{jamiolkowski1972,choi1975,watrous2018}, a map's entry at $(y,x),(y',x')$ is $\langle y|\mathcal C(\ket x\!\bra{x'})|y'\rangle$; for a prefix we use the chronological order $(Y_1,X_1,\ldots,Y_j,X_j)$. Causality is then expressed by the following partial-trace identities.

\begin{definition}[Causal controller data]\label{def:causal-controller-data}
Causal controller data consist of explicit rational Hermitian matrices $C^{(j)}$ of order $n^{2j}$, for $0\le j\le J$, satisfying
\begin{equation}\label{eq:controller-data}
 C^{(0)}=[1],\qquad C^{(j)}\succeq0,\qquad
 \operatorname{Tr}_{Y_j}C^{(j)}=C^{(j-1)}\otimes I_{X_j}
       \quad(1\le j\le J).
\end{equation}
An optional successful outcome is specified by a rational matrix $S$ with $0\preceq S\preceq C^{(J)}$.
\end{definition}

Tracing the latest output leaves the preceding process tensored with identity on the latest input, so an earlier output cannot depend on an input still to be supplied. We use these intermediate identities to compute the controller, and measure its cost in the explicit rational encoding of $C^{(j)}$ and the optional $S$. Writing $r_j=\rank C^{(j)}$ and $r_* =\max_jr_j$, we obtain the following finite-bit realization. Controller space includes signal, memory, phase, and clean work; environments internal to channel primitives are counted separately when their Stinespring implementations are supplied.

\Needspace{10\baselineskip}
\begin{proposition}[Compilation from causal Choi data]\label{prop:comb-compilation}
For causal controller data as in Definition~\ref{def:causal-controller-data} and rational $0<\epsilon<1$, a classical algorithm computes descriptions of $J$ known unitaries realizing $C^{(J)}$, and of the optional terminal instrument realizing $S$. After any $Q$ channel calls are connected, finite-gate synthesis has diamond-norm error at most $\epsilon$. It uses $O(1+n_s+\log r_*)$ controller qubits and $\widetilde O(Jn^2r_*^2)$ known gates. Classical cost is polynomial in the explicit rational input length and $\log(1/\epsilon)$.
\end{proposition}

The gate bound suppresses factors polynomial in register widths and logarithmic in $1/\epsilon$. Ideal gates realize the Choi data exactly, and the optional binary instrument has complementary outcome $C^{(J)}-S$; finite-gate error is uniform over the connected channels. Since the terminal input matrix has $n^{4J}$ entries, the bound concerns an explicit channel-level description. The compact coefficient input of polynomial Kraus synthesis specifies a different task, with known gates restricted to auxiliary labels.

To construct the processing isometries, we write a prefix index as $z_j=(y_1,x_1,\ldots,y_j,x_j)$, with $0\le x_t,y_t<n$ and $d_j=n^{2j}$. The causal identities become
\begin{equation}\label{eq:comb-causality}
 \sum_{y=0}^{n-1}C^{(j)}_{(z,y,x),(z',y,x')}
     =\delta_{xx'}C^{(j-1)}_{z,z'}
       \quad(1\le j\le J).
\end{equation}
These rational linear equations constrain each time cut, and hence contain the information needed for a sequential realization. We first factor each prefix as
\begin{equation}\label{eq:comb-factors}
 C^{(j)}=\mathsf R_j\mathsf R_j^\dagger,\qquad
 \mathsf R_j:\ d_j\times r_j,\quad r_j=\rank C^{(j)},
 \qquad \mathsf R_0=[1].
\end{equation}
The independent columns of $\mathsf R_j$ will serve as the memory coordinates after step $j$. To relate consecutive factors, whose column bases may be chosen independently, we form two matrices with columns indexed by $(z,x)$:
\begin{equation}\label{eq:comb-TB}
 \begin{aligned}
 T_j[(x',\beta),(z,x)]&=\delta_{x'x}(\mathsf R_{j-1})_{z,\beta},
   & T_j&:\ nr_{j-1}\times nd_{j-1},\\
 B_j[(y,\alpha),(z,x)]&=(\mathsf R_j)_{(z,y,x),\alpha},
   & B_j&:\ nr_j\times nd_{j-1}.
 \end{aligned}
\end{equation}
Causality equates their column Grams,
\begin{equation}\label{eq:comb-Gram-equality}
 B_j^\dagger B_j=T_j^\dagger T_j.
\end{equation}
Their common entry at $(z,x),(z',x')$ is $\delta_{xx'}C^{(j-1)}_{z',z}$. This index order follows from the unconjugated factors in Eq.~\eqref{eq:comb-TB} and will also fix the transpose in the terminal instrument. Since $\mathsf R_{j-1}$ has independent columns, $T_j$ has full row rank, so we can define the processing map by
\begin{equation}\label{eq:comb-W}
 W_j=B_jT_j^\dagger(T_jT_j^\dagger)^{-1},
       \qquad W_j:\ nr_j\times nr_{j-1}.
\end{equation}
The Gram equality gives $\ker T_j=\ker B_j$ and hence $W_jT_j=B_j$; moreover,
\[
 W_j^\dagger W_j
 =(T_j^+)^\dagger B_j^\dagger B_jT_j^+
 =(T_jT_j^+)^\dagger(T_jT_j^+)=I_{nr_{j-1}}.
\]
Thus $W_j$ is an isometry on all $nr_{j-1}$ input coordinates, which also implies $nr_j\ge nr_{j-1}$. Its defining equality gives the recurrence
\begin{equation}\label{eq:comb-factor-recurrence}
 (\mathsf R_j)_{(z,y,x),\alpha}
   =\sum_{\beta=0}^{r_{j-1}-1}
      W_j[(y,\alpha),(x,\beta)](\mathsf R_{j-1})_{z,\beta}.
\end{equation}
Starting from the one-dimensional memory at $j=0$, this recurrence reproduces $\mathsf R_j$ at every prefix, and tracing its final memory gives the prescribed $C^{(j)}$. Conversely, tracing the current output and memory in a sequential isometry yields Eq.~\eqref{eq:comb-causality}. Dilating known processing channels and retaining their environments extends this necessity argument to general controllers.

\subsection{Channel slots and terminal instruments}
We now connect the channel calls to the computed isometries. Grouping the entries of $W_j$ into signal matrices gives
\begin{equation}\label{eq:comb-signal-blocks}
 [W_j^{\alpha\beta}]_{yx}
              =W_j[(y,\alpha),(x,\beta)].
\end{equation}
If slot $t$ has Kraus matrices $K^{(t)}_{a_t}$, these known signal matrices produce the output Kraus matrices
\begin{equation}\label{eq:comb-Kraus}
 T_{\boldsymbol a,\alpha}=
 \sum_{\beta_1,\ldots,\beta_Q}
 W_J^{\alpha\beta_Q}K^{(Q)}_{a_Q}
 W_Q^{\beta_Q\beta_{Q-1}}\cdots
 K^{(1)}_{a_1}W_1^{\beta_1 0}.
\end{equation}
For $Q=1$ this is $\sum_{\beta_1}W_2^{\alpha\beta_1}K^{(1)}_{a_1}W_1^{\beta_1 0}$, and for $Q=0$ it is $W_1^{\alpha0}$. Summing over $\boldsymbol a$ and $\alpha$ gives a normalized channel, because we can successively use $\sum_aK_a^{(t)\dagger}K_a^{(t)}=I_n$ and $W_j^\dagger W_j=I$. A change of Kraus representation, after zero-padding the lists to equal length, only mixes a summed index by a unitary and therefore leaves the operation invariant. The controller consequently depends on the channel slots themselves.

The controller uses the $n_s$-qubit signal $\mathsf S$ and a memory $\mathsf M_{\rm c}$ with $R_b=2^{\lceil\log_2 r_*\rceil}$ labels. At step $j$, we embed the columns of $W_j$ on input strings $(x,\beta)$ with $\beta<r_{j-1}$ and output strings $(y,\alpha)$ with $\alpha<r_j$. To extend this isometry to all $nR_b$ basis states, as required for a complete gate~\cite{iten2016}, we use two-level unitary eliminations to send its columns to the first $nr_{j-1}$ coordinate columns. If $G_j$ is this elimination unitary and $P_j$ sends the supported input strings to those coordinates, the required gate is
\begin{equation}\label{eq:comb-full-unitary}
 U_j=G_j^\dagger P_j.
\end{equation}
Its supported columns implement $W_j$ deterministically, and the completion specifies every remaining column. During elimination, a norm-one column has a pivot of magnitude at least $1/\sqrt{nR_b}$, so certified pivot selection and the phase-sensitive compilation of Appendix~\ref{app:gates} apply.

We execute $U_1,\mathcal E_1,U_2,\ldots,\mathcal E_Q,U_J$, retaining the initially zero memory until the final trace (Figure~\ref{fig:comb}). Each channel acts on $\mathsf S$ and fixes $\mathsf M_{\rm c}$; all calls are forward and unconditional. If supplied Stinespring unitaries implement these primitives, we give slot $t$ a fresh zero-input register $\mathsf E_t^{\rm slot}$ of $q_{{\rm env},t}$ qubits. Retaining them until the final trace adds $\sum_tq_{{\rm env},t}$ qubits and ensures that each slot applies its declared channel.

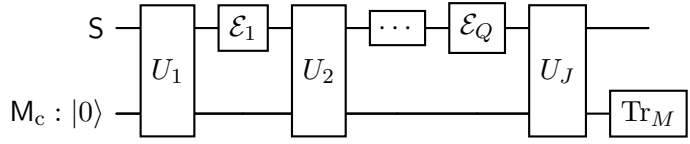
\begin{figure}[t]
\centering
\begin{quantikz}[column sep=0.32cm]
 \lstick{$\mathsf S$}&\gate[2]{U_1}&\gate{\mathcal E_1}&\gate[2]{U_2}&\gate{\cdots}&\gate{\mathcal E_Q}&\gate[2]{U_J}&\qw\\
 \lstick{$\mathsf M_{\rm c}:\ket0$}&\qw&\qw&\qw&\qw&\qw&\qw&\gate{\operatorname{Tr}_M}
\end{quantikz}
\caption{A channel-only comb. All $U_j$ are known full unitaries computed from the causal Choi matrices. The unknown channels act only on the signal, whereas the controller memory is retained throughout. A probabilistic comb additionally applies the terminal memory instrument before this trace. The diagram counts channel calls; a fully unitary realization of their internal action requires separately supplied Stinespring implementations and fresh environment inputs.}
\label{fig:comb}
\end{figure}

To realize a specified successful outcome, we use the supplied domination $0\preceq S\preceq C^{(J)}$, which is the probabilistic normalization condition of the network formalism~\cite[Theorems~9--10]{chiribella2009}. After constructing the deterministic controller, we compute the terminal memory effect from its factor:
\begin{equation}\label{eq:comb-effect}
 H_S=\mathsf R_J^+S(\mathsf R_J^+)^\dagger,\qquad E_S=H_S^{\mathsf T}.
\end{equation}
Domination places the columns of $S$ in the range of $\mathsf R_J$, and sandwiching by $\mathsf R_J^+$ gives $0\preceq H_S\preceq I_{r_J}$ and $\mathsf R_JH_S\mathsf R_J^\dagger=S$. Since transpose preserves positivity, $E_S$ is also an effect, and a final memory instrument with successful matrix $E_S^{1/2}$ has Choi matrix
\begin{equation}\label{eq:comb-effect-check}
 \mathsf R_J E_S^{\mathsf T}\mathsf R_J^\dagger=S.
\end{equation}
Summing the post-instrument memory label puts $(E_S)_{\beta\alpha}$ in front of column $\alpha$ of $\mathsf R_J$ and conjugate column $\beta$, which accounts for the ordinary transpose even when the Choi entries are complex. We implement the binary instrument by the known Julia unitary of $E_S^{1/2}$ on $\mathsf M_{\rm c}$ and an initially zero outcome qubit $\mathsf b_{\rm c}$, extending the effect by zero outside the $r_J$ memory labels. Retaining $\mathsf b_{\rm c}$ and discarding memory yields $S$ on outcome zero and $C^{(J)}-S$ on the other outcome.

A dominating deterministic causal completion is necessary as well: it is obtained by summing the outcomes of any physical instrument. Together with the construction above, this proves the existence criterion. The polynomial-bit compiler takes such a completion as exact rational input; finding one from $S$ alone is a separate semidefinite feasibility problem whose exact boundary behavior is not covered by that bound.

\subsection{Finite-bit compilation and resources}
To compute these gates with certified precision, we first determine the exact supports of the rational Choi matrices. Fraction-free elimination and positive-semidefiniteness tests give each rank and a factorization
\[
 C^{(j)}=Z_jD_jZ_j^\dagger,
\]
where $Z_j$ has rational independent columns and $D_j$ is positive rational diagonal. We retain the positive pivots in an exact pivoted $LDL^\dagger$ factorization, deleting zero residual rows, and take $\mathsf R_j=Z_jD_j^{1/2}$. The minor bounds of Appendix~\ref{app:arithmetic} bound every positive pivot below by $2^{-\poly(\ell)}$ in the explicit rational input length $\ell$.

In the ordering of Eq.~\eqref{eq:comb-TB}, the matrix $T_jT_j^\dagger$ repeats
$\mathsf R_{j-1}^{\mathsf T}\overline{\mathsf R_{j-1}}$ on its $n$ input labels. This Gram equals
$D_{j-1}^{1/2}(Z_{j-1}^\dagger Z_{j-1})^{\mathsf T}D_{j-1}^{1/2}$.
We therefore evaluate its inverse using a rational positive definite inverse and positive diagonal inverse roots. Eqs.~\eqref{eq:comb-W} and~\eqref{eq:comb-effect} have polynomial-bit entry evaluators computed directly from the original prefix matrices, and the root procedure of Appendix~\ref{app:arithmetic} handles a singular terminal effect.

There are $J$ known square unitaries of order $nR_b<2nr_*$, together with at most one known memory instrument. Two-level decomposition and the same controlled-gate and Solovay--Kitaev implementation used for the polynomial circuits require
\[
 O\!\left(J(nR_b+1)^2\poly\!\left(n_s+\log(r_*+1),
             \log\frac{J(nR_b+1)}{\epsilon}\right)\right)
\]
elementary gates and $O(n_s+\log(r_*+1))$ signal, memory, phase, and clean work qubits, plus one terminal outcome qubit when required. Supplied Stinespring environments contribute the separate count above. Matrix evaluation and gate compilation are polynomial in the explicit Choi input and precision lengths, as asserted in Proposition~\ref{prop:comb-compilation}.

Finally, we allocate total operator-norm error $\epsilon/2$ to the known unitaries. Replacing them one at a time changes the resulting channel or flagged instrument by at most $\epsilon$ in diamond norm: a unitary error $\eta$ contributes at most $2\eta$, and every intervening channel and final trace has diamond norm one. This proves the claimed bound uniformly over the connected channels and arbitrary external references.

\section{Joint residual certificates and degree reduction}\label{app:joint-details}
The joint synthesis theorem requires a certificate that preserves interference between input labels and a circuit that carries those labels through a common query sequence. We prove the extensions of the scalar construction needed for these two tasks, then specify the complete register implementation and its use for coherent Kraus operations. All computed matrices come from the supplied rational coefficients and query data.

\subsection{Joint residual certificates and their computation}\label{app:joint-coordinates}\label{app:proof-lem:block-certificate}\label{app:proof-prop:block-classical}
We first reduce the coefficient recurrence using the reachability space of all columns of $R_0$ and the observability space of all rows of $L$. Polynomiality is then equivalent to joint nilpotence of the reduced transitions, which the nested transition ranges decide as in the scalar case. To determine the degree and test for a zero target, we compute
\[
 H_0=R_0R_0^\dagger,\qquad H_{d+1}=\sum_jT_jH_dT_j^\dagger,
 \qquad \operatorname{Tr}(LH_dL^\dagger)
             =\sum_{|w|=d}\|P_w\|_F^2.
\]
The trace vanishes precisely when every coefficient matrix of degree $d$ is zero, so this rational recurrence tests all input columns together. We retain the same joint structure in the certificate by applying Lemma~\ref{lem:residual} to the constant rows $e_\nu^{\mathsf T}$ and the nonempty left residuals of the $v$ target rows. Their basis $f_1,\ldots,f_r$ has $u$ entries per row, $r\le s+u$, and coefficient arrays
\begin{equation}\label{eq:block-residual}
 f_\emptyset=C_f,\quad f_{jw}=A_j^{\rm res}f_w,\qquad
 P_\emptyset=D_p,\quad P_{jw}=P_jf_w.
\end{equation}
Here $f_w,C_f$ have size $r\times u$ and $P_j$ has size $v\times r$. The $u$ constant rows and the successors $x_jf_\ell$ are independent, giving $h=u+kr$ rows in $g$ and the coordinate matrices
\begin{equation}\label{eq:block-coordinates}
 E=[I_u\ 0\ \cdots\ 0],\quad
 J=[C_f\ A_1^{\rm res}\ \cdots\ A_k^{\rm res}],\quad
 P=[D_p\ P_1\ \cdots\ P_k].
\end{equation}
They satisfy $Eg_\emptyset=I_u$, $Eg_w=0$ for nonempty $w$, $Jg_w=f_w$, and $Pg_w=P_w$, while $L_j$ selects the $j$th $r$-row successor block. Substituting these maps into Eqs.~\eqref{eq:EM}--\eqref{eq:Phi} gives $E_\alpha,M_\mu$ of size $br\times h$ and the joint map $\Phi$. Evaluation of $g$ means
\[
 g(A)=\left[\sum_w(g_w)_{t\nu}A_w\right]_{\substack{0\le t<h\\0\le\nu<u}},
 \qquad g(A)\text{ has size }hn\times un.
\]
Thus the hereditary form $q_{Q_0}=g^*Q_0g$ evaluates to a $un\times un$ matrix and includes all input-label cross terms. We prove completeness for this general form, since the strict search needs it for both the target defect and the constant identity.

\begin{proof}[Proof of Theorem~\ref{lem:block-certificate}]
A factor of $S$ gives a polynomial square, and a factor of $T$ gives the defect of $\Delta(A)$ on the residual rows. Their sum is positive on every input column of length $un$, proving sufficiency with all cross terms retained.

For necessity, we index coefficient positions by a word and an input label. Independence of the rows, together with the unique starred--unstarred division, makes $Q_0\mapsto g^*Q_0g$ injective as in Theorem~\ref{thm:certificate}. Evaluation at the scaled word tuple, summed over the $u$ constant inputs, gives positive definite coefficient Grams and trace coercivity. The bounded-preimage argument of that theorem therefore closes the cone, so an infeasible positive polynomial would have a separator $H\succeq0$ with $\Phi^*(H)\succeq0$ and $\operatorname{Tr}(HQ_0)<0$.

We factor $H^{\mathsf T}=W^\dagger W$ with $n_H=\rank H$ independent rows and compute
\begin{equation}\label{eq:block-witness}
 R_H=WJ^{\mathsf T},\quad Z_{j,H}=WL_j^{\mathsf T},\quad
 Y_j=Z_{j,H}R_H^+,\quad \omega_\nu=WE^{\mathsf T}e_\nu.
\end{equation}
The matrices $R_H,Z_{j,H}$ have size $n_H\times r$. Because their localizing matrix has the form in Eq.~\eqref{eq:GNS-localizing}, its quadratic-form argument gives $\ker R_H\subseteq\ker Z_{j,H}$ and hence $Y_jR_H=Z_{j,H}$ and $\|\Delta(Y)\|\le1$, also for singular $H$. The joint extension enters in evaluating a polynomial row: with $f_\lambda=\lambda f$ and $\zeta(f_\lambda)=R_H\lambda^{\mathsf T}$, the residual identity becomes
\[
 \zeta(f_\lambda)
  =\sum_{\nu=0}^{u-1}(\lambda C_f)_\nu\omega_\nu
       +\sum_jY_j\zeta(f_{\lambda A_j^{\rm res}}).
\]
Every residual on the right has smaller degree. Induction, starting with constant rows, proves
\[
 R_H\lambda^{\mathsf T}
   =\sum_{\nu=0}^{u-1}(f_\lambda)_\nu(Y)\omega_\nu,
 \qquad
 \sum_\nu(g_t)_\nu(Y)\omega_\nu=We_t.
\]
The second equality also uses $Y_jR_H=Z_{j,H}$ for the successor rows. Let $\Omega$ be the column obtained by stacking $\omega_0,\ldots,\omega_{u-1}$. Then
\[
 \Omega^\dagger q_{Q_0}(Y)\Omega
   =\sum_{t,s}(Q_0)_{ts}(We_t)^\dagger(We_s)
   =\sum_{t,s}(Q_0)_{ts}H_{st}
   =\operatorname{Tr}(HQ_0)<0.
\]
The ordinary transpose in $H^{\mathsf T}=W^\dagger W$ gives exactly this complex index pairing, so the negative value contradicts positivity at the admissible tuple $Y$ and proves necessity.
\end{proof}

\begin{proof}[Proof of Proposition~\ref{prop:block-classical}]
We adapt the scalar conditioning bounds to the matrix constant $C_f$ and the $u$-column identity. In particular, the coercive functional now has value $u$ at $I_u$, which changes the outer search radius.

At the scaled creation tuple, the positive functional is
\[
 \ell_\rho(q)=\sum_{\nu=0}^{u-1}
 (e_\nu\otimes e_\emptyset)^\dagger q(\rho S^{(D)})
                         (e_\nu\otimes e_\emptyset),
 \qquad \ell_\rho(I_u)=u.
\]
The weighted coefficient Grams are computed by
\[
 W_f=C_fC_f^\dagger+\rho^2\sum_jA_j^{\rm res}W_f(A_j^{\rm res})^\dagger,
 \qquad W_g=\operatorname{diag}(I_u,\rho^2W_f,\ldots,\rho^2W_f).
\]
The recurrence terminates at the degree, and independence of the rows gives $W_f\succ0$. With the same $\rho$ as in Appendix~\ref{app:classical-margins}, we set
\begin{equation}\label{eq:block-conditioning}
 \lambda=\left(1+\sum_{\ell,t}
 (|\operatorname{Re}(W_f^{-1})_{\ell t}|+
  |\operatorname{Im}(W_f^{-1})_{\ell t}|)\right)^{-1},
 \qquad \chi=\tfrac34\rho^2\lambda.
\end{equation}
The query defect is at least $3I/4$ on that tuple, while $W_g\succeq\rho^2\lambda I$. Applying these two lower bounds to the certificate gives
\begin{equation}\label{eq:block-trace-cap}
 \ell_\rho(g^*(S+\Phi(T))g)\ge\chi(\operatorname{Tr}S+\operatorname{Tr}T),
 \qquad \operatorname{Tr}S+\operatorname{Tr}T\le\beta^2u/\chi.
\end{equation}

Let $c_f\ge\|C_f\|$ and $a_j\ge\|A_j^{\rm res}\|$ be rational entry-norm bounds. Using the coordinate bounds $c_j^\Delta$ from Appendix~\ref{app:classical-margins}, we define
\begin{equation}\label{eq:block-residual-bound}
 U_f=c_f\sum_{d=0}^{D-1}\left(\sum_jc_j^\Delta a_j\right)^d,
 \quad M_0=1+\left(b+\sum_j(c_j^\Delta)^2\right)U_f^2,
 \quad \sigma_0=(2M_0)^{-1}.
\end{equation}
The finite residual expansion gives $\|f(A)\|\le U_f$ and bounds $g^*(I_h+\Phi(I_{br}))g$ between zero and $M_0I_u$. Consequently $I_u-\sigma_0g^*(I_h+\Phi(I_{br}))g$ is positive; after obtaining its certificate, we add $\sigma_0I$ to both Grams to certify $I_u$ with spectral margin $\sigma_0$. The argument of Lemma~\ref{lem:margin} then gives margin $\mu=(1-\nu^2)\sigma_0$ for a target with $B_{\mathsf P}/\gamma\le\nu<1$.

The search uses $T\succeq\mu I/4$, $E^\dagger E-P^\dagger P/\gamma^2-\Phi(T)\succeq\mu I/4$, and an entry box of radius $U_0=1+2u/\chi$. For $d_0=(br)^2$ and $L_\Phi=1+\sum_\alpha\|E_\alpha\|_F^2+\sum_\mu\|M_\mu\|_F^2$, its radii are
\begin{equation}\label{eq:block-search-radii}
 r_{\rm in}=\min\{1,\mu/(8L_\Phi)\},\qquad
 R_{\rm out}=U_0\max\{1,d_0\}.
\end{equation}
The scalar perturbation argument applies to variable matrices of orders $br$ and $h$ with these radii, including the factor $u$ in $U_0$. Rational recurrence and minor bounds control their encoding lengths, so the capped search of Proposition~\ref{prop:strict} has the claimed polynomial bit complexity.

For a nonzero target, $t=\operatorname{Tr}(PW_gP^\dagger)>0$ satisfies $t\le uB_{\mathsf P}^2$, giving the positive rational lower bound $t/(u+t)\le B_{\mathsf P}$. Rational majorants $d_p\ge\|D_p\|$ and $p_j\ge\|P_j\|$ give the upper bound $d_p+\sum_jp_jc_j^\Delta U_f$. These bounds initialize the relative bisection of Theorem~\ref{thm:norm-search}, completing the norm-estimation claim.
\end{proof}

\subsection{Joint query sequences and compact row factorization}\label{app:joint-weights}\label{app:joint-row}\label{sec:block-queries}\label{sec:block-row}
The general construction first compensates the degree weights, so we need the scalar estimate in the joint operator norm. For any unit input and output columns, Lemma~\ref{lem:weights} applies to the corresponding scalar matrix element. Taking their supremum gives the second-derivative bound $e^2D^2B_{\mathsf P}/2$ with no dependence on $u,v$, and integration against the same positive kernel yields
\[
 \|\mathsf P-\mathcal W_v\mathsf P\|_{\SA,\Delta}\le\rho_N B_{\mathsf P},
 \qquad
 \|\mathcal W_v^{-1}\mathsf P\|_{\SA,\Delta}\le B_{\mathsf P}/(1-\rho_N),
 \qquad \rho_N=64D^2/(N+1)^2,
\]
when $D<N$ and $\rho_N<1$. Since the inverse divides each degree-$d$ coefficient matrix by $w_d$, the scalar normalization choice makes the compensated rectangular target strictly contractive.

For its certificate, or any strict certificate of a normalized target $\mathsf Q$, we take $K=T^{1/2}$ and $m=br$ and construct
\begin{equation}\label{eq:block-XY}
 X=\begin{pmatrix}KM_1\\\vdots\\KM_a\\E\end{pmatrix},\qquad
 Y=\begin{pmatrix}KE_1\\\vdots\\KE_b\\P_{\mathsf Q}\end{pmatrix},\qquad
 T_0=X^\dagger X.
\end{equation}
Their sizes are $(am+u)\times h$, $(bm+v)\times h$, and $h\times h$. The rational identity $X^\dagger X-Y^\dagger Y=S\succ0$ therefore gives the joint realization in Eq.~\eqref{eq:block-V}.

The bound $T_0\succeq S\succeq\mu I_h$ lets us apply Lemma~\ref{lem:extract} to obtain $VX=Y$ and $\|V\|<1$. Full column rank is consistent with $am+u\ge h$, since $k\le ab$. We evaluate $T_0^{-1}$ by rational arithmetic and $K$ by the certified root procedure of Appendix~\ref{app:arithmetic}.

To check the action on every input column at once, let $g_w$ be the $h\times u$ coefficient matrix and set $H_w=[KE_1g_w;\ldots;KE_bg_w]$, of size $bm\times u$. The constant and nonconstant coefficients of $VX=Y$ give
\begin{equation}\label{eq:block-Markov}
 H_\emptyset=\Bco,\quad\mathsf Q_\emptyset=\Dco,\qquad
 H_{jw}=\Aco F_j^{(m)}H_w,\quad
 \mathsf Q_{jw}=\Cco F_j^{(m)}H_w.
\end{equation}
These follow from $Xg_\emptyset=[0;I_u]$ and $Xg_{jw}=[F_j^{(m)}H_w;0]$, and their iteration proves Eq.~\eqref{eq:coefficient-factorization} with rectangular endpoints, including the vanishing coefficients above degree $D$. The energy identity then has defect $I_{un}-\mathsf Q(A)^\dagger\mathsf Q(A)$ and controls coherent superpositions of the input labels. We insert these endpoint blocks into the numerical stage
\begin{equation}\label{eq:block-stage}
 M_t=\begin{pmatrix}
 \Aco&\alpha_t\Bco&0_{bm\times v}\\
 0_{u\times am}&\zeta_t I_u&0_{u\times v}\\
 \xi_t\Cco&\xi_t\alpha_t\Dco&\gamma_t I_v
 \end{pmatrix}.
\end{equation}
Its block rows have lengths $bm,u,v$ and its block columns $am,u,v$. To verify contractivity on every column, we use the factorization in Eq.~\eqref{eq:Mt-factors}, replacing the pending scalar coordinate by $\C^u$ and the collected coordinate by $\C^v$. The input split uses $\alpha_tI_u,\zeta_tI_u$, the middle map is $V$ together with the two identity blocks, and the output combination uses $\xi_tI_v,\gamma_tI_v$. Since $\alpha_t^2+\zeta_t^2=\xi_t^2+\gamma_t^2=1$, these are respectively an isometry, a contraction, and a coisometry, whose product is Eq.~\eqref{eq:block-stage}.

We retain $\kappa=\min\{i,u+v\}$ identity labels and take
\begin{equation}\label{eq:block-memory}
 R=\max\left\{m,\left\lceil\frac{u+v}{\kappa}\right\rceil\right\},
 \qquad d_-=\kappa+b,\quad d_+=\kappa+a.
\end{equation}
Pending input label $\nu$ occupies flat identity index $\nu$, and output label $\mu$ occupies flat identity index $u+\mu$. Flat index $s$ means retained oracle label $\lfloor s/R\rfloor$ and memory label $s\bmod R$. The active coordinate $(\alpha,\lambda)$ occupies retained input label $\kappa+\alpha$ and memory $\lambda$, with $0\le\alpha<b$ and $0\le\lambda<m$. These assignments define the zero--one packing matrix $J_-$, of size $d_-R\times(bm+u+v)$. Using $a$ active output labels instead of $b$ active input labels defines $J_+$, of size $d_+R\times(am+u+v)$. In both cases Eq.~\eqref{eq:retained-labels} converts retained indices to their actual physical bit strings.

The packed query has three diagonal blocks: the queried active block on $m$ memory labels, $I_{un}$, and $I_{vn}$. The input insertion and output selection are the numerical matrices
\[
 I_{\rm pen}=\begin{pmatrix}0_{am\times u}\\I_u\\0_{v\times u}\end{pmatrix},
 \qquad P_{\rm out}=[0_{v\times bm}\ 0_{v\times u}\ I_v].
\]
Hence every chronological factor is explicitly
\begin{equation}\label{eq:block-query-cores}
 C_0=J_-M_0I_{\rm pen},\qquad
 C_t=J_-M_tJ_+^\dagger\ (0<t<Q),\qquad
 C_Q=P_{\rm out}M_QJ_+^\dagger.
\end{equation}
Theorem~\ref{thm:window-circuit} now gives the weighted output on all $u$ columns. We compute its compensated target using the $(D+1)s$-state recurrence of Eq.~\eqref{eq:degree-clock}: $\widetilde L=[L\ 0\ \cdots\ 0]$, the superdiagonal transition blocks are $T_j$, and block $d$ of $\widetilde R_0$ is $R_0/(\beta w_d)$. Its coefficient matrices are $P_w/(\beta w_{|w|})$, and the joint weight estimate above gives the same strict margin as in the scalar construction. The resulting certificate and query factors therefore cancel the weights on the entire output matrix.

On the row domain, the certificate of $\mathsf P/\beta$ instead gives a polynomial completion whose degree we can reduce once per query. From $E^\dagger E-P^\dagger P/\beta^2=S+\Phi(T)$, we form
\[
 \Theta(x)=\begin{pmatrix}\mathsf P(x)/\beta\\S^{1/2}g(x)\end{pmatrix}.
\]
The scalar word-matrix argument applies to this $(v+h)\times u$ completion and all pairs of its input coordinates: at cutoff $2D$, the query defect annihilates the residual outputs from words of length at most $D$. We therefore obtain the joint convolution identities
\begin{equation}\label{eq:block-inner}
 \sum_w\Theta_w^\dagger\Theta_w=I_u,\qquad
 \sum_w\Theta_w^\dagger\Theta_{wz}=0_{u\times u}\quad(z\ne\emptyset).
\end{equation}

To compute the complementary polynomial in finite-state form, let $G_\emptyset=[I_u;0_{kr\times u}]$ and let $T_j^g$ have $J$ in its $j$th $r$-row block, with zeros elsewhere. These transitions give $g_w=T_w^gG_\emptyset$, so we obtain its output matrix and joint Gram as
\[
 O_\Theta=\begin{pmatrix}P/\beta\\S^{1/2}\end{pmatrix},\qquad
 O_\Theta^\dagger O_\Theta=P^\dagger P/\beta^2+S.
\]
Taking $N_j=T_j^g$, $b_j=T_j^gG_\emptyset$, $O=O_\Theta$, and $d_0=O_\Theta G_\emptyset$ gives the required rational transitions and joint output Gram. The nonconstant reachable dimension is at most $s_0\le kr$, and the initial output length is $q_0\le v+h$.

Degree reduction must preserve these full $u\times u$ correlations. If $Z$ is an orthonormal frame for the highest-degree coefficient range and $\Pi=ZZ^\dagger$, Eq.~\eqref{eq:block-inner} gives $\Pi\Theta_\emptyset=0$, so the update
\[
 \Psi_w=\begin{pmatrix}(I-\Pi)\Theta_w\\Z^\dagger\Theta_{1w}\\\vdots\\Z^\dagger\Theta_{kw}\end{pmatrix}
\]
lowers the degree, and its correlations satisfy
\[
 \sum_w\Psi_w^\dagger\Psi_{wz}
 =\sum_w\Theta_w^\dagger\Theta_{wz}-\Theta_\emptyset^\dagger \Pi\Theta_z
 =\sum_w\Theta_w^\dagger\Theta_{wz}.
\]
Thus every reduction preserves the joint identities. To compute it from the recurrence, we use $b_j$ of size $s_0\times u$, $O_t$ of size $q_t\times s_0$, and $d_t$ of size $q_t\times u$, reducing to the state space spanned by all columns $N_wb_j$. On this jointly reachable space, the initial correlation identities are exactly
\[
 d_0^\dagger d_0+\sum_jb_j^\dagger Hb_j=I_u,
 \qquad O^\dagger d_0+\sum_jN_j^\dagger Hb_j=0_{s_0\times u}.
\]
Indeed, the nonempty correlation at $vj$ is
$(d_0^\dagger O+\sum_\ell b_\ell^\dagger HN_\ell)N_vb_j$; joint reachability makes its vanishing equivalent to the second condition.

We use the rational bases $U_\ell,R_t$ and shorted matrices $H_t$ of Appendix~\ref{app:row-compact}, with the matrix-valued $b_j,d_t$ in Eq.~\eqref{eq:row-updates}. The metric proof of Theorem~\ref{thm:row-compiler} depends only on $O_t$ and the transitions, and hence gives
\[
 H_t=\sum_wN_w^\dagger O_t^\dagger O_tN_w,
 \quad H_tR_t=O_t^\dagger O_tR_t,
 \quad K_t=(O_tR_t)^\dagger(O_tR_t)\succ0
\]
for each nonempty eliminated basis. Joint reachability also identifies the highest-degree range: after writing a state column in terms of $N_vb_j$ and multiplying by $O_tN_w$ with $|w|=d-t-1$, the terms with nonempty $v$ vanish in the previously eliminated range. The remaining columns are exactly the highest-degree coefficients, so $O_tR_tK_t^{-1/2}$ is the required frame. The preserved correlation identities give $d_d^\dagger d_d=I_u$ at the final constant stage and the rational stage Grams
\[
 O_t^\dagger O_t=H_t-\sum_jN_j^\dagger H_tN_j,\quad
 O_t^\dagger d_t=-\sum_jN_j^\dagger H_tb_j,\quad
 d_t^\dagger d_t=I_u-\sum_jb_j^\dagger H_tb_j.
\]
These Grams bound the numerical output matrices, and the eliminated dimensions telescope to $s_0$, giving $q_t\le q_0+ks_0$. For a nonempty basis, the rational positive matrix $K_t$ has the scalar proof's minor lower bound, so its certified root evaluation and guard-bit argument apply; an empty basis simply contributes no inverse-root block.

With $R_*\ge q_0+ks_0$, the packing and output matrices from Section~\ref{sec:row} give
\begin{equation}\label{eq:joint-row-cores}
 \begin{aligned}
 C_0&=E_{d-1}d_d &&\bigl((k+1)R_*\times u\bigr),\\
 C_t&=E_{d-1-t}R_{d-t}^{\rm out}\quad(1\le t<d)
        &&\bigl((k+1)R_*\times2R_*\bigr),\\
 C_d&=[I_v\ 0]R_0^{\rm out} &&\bigl(v\times2R_*\bigr).
 \end{aligned}
\end{equation}
We apply the last factor found first, using the same full row query at each step. The endpoint lists below place $C_0$ on an arbitrary data input and $C_d$ on the $v$ selected outputs.

\subsection{Full unitary implementation and resources}\label{app:joint-endpoints}\label{app:joint-resources}
Both factorizations act on a $q_L=\lceil\log_2\max\{u,v\}\rceil$-qubit data register and have rectangular endpoints,
\begin{equation}\label{eq:block-core-sizes}
 C_0:\ d_-R\times u,\qquad
 C_t:\ d_-R\times d_+R\ (0<t<Q),\qquad
 C_Q:\ v\times d_+R.
\end{equation}
We implement the first endpoint by a Julia gate that transfers the data dependence to the continuing processing amplitudes with $\mathsf L=0$; the last gate decodes the output into $\mathsf L$. The complementary amplitudes remain in the full unitary under failure labels, so this transfer preserves arbitrary input coherences. With $D_L=2^{q_L}$, $R_b=2^{\lceil\log_2R\rceil}$, and register order $(\mathsf O,\mathsf M,\mathsf L,\mathsf d)$, the required lexicographic lists are
\[
 \mathcal I_\pm=
 ((\sigma_\pm(\alpha),\lambda,0,0))_{0\le\alpha<d_\pm,\,0\le\lambda<R},
 \qquad \mathcal D_h=((0,0,\ell,0))_{0\le\ell<h}.
\]
The input/output pairs are $(\mathcal D_u,\mathcal I_-)$ initially, $(\mathcal I_+,\mathcal I_-)$ internally, and $(\mathcal I_+,\mathcal D_v)$ finally. For a matrix with $h_2$ rows and $h_1$ columns, we append the first $h_2$ strings with $\mathsf d=1$ to its input list and the first $h_1$ to its output list, in lexicographic order on $(\mathsf O,\mathsf M,\mathsf L)$. The bound $h_1,h_2\le d_oR_bD_L$ ensures that these directly computable lists fit, each with $h_1+h_2$ distinct strings. Completing them to their common union $\mathcal T_C$ in a fixed order and applying Eq.~\eqref{eq:support-unitary} specifies the whole Julia gate $J_C$, including its reverse columns and binary padding.

As in Theorem~\ref{thm:unitary}, $G_t$ applies $J_C$ on $\mathsf E=0$ and identity otherwise, while $D_t,F_t$ mark defects with fresh failure labels. The raw query fixes $\mathsf L,\mathsf M,\mathsf d,\mathsf E$ even on failed columns, so the theorem's induction applies to all $u$ data inputs simultaneously. We finally retain $\mathsf L<v$ coherently and select zero on the oracle, processing, dilation, failure, phase, and clean work outputs. These complete endpoints give the rectangular target for every compatible oracle completion.

\begin{lemma}[Constant rectangular targets]\label{lem:block-constant}
For a known rational $v\times u$ matrix $P_\emptyset$ and rational $0<\tau,\epsilon<1$, a classical algorithm computes a normalization $B<\beta\le(1+\tau)B$ when $B=\|P_\emptyset\|>0$, and a known unitary whose selected data block approximates $P_\emptyset/\beta$ within $\epsilon$. It uses no oracle query. Its classical and elementary-gate costs are polynomial in the explicit rational matrix, $\log(1/\tau)$, and $\log(1/\epsilon)$; the signal is unchanged. The zero matrix also uses no query.
\end{lemma}
\begin{proof}
For $P_\emptyset\ne0$, we put $G=P_\emptyset^\dagger P_\emptyset$ and $t=\operatorname{Tr}G$, giving the rational bracket $t/(u+t)\le\|P_\emptyset\|\le1+t$. Bisection with exact tests of $b^2I_u-G\succeq0$ yields an upper estimate $\widehat B$ of relative error $\tau/8$, so $\beta=(1+\tau/2)\widehat B$ lies in the required interval. We then apply the Julia matrix of $C=P_\emptyset/\beta$ on $(\mathsf L,\mathsf d)$, with input list $((\nu,0))_{\nu<u}$ followed by $((\mu,1))_{\mu<v}$ and output list $((\mu,0))_{\mu<v}$ followed by $((\nu,1))_{\nu<u}$. Completion on their common union by Eq.~\eqref{eq:support-unitary} gives a full unitary on $q_L+1$ qubits, to which the known-gate compiler applies directly. For the zero matrix, a Pauli $X$ on $\mathsf d$ gives the zero selected block at any positive normalization.

\end{proof}

For the general construction we use $m\le b((D+1)s+u)$ and the multiplicity in Eq.~\eqref{eq:block-memory}. For a row input we use $R=R_*=v+h+k^2r$, with $d_-=k+1$ and $d_+=2$. In either case, we define the support and width bounds
\begin{equation}\label{eq:block-gate-size}
 s_*=2((d_-+d_+)R+u+v),\qquad
 n_*=q_o+\lceil\log_2R\rceil+q_L+\lceil\log_2(2Q+1)\rceil+2.
\end{equation}
These widths include the supplied oracle and arbitrary data register as well as the additional processing qubits. The oracle, memory, dilation, failure, phase, and clean work registers are initialized at zero; the supported data input is supplied by the application.

\begin{proposition}[Resources for rectangular coefficients]\label{prop:block-resources}
For a rational finite-state rectangular target of degree $D\ge1$ and rational $0<\tau,\epsilon<1$, the construction implements $\mathsf P/\beta$ with $B_{\mathsf P}<\beta\le(1+\tau)B_{\mathsf P}$ and error at most $\epsilon$. It uses $Q=O(D/\sqrt\tau)$ queries, $O(n_*)$ non-signal qubits, and
\[
 O\!\left((Q+1)s_*^2\poly\!\left(n_*,
                  \log\frac{(Q+1)s_*}{\epsilon}\right)\right)
\]
known gates. Classical cost is polynomial in these parameters and the rational input length. For a row input, $Q=D$ and the margin dependence is polynomial in $\log(1/\tau)$.
\end{proposition}
\begin{proof}
The joint certificate and norm search give the normalization, and the compensated stage sequence implements the target exactly with ideal known gates. On the row domain, the factorization terminates at degree $D$, which is also necessary: the scalar coefficient-uniqueness argument applied to a nonzero degree-$D$ entry gives the same lower bound.

The complete Julia gates above have support at most $s_*$ on strings of width at most $n_*$, so Theorem~\ref{thm:gates} and Appendix~\ref{app:gates} give the stated gate count, including $O(n_*)$ clean work bits. In the general construction, $\kappa\le u+v$ and $R\le\max\{m,u+v\}$ bound the support polynomially in the target and joint-block dimensions, even though the retained identity lists can depend on $i$. Manipulating their physical labels incurs the stated $q_o$-bit width. Allocating total full-operator error $\epsilon$ to the $Q+1$ known matrices bounds the selected-block error by $\epsilon$ on every data input and every oracle completion. Together with Lemma~\ref{lem:block-constant}, this proves the resource assertions of Theorem~\ref{thm:matrix-synthesis}.
\end{proof}

For normalization one, Corollary~\ref{cor:normalization-one} applies to the final $v$-row contraction because multiplication by $\beta_0<1$ preserves contractivity. The coherent error bound also extends to rectangular targets. Let $J_U,J_Y$ embed $(\nu,x)$ with $\nu<u$ and $(\mu,x)$ with $\mu<v$ in the full register, setting all other bits to zero; their sizes are $d_{\rm full}\times un$ and $d_{\rm full}\times vn$. If $V_{\rm tar}^\dagger V_{\rm tar}=I_{un}$ and $\|J_Y^\dagger\mathcal UJ_U-V_{\rm tar}\|\le\delta$, Lemma~\ref{lem:coherent-output} gives $\|\mathcal UJ_U-J_YV_{\rm tar}\|\le\sqrt{2\delta}$. Thus the data output remains coherent while the other ancillary registers approach zero.

\subsection{Coherent Kraus access and output operations}\label{app:Kraus-oracle}\label{app:Kraus-histories}\label{app:Kraus-idle-example}\label{app:Kraus-erasure}
For Kraus synthesis we are given the complete $U_K$ with an identity sector, as in Eq.~\eqref{eq:oracle-full}, with $a=k$, $b=1$, and $d_o\ge i+k+1$. The active input label $i$ and output labels $i,\ldots,i+k-1$ give the matrix
\begin{equation}\label{eq:Kraus-full-unitary}
 U_K=\begin{pmatrix}
 I_{in}&0&0\\
 0&K_1&B_{1,K}\\
 \vdots&\vdots&\vdots\\
 0&K_k&B_{k,K}\\
 0&0&D_K
 \end{pmatrix}.
\end{equation}
Its block column lengths are $in,n,(d_o-i-1)n$ and its block row lengths are $in$, followed by $k$ rows of length $n$ and a final row of length $(d_o-i-k)n$. Thus $B_{a,K}$ has size $n\times(d_o-i-1)n$ and $D_K$ has size $(d_o-i-k)n\times(d_o-i-1)n$. The Kraus identity forces the last zero in the active column, but failed components can still reach the unspecified columns; the full-unitary argument above therefore covers every compatible choice of these blocks.

For a fixed history length, clean-input Stinespring access alone also gives a direct construction. We make $d$ calls with fresh zero-input environments $\mathsf H_1,\ldots,\mathsf H_d$ of $q_{\rm env}$ qubits each, retaining their labels. The resulting isometry has row block $K_{\boldsymbol a}=K_{a_d}\cdots K_{a_1}$ on chronological history $\boldsymbol a=(a_1,\ldots,a_d)$. A known $v\times k^d$ matrix $C$ is a selected block on every input of this environment precisely when
\begin{equation}\label{eq:history-effect}
 G=C^\dagger C\preceq I_{k^d}.
\end{equation}
Necessity follows from contractivity of a unitary block, and the rectangular Julia matrix of $C$ proves sufficiency on the recorded labels. Tracing its output label gives the successful operation
\begin{equation}\label{eq:history-operation}
 \sum_b\left(\sum_{\boldsymbol a}C_{b,\boldsymbol a}K_{\boldsymbol a}\right)
 \rho\left(\sum_{\boldsymbol a}C_{b,\boldsymbol a}K_{\boldsymbol a}\right)^\dagger
 =\sum_{\boldsymbol a,\boldsymbol c}
     G_{\boldsymbol c,\boldsymbol a}
           K_{\boldsymbol a}\rho K_{\boldsymbol c}^\dagger.
\end{equation}
The coefficient is $G_{\boldsymbol c,\boldsymbol a}$ because the two amplitudes are paired as $\sum_bC_{b,\boldsymbol a}\overline{C_{b,\boldsymbol c}}$. An arbitrary supplied effect $0\preceq G\preceq I$ is realized by $C=G^{1/2}$, giving a measurement on this fixed purification, as in terminal-memory realizations of instruments~\cite[Theorems~9--10]{chiribella2009}.

Under an environment-basis change $V_{\rm hist}$, the same physical effect has coordinates $V_{\rm hist}GV_{\rm hist}^\dagger$; fixing the numerical effect instead changes the experiment. The invariant choice $G=I$ gives ordinary channel composition. This construction takes $vk^d$ explicit entries for dense $C$ and $d q_{\rm env}$ environment qubits, in addition to the signal and known-gate workspace. By comparison, the joint polynomial algorithm computes from a compact recurrence and normalizes over the full column ball, so its normalization problem differs from that of an explicit history filter.

To combine different lengths, the supplied identity sector provides a coherent reference amplitude. The following degree-one example makes its role explicit.

\begin{example}[A coherent relaxation filter]\label{ex:Kraus-filter}
For a two-Kraus implementation, we label the two output rows by $0,1$ and consider
\[
 F_0(K)=\frac{I+K_1}{2},\qquad F_1(K)=\frac{K_2}{2}.
\]
We realize this column with one query by preparing $(\ket0_{\mathsf O}+\ket i_{\mathsf O})/\sqrt2$ and applying $U_K$, which gives rows $I/\sqrt2,K_1/\sqrt2,K_2/\sqrt2$ on labels $0,i,i+1$. We then apply the selected known matrix
\[
 C=\begin{pmatrix}1/\sqrt2&1/\sqrt2&0\\0&0&1/\sqrt2\end{pmatrix},\qquad
 CC^\dagger=\operatorname{diag}(1,1/2),
\]
whose contractivity permits the rectangular Julia completion of Section~\ref{sec:unitary-implementation}. Decoding its two rows into one output qubit gives $(F_0,F_1)^{\mathsf T}$; replacing $C$ by $C/\beta$ gives normalization $\beta\ge1$. To verify the joint norm, we expand the squared outputs at any column contraction $(A_1,A_2)$:
\[
 F_0(A)^\dagger F_0(A)+F_1(A)^\dagger F_1(A)
 =\frac14\left(I+A_1+A_1^\dagger+A_1^\dagger A_1+A_2^\dagger A_2\right)
 \preceq I.
\]
Equality at $(A_1,A_2)=(I,0)$ proves $B_{\mathsf F}=1$, so this direct one-query circuit attains the boundary normalization. It is a factorization of this particular target; the general column-input theorem gives $O(D/\sqrt\tau)$ queries. At a physical isometric Kraus column, its success probability is
\[
 p_\beta(\rho)=\frac{1+\operatorname{Re}\operatorname{tr}(K_1\rho)}{2\beta^2}.
\]
This dependence on $\operatorname{Re}\operatorname{tr}(K_1\rho)$ is the interference with the idle amplitude. For example, we take the qubit channel with Kraus matrices
\[
 K_1=\begin{pmatrix}1&0\\0&1/2\end{pmatrix},\qquad
 K_2=\begin{pmatrix}0&\sqrt3/2\\0&0\end{pmatrix}.
\]
On input $\ket1\!\bra1$, the filter succeeds with probability $3/(4\beta^2)$ and gives the conditional output
\[
 \frac34\ket1\!\bra1+\frac14\ket0\!\bra0.
\]
The original channel gives the opposite weights, so this filter trades a heralded failure outcome for a smaller relaxation probability on the chosen input.
\end{example}

Finally, an application may require a channel with an explicit failure outcome. We retain an outcome qubit with $\ket{\mathrm s}=\ket1$, $\ket{\mathrm f}=\ket0$ and replace failure by a fixed signal state. For $T=\mathsf F(K)/\beta$, the desired extension is
\begin{equation}\label{eq:Kraus-erasure}
 \mathcal N(\rho)=
 \ket{\mathrm s}\!\bra{\mathrm s}\otimes\mathcal M_\beta(\rho)
 +\ket{\mathrm f}\!\bra{\mathrm f}\otimes
       \ket0\!\bra0\,\tr\bigl[(I_n-T^\dagger T)\rho\bigr].
\end{equation}
Let $q_{\rm pred}$ count the ancillary bits inspected by the final success predicate, including the output-label range test. We charge the reversible predicate workspace separately from the fresh signal register used on failure.

\begin{proposition}[Flagged erasure after a polynomial operation]\label{prop:Kraus-erasure}
The successful polynomial operation admits the extension in Eq.~\eqref{eq:Kraus-erasure} with no additional oracle call, one outcome qubit, $n_s$ fresh signal qubits, and $O(q_{\rm pred})$ clean work qubits. Its postprocessing uses $O(q_{\rm pred}+n_s)$ elementary gates. Selected-column error $\eta$ and postprocessing operator error $\eta_{\rm pp}$ give diamond-norm error at most $4\eta+2\eta_{\rm pp}$.
\end{proposition}
\begin{proof}
We initialize the outcome qubit $\mathsf b_{\rm out}$, a fresh $n_s$-qubit signal register $\mathsf S_0$, and predicate workspace $\mathsf W_{\rm out}$ at zero. After computing the final success predicate into $\mathsf b_{\rm out}$ and uncomputing its workspace, we conditionally swap $\mathsf S$ and $\mathsf S_0$ on failure. Tracing the processing outputs, polynomial output label, and $\mathsf S_0$ leaves the failed signal at zero and retains the outcome. These additional swaps act on the signal, so the proposition counts them and the fresh signal storage separately from the polynomial compiler's auxiliary gates.

To verify the channel also on an entangled input, we group the rejected rows of the full ideal isometry into $n\times n$ Kraus matrices $B_\ell$, where $[B_\ell]_{y,x}=\langle\ell,y|\mathcal U|0,x\rangle$ and $\ell$ runs over the rejected ancillary labels. Complete column normalization gives
\[
 \sum_\ell B_\ell^\dagger B_\ell=I_n-T^\dagger T.
\]
The success and failure ancillary ranges are orthogonal, so tracing the processing registers removes their cross terms. For an input $\rho_{SR}$ on $\mathsf S$ and $\mathsf R_{\rm ref}$, with $I_R$ denoting the identity on the reference, the failure outcome therefore equals
\[
 |\mathrm f\rangle\langle\mathrm f|\otimes|0\rangle\langle0|_{\mathsf S}
 \otimes\operatorname{Tr}_{\mathsf S}\!\left[
 ((I_n-T^\dagger T)\otimes I_R)\rho_{SR}\right].
\]
This evaluates the completely positive failure functional while retaining its correlations with the reference. The bound $\|\widetilde T^\dagger\widetilde T-T^\dagger T\|\le2\eta$ limits its diamond-norm error to $2\eta$; the successful map contributes another $2\eta$. Approximating the postprocessing unitary to operator error $\eta_{\rm pp}$ adds at most $2\eta_{\rm pp}$, giving the claimed total.

For the resource bound, the zero tests and comparison with $v$ use a linear-size reversible conjunction and range-comparison circuit with linear clean workspace. We copy its result to the outcome bit and reverse the workspace computation, then use $n_s$ known Fredkin gates for the controlled swaps, each with a constant-size elementary decomposition. Assigning their approximation budgets gives the stated postprocessing costs.
\end{proof}

\section{Examples of polynomial and function transformations}\label{app:worked-scalar}
\subsection{From a polynomial to a matrix-function algorithm}\label{sec:applications}
To use the synthesis theorem for a matrix function $f$, we need a rational finite-state polynomial $p_D$ whose approximation is uniform on the declared query domain:
\[
 \sup_{n,\|\Delta(A)\|\le1}\|f(A)-p_D(A)\|\le\eta_D.
\]
We apply Theorem~\ref{cor:scalar-synthesis} to this polynomial and use the triangle inequality to obtain a block encoding of $f/\beta_D$ with error at most $\eta_D/\beta_D+\epsilon_g$, where $\beta_D$ is the computed normalization. The application must therefore supply both an approximation bound and a compact description of the approximating coefficients. A uniformly contractive numerical recurrence supplies one such description by truncation, and the example below also provides the lower singular-value bound needed to control success probability.

We take a query for the diagonal block
$\Delta(A)=\begin{pmatrix}A_1&0\\0&A_2\end{pmatrix}$, so its promise is $\|A_1\|,\|A_2\|\le1$, with arbitrary order of multiplication. We denote by $\mathcal L$ the words in letters $1,2$ with no consecutive occurrence of $2$. At $\lambda=1/8$, we define the scalar-coefficient series and truncation
\begin{equation}\label{eq:language-series}
 f(A)=\sum_{w\in\mathcal L}\lambda^{|w|}A_w,
 \qquad
 p_D(A)=\sum_{\substack{w\in\mathcal L\\|w|\le D}}
                                     \lambda^{|w|}A_w.
\end{equation}
The first terms are $I_n+\lambda(A_1+A_2)+\lambda^2(A_1^2+A_1A_2+A_2A_1)$. The restriction thus distinguishes words with the same letter counts and different orders.

The coefficients are specified by the two-state numerical matrices
\begin{equation}\label{eq:language-automaton}
 M_1=\begin{pmatrix}1&0\\1&0\end{pmatrix},\qquad
 M_2=\begin{pmatrix}0&1\\0&0\end{pmatrix},\qquad
 \ell=[1\ 0],\qquad r_0=\begin{pmatrix}1\\1\end{pmatrix}.
\end{equation}
The scalar $\ell M_wr_0$ is one exactly when $w\in\mathcal L$. The superdiagonal block construction of Eq.~\eqref{eq:degree-clock}, with blocks $\lambda M_j$ and terminal columns all equal to $r_0$, truncates the degree at $D$ and uses $2(D+1)$ states. The number of accepted words is exponential in $D$, but the algorithm works on these numerical matrices and their residual coordinates.

\begin{proposition}[A function transformation with constant success probability]\label{prop:language}
For $0<\epsilon<1$, the series Eq.~\eqref{eq:language-series} has a forward-query block encoding with computed normalization at most two and selected error at most $\epsilon/6$. It uses $O(\log(1/\epsilon))$ joint queries, with classical and known-gate costs polynomial in $\log(1/\epsilon)$ and the oracle width. For every normalized input, measurement of the final ancillary zero string succeeds with probability at least $1/16$, and the conditional output differs in vector norm from $f(A)\ket\psi/\|f(A)\ket\psi\|$ by at most $\epsilon$.
\end{proposition}
\begin{proof}
There are at most $2^d$ words of length $d$, and each promised matrix product has norm at most one. The series therefore converges uniformly and satisfies
\[
 \|f-p_D\|\le\frac13\,4^{-D},\qquad
 \|p_D-I_n\|,\ \|f-I_n\|\le\frac13.
\]
Thus $p_D$ and $f$ have smallest singular value at least $2/3$, while $1\le B_D\le4/3$, the lower norm bound following at the zero tuple. We take $D\ge\lceil\log_4(4/\epsilon)\rceil$, $\tau=1/2$, and known-gate error $\epsilon_g=\epsilon/12$. The computed normalization obeys $1<\beta_D\le(3/2)B_D\le2$, and the truncation error is at most $\epsilon/12$. The total selected error relative to $f/\beta_D$ is at most $\epsilon/6$. The degree is logarithmic and the representation has $O(D)$ states, so the query and compilation bounds follow from Theorem~\ref{cor:scalar-synthesis}.

For a unit input $\ket\psi$, we write $\widetilde b$ for the actual selected output column. Its distance from $p_D(A)\ket\psi/\beta_D$ is at most $\epsilon/12$, and hence $\|\widetilde b\|\ge1/3-\epsilon/12\ge1/4$. This gives the success probability. The desired unnormalized column $b=f(A)\ket\psi/\beta_D$ has norm at least $1/3$. The elementary inequality
\[
 \left\|\frac{\widetilde b}{\|\widetilde b\|}
                  -\frac b{\|b\|}\right\|
             \le\frac{2\|\widetilde b-b\|}{\|b\|}
\]
then bounds the conditional error by $\epsilon$. The construction uses the forward-query circuit followed by the final ancillary measurement, and the singular-value estimate supplies its success lower bound.
\end{proof}

For this positive-coefficient example, we can also compare a compact LCU construction. Let $a_d$ count the allowed words of length $d$, and set $L_D=\sum_{d=0}^Da_d\lambda^d$. The triangle inequality gives $\|p_D(A)\|\le L_D$, and the admissible tuple $A_1=A_2=I_n$ attains equality, so $B_D=L_D$. The coefficient state can be prepared using two-state continuation counts: first choose length $d$ with probability $a_d\lambda^d/L_D$, and then choose each letter with probability proportional to its number of allowed continuations. The counts require only a length-$D$ integer recurrence. Reversible preparation retains the resulting word, a sequence of $D$ diagonal queries applies its letters in reverse order with idle padding, and inverse preparation yields $p_D/L_D$ by the LCU construction~\cite{childswiebe2012,gilyen2019}. At each query, known word-controlled permutations route the selected letter through its diagonal block, and a fresh failure label records outputs outside that block before the successful oracle label is returned to zero. The same forward query acts on every failure component, as in Theorem~\ref{thm:unitary}. These permutations and markers are reversible bit operations, so their gate cost is polynomial in $D$, $q_o$, and the requested precision bits. Thus the example demonstrates an end-to-end function algorithm, including its success estimate; both approaches exploit its short recurrence. In addition to the oracle register, the explicit word-register implementation stores $O(D)$ label qubits, whereas Eq.~\eqref{eq:main-memory} gives $O(q_o+\log D)$ total non-signal qubits for our construction in this example. This space comparison concerns the two specified implementations.

\subsection{A certificate, transfer matrix, and query circuit}
We now work through a separate target for which the entire certificate and its resulting matrices can be displayed explicitly. Its constant and quadratic terms show how the same computation accounts for coefficients, degree compensation, and the physical query count.

\begin{example}[A complete numerical synthesis example]\label{ex:residual-running}\label{ex:general-certificate}
We use $\Delta(A)=\operatorname{diag}(A_1,A_2)$ and $p_0(x)=1/4+x_1x_2/12$. A three-state representation has $T_1=e_0e_1^{\mathsf T}$, $T_2=e_1e_2^{\mathsf T}$, $\ell=e_0^{\mathsf T}$, and $r_0=(1/4,0,1/12)^{\mathsf T}$, where $e_0,e_1,e_2$ are the coordinate columns of length three. Its residual basis is $(1,x_2)$, so
\[
 g=(1,x_1,x_1x_2,x_2,x_2^2)^{\mathsf T},\quad
 E=(1,0,0,0,0),\quad
 J=\begin{pmatrix}1&0&0&0&0\\0&0&0&1&0\end{pmatrix}.
\]
The successor $x_2^2$ is retained although its target coefficient is zero. We take $i=1$ and $q_o=3$, giving an $8n\times8n$ oracle with identity, active, and hidden groups of sizes $n,2n,5n$. Two queries with $v=(1,1,1)^{\mathsf T}/\sqrt3$ give $w_0=1,w_1=2/3,w_2=1/3$. Consequently, we realize $q=(1+x_1x_2)/4$, whose target row is $P_q=(1,0,1,0,0)/4$, and the finite window returns $p_0$ at normalization one.

The positive certificate is the pair in Eq.~\eqref{eq:example-certificate}, whose four transfer blocks were computed in Section~\ref{sec:extract}.

The blocks in Eq.~\eqref{eq:example-transfer-blocks} give the full numerical matrix, with eight active coordinates followed by the scalar input and output,
\begin{equation}\label{eq:example-V}
 V=\tfrac18E_{1,9}+\tfrac1{\sqrt2}(E_{2,7}+E_{9,2})
                  +\tfrac12E_{7,9}+\tfrac14(E_{8,7}+E_{9,9}),
\end{equation}
where $E_{r,s}$ is the $9\times9$ matrix with one at $(r,s)$ and zero elsewhere, using one-based indices.
Its four blocks are those already checked in Eq.~\eqref{eq:example-transfer-blocks}; we now use them to determine the physical query matrices.

The uniform weights and Eq.~\eqref{eq:example-V}, substituted in the partial-norm ratios and Eq.~\eqref{eq:Mt}, specify three $10\times10$ update matrices with $R=4$. The packing matrices have twelve rows and ten columns, with their entries given by Eqs.~\eqref{eq:packing-minus}--\eqref{eq:packing-plus}; hence $C_0,C_1,C_2$ have sizes $12\times1$, $12\times12$, and $1\times12$. Their selected product is $p_0(A)$ by Theorem~\ref{thm:window-circuit}. The complete circuit has exactly two original queries and the chronological known gates $G_0,D_0,F_1,G_1,D_1,F_2,G_2$ interleaved as in Figure~\ref{fig:whole-circuit}. The original oracle, two memory qubits, dilation qubit, and three failure qubits give $n_0=9$; Theorem~\ref{thm:gates} uses at most twenty ancillary qubits including its phase and clean work.

The chosen certificate implements the target at normalization one, while its uniform norm is $1/3$; applying the general norm search to the same coefficient input yields the near-optimal normalization of Theorem~\ref{cor:scalar-synthesis}. This example displays the certificate-to-circuit calculation at a fixed normalization, which keeps each matrix small enough to inspect directly.
\end{example}

\section{Exact certificates at a prescribed query count}\label{app:fixed-query}
The near-optimal normalization algorithm allows its query count to depend on the available margin. We now fix the count in advance and characterize exact synthesis by the Gram matrices of all query histories, retaining their cross terms so that the coefficient constraints preserve interference. In this characterization we retain every supplied identity label and take $\kappa=i$, $d_-=i+b$, and $d_+=i+a$; the resulting arrays may be exponentially larger than those used by the polynomial-time algorithm, which retains at most two identity labels.

The numerical coefficient matrices of the designated block are
\begin{equation}\label{eq:Fhat}
 \widehat F_0=\begin{pmatrix}I_i&0_{i\times b}\\0_{a\times i}&0_{a\times b}\end{pmatrix},
 \qquad
 \widehat F_j=\begin{pmatrix}0_{i\times i}&0_{i\times b}\\0_{a\times i}&F_j\end{pmatrix}
       \quad(1\le j\le k).
\end{equation}
Their common size is $d_+\times d_-$. We use $\widehat F_j^{(R)}$ for the numerical $d_+R\times d_-R$ matrix with block $(\alpha,\beta)$ equal to $(\widehat F_j)_{\alpha\beta}I_R$, following Eq.~\eqref{eq:F-memory}. With $A_0=I_n$, they describe the scalar coefficient of each possible letter in one designated query. We index a length-$t$ query history by a word over $\{0,1,\ldots,k\}$, written in matrix-product order; the newest letter is prefixed. The number of histories is $L_t=(k+1)^t$. Deleting zero letters from a history retains its ordered active word.

For a Hermitian matrix $G$ of order $d_-L_t$, with indices $(\alpha,w)$ where $0\le\alpha<d_-$ and $|w|=t$, we define numerical linear maps
\begin{equation}\label{eq:history-maps}
 [\mathcal P_t(G)]_{w,v}=\sum_{\alpha=0}^{d_--1}G_{(\alpha,w),(\alpha,v)},
 \qquad
 [\mathcal T_t(G)]_{jw,\ell v}
 =\sum_{\alpha,\beta}(\widehat F_j^\dagger\widehat F_\ell)_{\alpha\beta}
                         G_{(\alpha,w),(\beta,v)}.
\end{equation}
Their output orders are $L_t$ and $L_{t+1}$, respectively. In particular, $\mathcal P_0(G)$ is a scalar. This Gram description retains cross terms between distinct histories; replacing it by separate bounds on each history would lose the interference in the target coefficients. This follows the semidefinite description of query algorithms~\cite{barnum2003,barnum2007} and its connections to completely bounded forms~\cite{arunachalam2019,gribling2019}; the condition below records the specific coefficient and designated-block constraints of Eq.~\eqref{eq:Fhat}.

\begin{definition}[A certificate for a fixed query count]\label{def:fixed-query-certificate}
For $Q\ge1$, $\beta>0$, and a scalar target of degree at most $Q$, a query certificate consists of positive semidefinite matrices $G_t$ of order $d_-L_t$ for $0\le t<Q$ and a numerical row $K=[K_w]_{|w|=Q}$ of length $L_Q$ satisfying
\begin{align}
 \mathcal P_0(G_0)&\le1,\label{eq:history-initial}\\
 \mathcal P_t(G_t)&\preceq\mathcal T_{t-1}(G_{t-1})
                           \quad(1\le t<Q),\label{eq:history-internal}\\
 \begin{pmatrix}\mathcal T_{Q-1}(G_{Q-1})&K^\dagger\\K&1\end{pmatrix}
                              &\succeq0,\label{eq:history-final}\\
 \sum_{\substack{|w|=Q\\\mathrm{erase}_0(w)=v}}K_w&=c_v/\beta
                    \quad\text{for every active word }|v|\le Q.
                    \label{eq:history-coeff}
\end{align}
All unspecified target coefficients, including those above the true degree, are zero in the last equation.
\end{definition}

\begin{theorem}[Exact synthesis at a prescribed query count]\label{thm:fixed-query}
Let $Q\ge1$, $\beta>0$, and $p$ have degree at most $Q$. In the model of Section~\ref{sec:unitary}, a completion-independent $Q$-query implementation of $p(A)/\beta$ exists if and only if a certificate in Definition~\ref{def:fixed-query-certificate} exists. Gram factorization and Moore--Penrose inversion of a supplied certificate determine all numerical query matrices and their full unitary implementation.
\end{theorem}
\begin{proof}[Proof of Theorem~\ref{thm:fixed-query}]
We first reduce a general admissible circuit to its designated query product. We take $P_-$ and $P_+$ to be the diagonal zero--one matrices of order $d_o$ selecting the first $d_-$ input labels and the first $d_+$ output labels, respectively. For $|z|=1$, the matrix
\[
 U_A(z)=\bigl((P_++z(I-P_+))\otimes I_n\bigr)\,U_A\,
                            \bigl((P_-+z(I-P_-))\otimes I_n\bigr)
\]
is another allowed completion of the same designated block. The four rectangular corners of $U_A(z)$ have powers $z^0,z^1,z^1,z^2$ according to whether their input and output labels are designated. Replacing every query by this same $U_A(z)$ makes the selected circuit block a polynomial in $z$ of degree at most $2Q$, with no negative powers. Its average over the $(2Q+1)$st roots of unity is therefore its constant coefficient. Only the designated corner can occur at each query in that coefficient, so the average has the form Eq.~\eqref{eq:contraction-product}, with numerical contractions obtained by taking the corresponding submatrices of the known unitary gates. Completion independence identifies this average with the required target. The same argument works for a uniformly accurate circuit, even if its approximate output retains completion dependence, because averaging does not increase its worst-case operator error.

Expanding a designated pre-query contraction product in ordered words gives, for each history $w$, a numerical coefficient column $A_{t,w}$ of length $d_-R_t$. We partition each column into $d_-$ blocks of length $R_t$ and denote those blocks by $X_{t,\alpha w}$. All these are numerical columns of scalar coefficients. Their full Gram is
\[
 (G_t)_{(\alpha,w),(\beta,v)}
                      =X_{t,\alpha w}^\dagger X_{t,\beta v}.
\]
The Gram of the whole columns $A_{t,w}$ is $\mathcal P_t(G_t)$. After the next query, its column at history $jw$ is $\widehat F_j^{(R_t)}A_{t,w}$, whose Gram is $\mathcal T_t(G_t)$. Contractivity of each following numerical matrix gives Eq.~\eqref{eq:history-internal}; the initial and final matrices give Eq.~\eqref{eq:history-initial} and Eq.~\eqref{eq:history-final}. Summing histories with the same active word gives Eq.~\eqref{eq:history-coeff}. Equality on all admissible tuples is equality of free coefficients, as justified below, so no coefficient equation is omitted.

Conversely, a factorization $G_t=X_t^\dagger X_t$ gives a numerical matrix $X_t$ with columns indexed by $(\alpha,w)$. We may use the full principal square root, whose row count remains positive even when $G_t=0$, or choose a factor with fewer rows when its rank is positive. We form $A_t$ by stacking its column blocks $X_{t,\alpha w}$ into one column $A_{t,w}$ for each history $w$. Then $A_t^\dagger A_t=\mathcal P_t(G_t)$. For $1\le t\le Q$, the numerical matrix $B_t$ has columns
\[
 B_{t,jw}=\widehat F_j^{(R_{t-1})}A_{t-1,w},
 \qquad B_t^\dagger B_t=\mathcal T_{t-1}(G_{t-1}).
\]
For $t<Q$, the Gram inequality implies $\|A_tx\|\le\|B_tx\|$ for every numerical coordinate column $x$, and hence $\ker B_t\subseteq\ker A_t$. Consequently $C_t=A_tB_t^+$ is contractive and satisfies $C_tB_t=A_t$; it acts as zero on the orthogonal complement of $\operatorname{ran}B_t$,
equivalently $C_tx=0$ whenever $B_t^\dagger x=0$. The initial column is $C_0=A_0$. The Schur complement of the final block gives $K^\dagger K\preceq B_Q^\dagger B_Q$, so the same factorization gives $C_Q=KB_Q^+$. These are the concrete matrix formulas underlying Douglas factorization~\cite{douglas1966}.

Their coefficient product is exactly $K_w$ at every history. Padding the finite memory multiplicities to a common $R$ by zero matrix entries gives the sizes used in Section~\ref{sec:unitary-implementation}. Theorem~\ref{thm:unitary} implements this product using the original full oracle, and Eq.~\eqref{eq:history-coeff} makes its selected block $p(A)/\beta$. This proves sufficiency, with success selected only at the final stage.

To prove that equality on every allowed tuple implies equality of coefficients, we consider a word $v=j_1\cdots j_d$. We take $(d+1)\times (d+1)$ matrices whose only possible nonzero entries lie on the edges of the path $0,1,\ldots,d$, assigning letters $j_d,\ldots,j_1$ to its consecutive edges. The matrix entry from column zero to row $d$ of a polynomial at this tuple isolates exactly its coefficient at $v$: shorter products do not reach $d$, longer products vanish, and a different length-$d$ word misses an edge. A sufficiently small common rational scaling makes Eq.~\eqref{eq:pencil} contractive. The zero tuple recovers the constant coefficient. Thus equality at all allowed finite tuples is coefficientwise equality, completing necessity.
\end{proof}

The fixed-query criterion has $(k+1)^Q$ terminal history indices, so its explicit description can be exponentially larger than the finite-state target. For algebraic input coefficients and normalization, its feasibility is decidable by real-algebraic sampling, since all inequalities can be written using principal minors and all coefficient equalities are finite~\cite{basupollackroy2006}. An algebraic feasible point then determines exact matrix expressions for the known gates. For a supplied rational feasible point, the extraction also has polynomial bit complexity in that point's explicit size: the post-query Grams are rational, and $B_t^+=(B_t^\dagger B_t)^+B_t^\dagger$ reduces the inverse to the rational Gram's exactly computable range. These extraction bounds include singular supplied certificates and are measured in their explicit size.

Every strict normalization also admits some finite query realization. The zero target has a selected zero block with no queries at every $\beta>0$. For a nonzero target and a prescribed $\beta>B$, we choose a rational $0<\tau<1$ with $(1+\tau)B<\beta$. The general-input construction, with the separate known-gate construction for constants, supplies an ideal circuit at a normalization $\beta_0\le(1+\tau)B<\beta$; multiplying its final numerical row by the known contraction $\beta_0/\beta$ then gives exactly $p/\beta$ with the same queries. At a prescribed query count, exact feasibility is determined by the Gram criterion above, including at the boundary $\beta=B$; the norm bound alone does not determine that query count.

\clearpage
% \titleformat{\section}{\normalfont\Large\bfseries}{\thesection}{1em}{}
% \phantomsection
% \addcontentsline{toc}{section}{References}
% \bibliographystyle{alpha}
% \bibliography{ref}
\end{document}